\documentclass{article}
\usepackage{booktabs}

\usepackage[round]{natbib}

\usepackage{hyperref}
\usepackage{amssymb, amsthm,amsmath}
\usepackage{float}
\usepackage{amsmath}
\usepackage{graphicx,psfrag,epsf}
\usepackage{enumerate}
\usepackage{natbib}
\usepackage{url} 
\usepackage{tikz}
\usepackage{pgfplots}
\usepackage{subcaption}
\usepackage[most]{tcolorbox}
\usepackage{a4wide}

\pgfplotsset{compat=1.18}

\usepackage[utf8]{inputenc}
\usetikzlibrary{patterns}
\usetikzlibrary{shapes}
\usepackage[linesnumbered,ruled]{algorithm2e}

\usepackage{bm}
\usepackage{dsfont}
\usepackage{amssymb}
\usepackage{amsfonts}
\usepackage{amsthm}
\usepackage{mathtools}
\usepackage{xcolor}
\usepackage{multirow}

\newtheorem{theorem}{Theorem}[section]
\newtheorem{lemma}[theorem]{Lemma}
\newtheorem{proposition}[theorem]{Proposition}

\newtheorem{remark}{Remark}[section]

\newcommand{\Cp}{\ensuremath{F}}

\newcommand{\mP}{\ensuremath{\mathbb P}}

\newcommand{\et}[1]{\textcolor{black}{#1}}
\newcommand{\ind}[1]{{\mathbf{1}\{#1\}}}
\newcommand{\range}[1]{[#1]}
\newcommand{\R}{\mathbb{R}}
\newcommand{\E}{\ensuremath{\mathbb E}}
\renewcommand{\P}{\ensuremath{\mathbb P}}

\renewcommand{\l}{\ell}
\newcommand{\wh}{\widehat}

\SetKwComment{Comment}{/* }{ */}

\begin{document}

\title{Powerful batch conformal prediction for classification}

\author{Ulysse Gazin\footnote{Universit\'e Paris Cit\'e and Sorbonne Universit\'e, CNRS, Laboratoire de Probabilit\'es, Statistique et Mod\'elisation. Email: ugazin@lpsm.paris} $\quad$ Ruth Heller\footnote{Department of Statistics and Operations Research, Tel-Aviv University. Email: ruheller@gmail.com} $\quad$ Etienne Roquain\footnote{Sorbonne Universit\'e and Universit\'e Paris Cit\'e, CNRS, Laboratoire de Probabilit\'es, Statistique et Mod\'elisation. Email: etienne.roquain@upmc.fr} $\quad$ Aldo Solari\footnote{Department of Economics, Ca Foscari University of Venice. Email: aldo.solari@unive.it}}

  \maketitle

\bigskip

\begin{abstract}

In a split conformal framework with $K$ classes, a calibration sample of $n$ labeled examples is observed for inference on the label of a new unlabeled example. We explore the setting where a `batch' of $m$ independent such unlabeled examples is given, and the goal is to construct a batch prediction set with 1-$\alpha$ coverage. Unlike individual prediction sets, the batch prediction set is a collection of label vectors of size $m$, while the calibration sample consists of univariate labels. 
A natural approach is to apply the Bonferroni correction, which concatenates individual prediction sets at level $1-\alpha/m$.  We propose a uniformly more powerful solution, based on specific combinations of conformal $p$-values that exploit the Simes inequality. We provide a general recipe for valid inference with any combinations of conformal $p$-values, and compare the performance of several useful choices.  Intuitively, the pooled evidence of relatively
 `easy' examples within the batch can help provide narrower batch prediction sets. Additionally, we introduce a more computationally intensive method that aggregates batch scores and can be even more powerful.    The theoretical guarantees are established when all examples are independent and identically distributed (iid), as well as more generally when iid is assumed only conditionally within each class. Notably, our results remain valid under label distribution shift, since the distribution of the labels need not be the same in the calibration sample and in the new batch. The effectiveness of the methods is highlighted through illustrative synthetic and real data examples.
\end{abstract}

\bigskip

\noindent%

\section{Introduction}

Conformal prediction is a popular tool for providing prediction sets with valid coverage \citep{vovk2005algorithmic}. The strength of the approach is that the guarantee holds  for any underlying data-distribution,  and can be combined with any machine learning algorithm.
In this paper, we follow the split/inductive conformal prediction in a classification setting for which a machine has been pre-trained on an independent training sample \citep{papadopoulos2002inductive, vovk2005algorithmic,lei2014conformal} and an independent calibration sample with {\it individual} labeled examples is available. We would like to use the calibration sample efficiently,  to derive the prediction set for the label vector of a {\it batch} of new examples, without making any distributional assumption.

Formally, let $X_i \in \mathcal{X}$ (the space $\mathcal{X}$ is without restrictions) be the covariate and $Y_i\in [K]$\footnote{Throughout the paper, we denote by $\range{\ell}$ the set $\{1,\dots,\ell\}$, for any integer $\ell\geq 1$.} be the class label for example $i$. We observe a calibration sample $\{(X_i,Y_i),i\in [n]\}$, and only the covariates from the batch $\{(X_{n+i}, Y_{n+i}),i\in [m]\}$. We assume that a machine has been pre-trained (with an independent training sample) and is able to produce non-conformity scores $S_{k}(x)$ for any label $k\in [K] $ and any {\it individual} covariate $x\in  \mathcal{X}$. 
The considered task is to produce a collection $\mathcal{C}^m_\alpha$ (called a batch prediction set) of batch label vectors $y:=(y_{i})_{i\in\range{m}}\in \range{K}^m$ such that one of the two following guarantees holds:
\begin{align}
&\P( (Y_{n+i})_{i\in [m]}\in \mathcal{C}^m_\alpha)\geq 1-\alpha\:; \label{aim}\\
&\P( (Y_{n+i})_{i\in [m]}\in \mathcal{C}^m_\alpha\:|\: (Y_{n+i})_{i\in [m] } = y)\geq 1-\alpha\: ,\label{aimcond}
\end{align}
where the guarantee in \eqref{aimcond} is meant to hold for any possible batch $y\in \range{K}^m$.
The unconditional guarantee in \eqref{aim} is considered for the {\it iid model}, for which the probability is taken with respect to (wrt) the sample $\{(X_i,Y_i),i\in [n+m]\}$ which is assumed to have iid components. By contrast,  the stronger conditional guarantee in \eqref{aimcond} is considered for the {\it conditional model} where the batch label vector $(Y_{n+i})_{i\in [m]}$ is fixed and the probability is taken wrt the distribution of the calibration sample $\{(X_i,Y_i),i\in [n]\}$ and the conditional distribution of $(X_{n+i})_{i\in [m]}$ given $(Y_{n+i})_{i\in [m]}$. Note that by independence, the conditional distribution of $(X_{n+i})_{i\in [m]}$ given $(Y_{n+i})_{i\in [m]}$ is simply equal to the product of the marginal distributions of $X_{n+i}$ given $Y_{n+i}$ for $i\in [m]$.

{While unconditional guarantees of the type \eqref{aim} are the most used targets for inference in the conformal literature \citep{angelopoulos2021gentle,angelopoulos2024theoretical}, we emphasize that \eqref{aimcond} is a much stronger guarantee \citep{vovk2005algorithmic,sadinle2019least, romano2020malice}, which is a particular case of  Mondrian conformal prediction. In our framework, since the true labels are fixed, the batch prediction set can be seen as a {\it batch confidence set}, that is, it is valid for all possible values of the true labels, and covers the case of a label distribution shift between the calibration sample and the batch: while methods built for the iid case implicitly use exchangeability of the labels and thus fail to cover the true batch in that case (see \S~\ref{sec:ClassvFull} for an illustration), methods with conditional coverage \eqref{aimcond} cover the true batch even if the classes are arbitrarily unbalanced. This is of practical importance given that this situation is commonly met in real data sets.}

The typical inference on a `batch' only reports  a  prediction set for each example \citep{lee2024batchpredictiveinference}.  By providing powerful methods that guarantee \eqref{aim},\eqref{aimcond}, the inference is far more flexible. 
First, we can extract  a prediction set for each example with a $1-\alpha$ coverage guarantee: for instance, \eqref{aimcond} entails for all $y\in \range{K}^m$,
$$
\P(\forall i\in [m],\: Y_{n+i}\in \mathcal{C}^m_{i,\alpha}\:|\: (Y_{n+i})_{i\in [m] } = y) \geq 1-\alpha,
$$ where $\mathcal{C}^m_{i,\alpha}$ is the set of the $i$-th coordinates of all the vectors in $\mathcal{C}^m_\alpha$, that is, 
$
\mathcal{C}^m_{i,\alpha}=\{y_{i}\in \range{K} \::\:\exists (y_{j})_{j\in [m]\backslash\{i\}}\in \range{K}^{m-1} \::\: (y_{j})_{j\in [m]}\in \mathcal{C}^m_\alpha \}.
$   
In addition to this, we can also extract from  the resulting batch prediction set bounds on the number of examples from each class.
For any possible batch vector $y\in \range{K}^m$, let
\begin{align}\label{countk}
m_{k}(y):=\sum_{i=1}^m \ind{y_{i}=k},\:\:k\in [K],
\end{align}
be the number of examples from class $k$ in the batch $y$. The guarantees \eqref{aim},\eqref{aimcond} ensure that with (conditional) probability at least $1-\alpha$, all unknown numbers $m_{k}((Y_{n+i})_{i\in [m]})$ are included in a range 
\begin{align}
[\ell_{\alpha}^{(k)},u_{\alpha}^{(k)}] &:=[\min \mathcal{N}_k(\mathcal{C}^m_\alpha), \max \mathcal{N}_k(\mathcal{C}^m_\alpha)],\label{lbub}
\end{align}
where $
\mathcal{N}_k(\mathcal{C}^m_\alpha):=\{ m_{k}(y)\:: \: y\in \mathcal{C}^m_\alpha\}$, 
for all $k\in [K]$.

We mention two applications of our work,  where the covariate  corresponds to an image and we should produce a prediction set for the label vector of a {\it batch} of such images:
\begin{itemize}
\item[(i)] Reading zip code \citep{vovk2013transductive}: given a machine trained to classify hand-written digits, we observe a written zip code, that is a batch of $m=5$ images, and we should produce a list of plausible zip codes (a subset of $\range{K}^m$) for this batch; building $\mathcal{C}^m_\alpha$ ensuring \eqref{aim} or \eqref{aimcond} provides a solution, see also  Figure~\ref{fig:USPS} below.
\item[(ii)] Survey animal populations: given a machine trained to classify animal images, we observe a set of $m$ animal images and we should provide a prediction sets for the counts of each animal; 
building $[\ell_{\alpha}^{(k)},u_{\alpha}^{(k)}]$ as in  \eqref{lbub} provides a solution, see illustrations in \S~\ref{sm-subsec-CIFAR}.
\end{itemize}

In a very recent paper, \cite{lee2024batchpredictiveinference} suggest constructing 
prediction sets for functions of the batch points (e.g., for the mean or median outcome of the batch), assuming exchangeability of the calibration and test data, for both regression and classification. Their motivation is thus the same as ours, of providing model-free joint inference on multiple test points. They did not develop methodology targeting the inferential  guarantees \eqref{aim},\eqref{aimcond}. For their aims, they use a similar approach to the permutation approach we suggest in \S~\ref{sec-batchscores}.

The guarantee \eqref{aim} has been considered in \cite{vovk2013transductive}. To achieve the $1-\alpha$ guarantee, the problem of a batch prediction set is seen as the problem of testing at level $\alpha$ each of the $y \in \range{K}^m$ possible sets of labels. \cite{vovk2013transductive} suggested in the full/transductive conformal setting using Bonferroni for each partitioning hypothesis. The advantage is that only $m\cdot K$ conformal $p$-values, i.e., $K$ for each example, need to be computed. So there is no need to go over all $K^m$ possible vectors of labels since $m\cdot K$ computations are enough. However, the computational simplicity comes at a severe cost: the batch prediction set using Bonferroni may be unnecessarily large, and thus less informative, than using more computationally intensive methods. 

Our main contributions are as follows. 
We start by casting the problem of finding the batch prediction set as the problem of finding all the vectors that are not rejected when testing each of the $y\in \range{K}^m$ possible sets of labels in the conditional and iid models. 
By using the well-known Simes test, we show that there is a uniformly better (i.e., narrower) batch prediction set than Bonferroni's, that we refer to as the {\it Simes batch prediction set}. We further introduce an adaptive variant (its theoretical guarantee are established for conformal $p$-values with a possible label shift, enriching the available literature \citep{Storey2003,bates2023testing}). We show how to construct batch prediction sets with any $p$-value combining function in Algorithm   \ref{alg:general}. We also provide a computational shortcut algorithm to compute the bounds \eqref{lbub} that maintains the $1-\alpha$ coverage guarantee. 
Finally, we provide an alternative method that combines batch scores rather than conformal $p$-values in Algorithm \ref{alg:batchscores}. We suggest the estimated likelihood ratio statistic, and show it has excellent power but a large computational cost compared to the methods that combine conformal $p$-values.   We demonstrate the usefulness of our recommendations for image classification and  USPS  digits problems.\footnote{The code used in all our experiments is made publicly available at \url{https://github.com/ulyssegazin/BatchCP_Classification}}

The novel methods are available in two versions, each being valid for the iid or conditional model. The theoretical proofs are deferred to the supplementary file. The latter also contains additional illustrations, numerical experiments and mathematical materials.

To illustrate our method, 
Table~\ref{tabUSPS} provides an example of batch prediction set for the particular zip code displayed in Figure~\ref{fig:USPS}. For each combining function, Bonferroni or Simes, the proposed batch prediction set can be expressed as the batches with a $p$-value larger than $\alpha$ (see \eqref{PredBonf}, \eqref{PredSimes} below). At $5\%$, we see that the Bonferroni batch prediction set is of size $8$, whereas the Simes batch prediction set is of size $6$ and is able to exclude the batches $(0,6,5,5,4)$ and $(0,6,6,5,4)$ from the prediction set. This is because all digits of the batch are acceptable according to Bonferroni's method, but are not acceptable {\it together} according to Simes' method.
To show that this phenomenon is not due to the particular data generation, a violin plot for $500$ replications is provided in Figure~\ref{violinUSPS}. Below the violin plot, the scatter plot of the number of rejections by each method clearly shows that the batch prediction set using Simes  can be much narrower than using Bonferroni (and is never larger than using Bonferroni).

{\small
 \begin{table}[ht]
 \centering
 \subfloat[]{\includegraphics[width=0.4\textwidth]{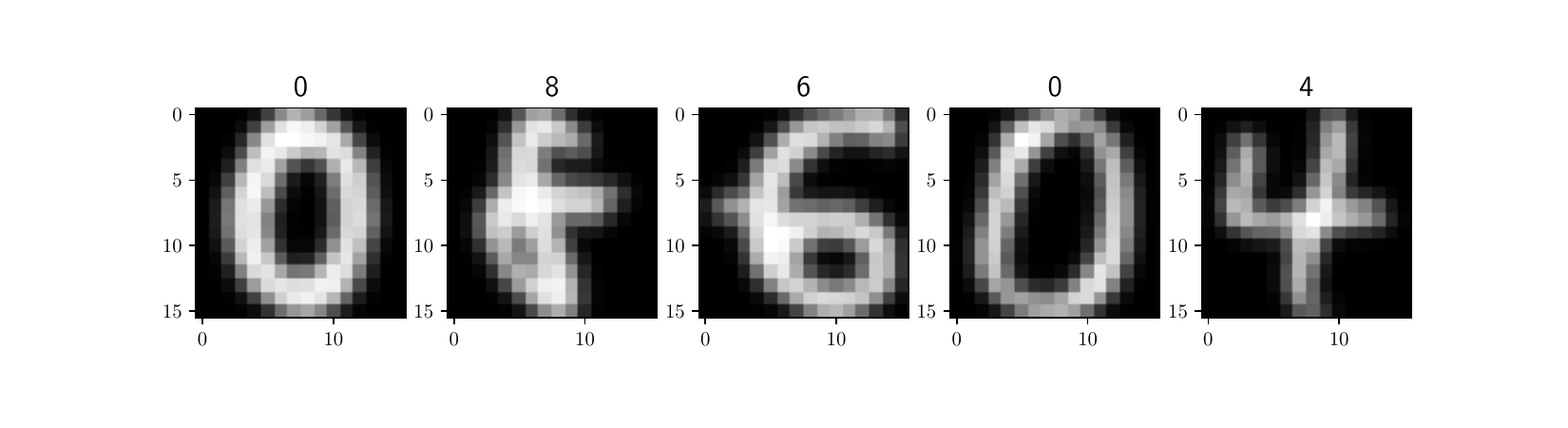} \label{fig:USPS}}\\
    \subfloat[]{ 
\begin{tabular}{|lllll|ll|}
\toprule
0 & 8 & 6 & 0 & 4 & Bonferroni & Simes \\
\midrule
0 & 6 & 5 & 5 & 4 & {\bf 0.065} & 0.038 \\
0 & 6 & 6 & 5 & 4 & {\bf 0.065} & 0.038 \\
0 & 6 & 5 & 0 & 4 & {\bf 0.065} & {\bf 0.065} \\
0 & 6 & 6 & 0 & 4 & {\bf 0.065} & {\bf 0.065} \\
0 & 8 & 5 & 5 & 4 & {\bf 0.077} & {\bf 0.077} \\
0 & 8 & 6 & 5 & 4 & {\bf 0.077} & {\bf 0.077} \\
0 & 8 & 5 & 0 & 4 & {\bf 0.277} & {\bf 0.277} \\
0 & 8 & 6 & 0 & 4 & {\bf 0.605} & {\bf 0.345} \\
\bottomrule
\end{tabular}}
 \caption{\label{tabUSPS} 
 Batch prediction sets  at level $0.05$ for Bonferroni's and Simes' methods computed on the particular batch displayed above the table (from the USPS dataset provided by the US Postal Service for the paper \cite{LeCun1989HandwrittenDR}) previously studied by \cite{vovk2013transductive}). Columns 6 and 7 provide the combination $p$-values using combining functions  \eqref{FBonf} and \eqref{FSimes}, respectively. The batch prediction set corresponds to batch $p$-values displayed in bold. }
 \end{table}
}

\begin{figure}[h!]
\vspace{-1cm}
\begin{center}

\begin{tabular}{r}

\hspace{-4.5mm}\includegraphics[scale=0.27]{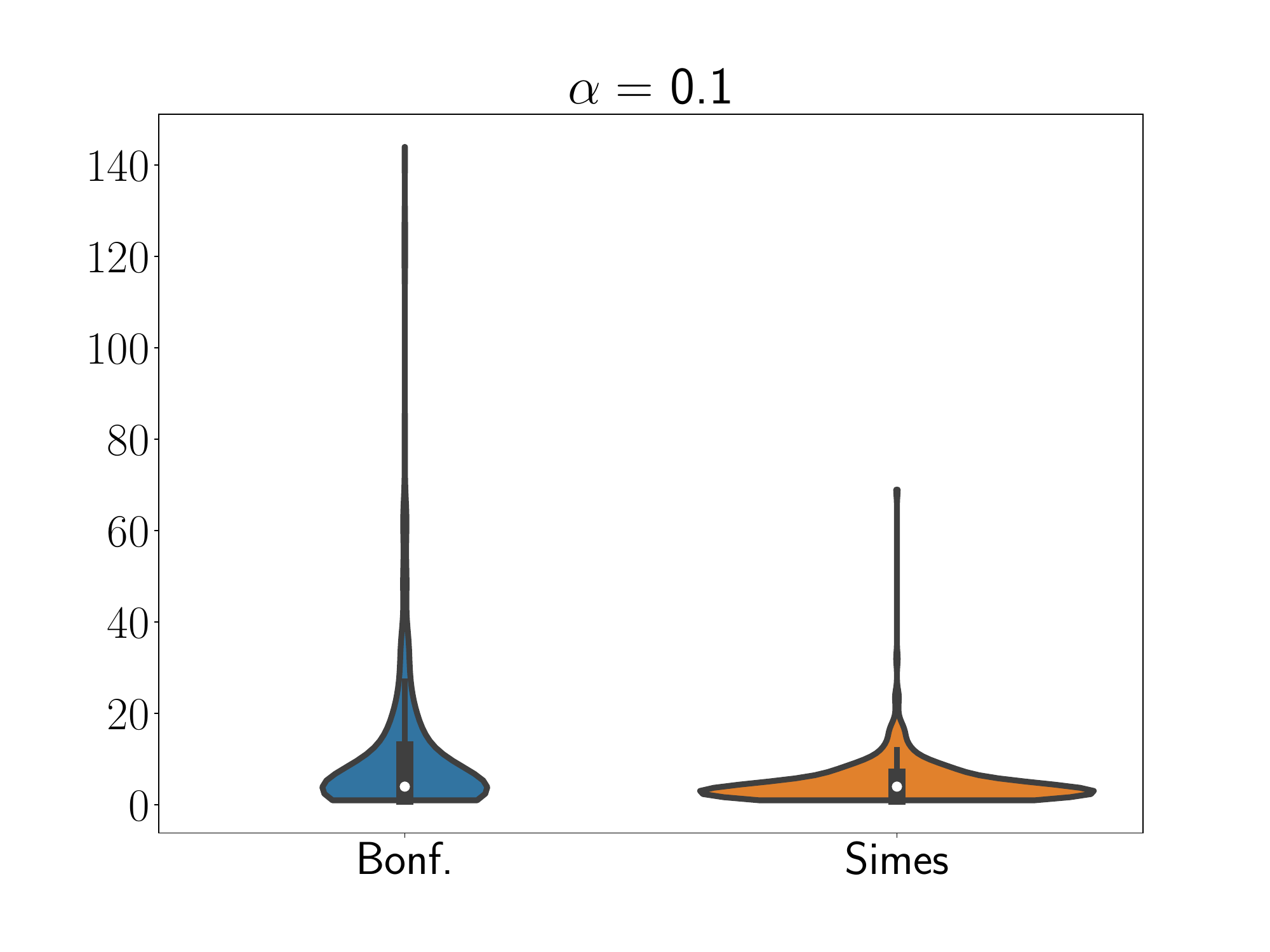}\vspace{-1.2cm}\\
\hspace{-4.5mm}\includegraphics[scale=0.27]{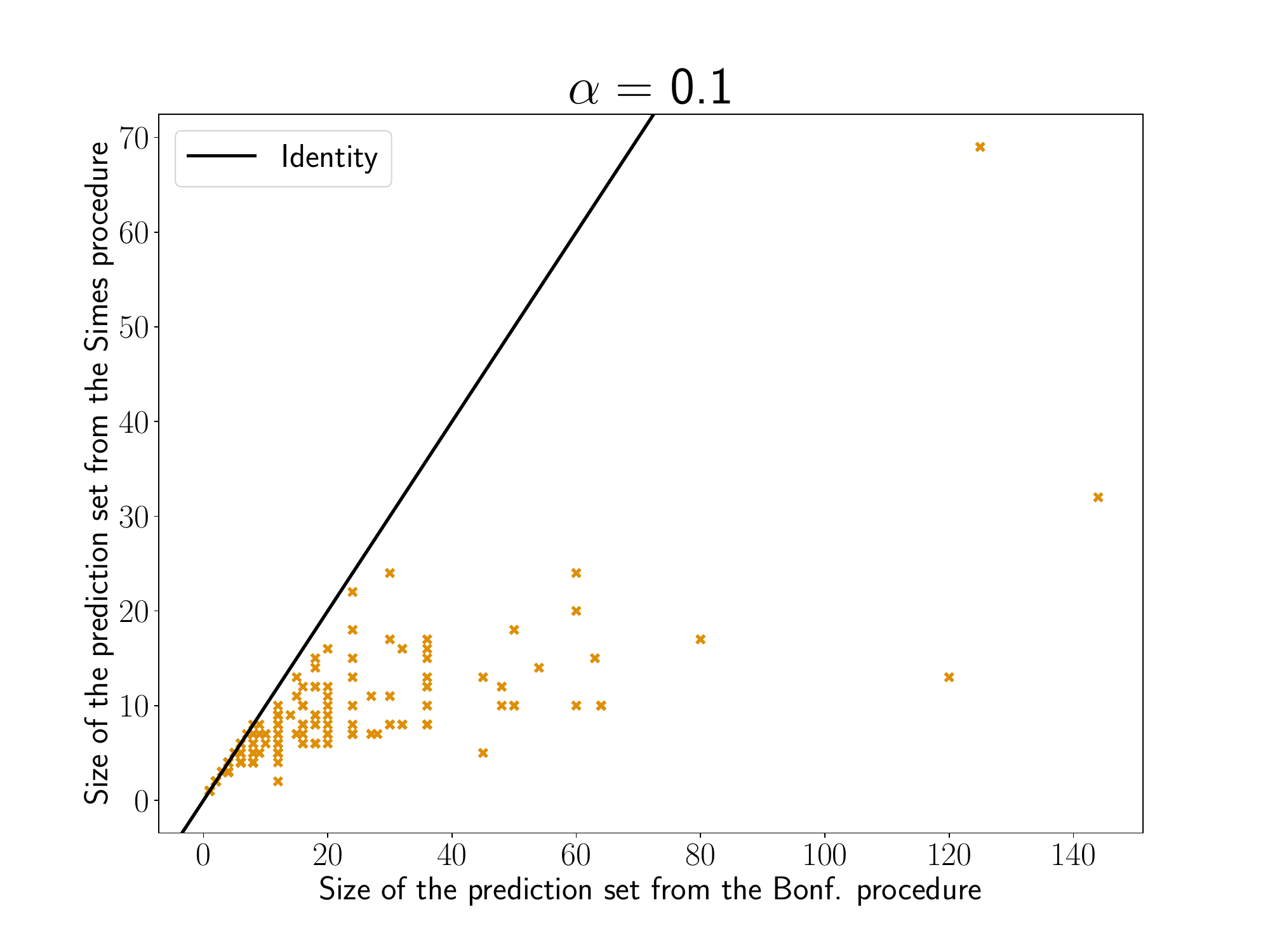} \vspace{-5mm}\\
\end{tabular}
\caption{Violin plots (top row) and scatter plot (bottom row)  for the size of the batch prediction sets of Bonferroni's and Simes' methods ($m=3$, $K=10$, $500$ replications) for $\alpha= 0.1$. Generated from the USPS dataset, as described in \S~\ref{sec:data}.\label{violinUSPS}}
\end{center}
\end{figure}

Finally, let us describe some related works.  Our methodology is tightly related to the multiple testing literature, in particular \cite{BY2001,BKY2006,barber2017p,Bogomolov23, HellerSolari23}, where Simes and adaptive Simes variants are shown to be useful for inference on a family of null hypotheses.
Existing work for the task of building prediction sets concentrated thus far primarily on providing  a false coverage rate (FCR) guarantee  \citep{bates2023testing,gazin2023transductive,gazin2024selecting,jin2024confidence}. To derive our theoretical results, we rely on the literature on conformal novelty detection  \citep{bates2023testing,marandon2024adaptive} under the `full null' configuration, that is, when the test sample is not contaminated by novelties. While we show that these works yield {\it de facto} the unconditional guarantee \eqref{aim}, we extend the theory to also cover the more challenging conditional guarantee \eqref{aimcond}. We emphasize that our work consider the setting where we observe a calibration sample of {\it examples} (not batches), as in \cite{lee2024batchpredictiveinference}. If a calibration sample of {\it batches} is at hand, the usual conformal inference pipeline can (and should) be used by defining batch scores that take into account the interaction between batch elements  \citep{messoudi2020conformal,messoudi2021copula,johnstone2021conformal,johnstone2022exact}. In our work, the batch examples are assumed independent and the calibration sample only contains scores for individual examples, so our setting is markedly different.

\section{Methods using combinations of conformal $p$-values}

Henceforth, we make the classical assumption that the scores $S_{Y_{i}}(X_i)$, $i\in [n+m]$, have no ties almost surely. 

\subsection{Conformal $p$-values}

For $k\in [K]$, we consider the conformal $p$-value \citep{vovk2005algorithmic} for testing the null ``$Y_{n+i} = k $'' versus ``$Y_{n+i}\neq  k $'' in the test sample. 
Formally, the  $p$-value family $(p^{(k)}_i, k\in [K], i\in [m])$ is given as follows: 
\begin{equation}\label{standardpvalue}
p^{(k)}_i= \frac{1}{|\mathcal{D}^{(k)}_{{\tiny \mbox{cal}}}|+1}\Big(1+\sum_{j\in \mathcal{D}^{(k)}_{{\tiny \mbox{cal}}}} \ind{S_{Y_j}(X_j)\geq S_{k}(X_{n+i})} \Big),
\end{equation} 
with $\mathcal{D}^{(k)}_{{\tiny \mbox{cal}}}$ being either $\range{n}$, of size $n$, in the iid setting
 or $\{j \in \range{n}\::\: Y_j= k\}$, of size $n_k$, in the conditional setting. The $p$-values in \eqref{standardpvalue} are  referred to as  {\it full-calibrated} $p$-values in the iid setting and  {\it class-calibrated} $p$-values in the conditional setting.

Since scores $\lbrace S_{Y_j}(X_j), j\in \mathcal{D}^{(Y_{n+i})}_{{\tiny \mbox{cal}}}\rbrace\cup \{S_{Y_{n+i}}(X_{n+i})\}$ are exchangeable both in the iid and class-conditional setting, the following, well known property, holds.

\begin{proposition}\label{prop:marginal}
    The conformal $p$-values are marginally super-uniform, that is, for all $i\in \range{m}$, for all $u\in [0,1]$, $\P(p^{(Y_{n+i})}_i\leq u)\leq u$ for full-calibrated $p$-values and $\P(p^{(Y_{n+i})}_i\leq u\:|\: (Y_{j})_{j\in \range{n+m}})\leq u$ for class-calibrated $p$-values.
    \end{proposition}
Proposition~\ref{prop:marginal} ensures that each individual label set 
$
\mathcal{C}_{i,\alpha}:=\{ y_{i} \in \range{K}\::\: p^{(y_{i})}_i> \alpha \}
$
is a prediction set for $Y_{n+i}$ of (conditional) coverage at least $1-\alpha$. 

\subsection{Bonferroni batch prediction set}

The Bonferroni batch prediction set is given as follows:
\begin{align}\label{PredBonf}
 \begin{array}{ll}
\mathcal{C}^m_{\alpha,\mbox{\tiny Bonf} }:=&\{ y=(y_{i})_{i\in [m]} \in \range{K}^m\::\:\\ 
&\:\:\: \Cp_{\mbox{\tiny Bonf}}((p^{(y_{i})}_i)_{i\in \range{m}})> \alpha \},
\end{array}
\end{align}
where the $p$-value for the batch $y$ and for the Bonferroni method is given by
\begin{align}\label{FBonf}
\Cp_{\mbox{\tiny Bonf}}((p^{(y_{i})}_i)_{i\in \range{m}}):=m\min_{i\in \range{m}}\{p^{(y_{i})}_i\}.
\end{align}
Hence, this prediction set is rectangular:
 $\mathcal{C}^m_{\alpha,\mbox{\tiny Bonf} } = \times_{i=1}^m \{ k\in  \range{K}\::\: p^{(k)}_i> \alpha/m\}$;
it is simply the product of standard individual conformal prediction sets, taken at level $1-\alpha/m$.
By Proposition~\ref{prop:marginal} and a simple union bound, it is clear that \eqref{aim} and \eqref{aimcond} hold by using the full-calibrated  and class-calibrated $p$-values, respectively.

\subsection{Simes batch prediction set}

Let us denote by $p_{(\ell)}((y_{i})_{i\in[m]})$ the $\ell$-th largest element among the vector $(p^{(y_{i})}_i, i\in[m])$.
The Simes  batch prediction set is given as follows:
\begin{align}\label{PredSimes}
 \begin{array}{ll}
\mathcal{C}^m_{\alpha,\mbox{\tiny Simes} }:=&\{ y=(y_{i})_{i\in [m]} \in \range{K}^m\::\:\\ &\:\:\:\Cp_{\mbox{\tiny Simes}}((p^{(y_{i})}_i)_{i\in \range{m}}) > \alpha \},\end{array}
\end{align}
where the $p$-value for the batch $y$ and for the Simes method is given by
\begin{equation}
    \label{FSimes}
    \Cp_{\mbox{\tiny Simes}}((p^{(y_{i})}_i)_{i\in \range{m}}):=\min_{\ell\in \range{m}} \{m p_{(\ell)}(y)/\ell\}.
\end{equation}
The latter always improves the Bonferroni batch prediction set, that is, $\mathcal{C}^m_{\alpha,\mbox{\tiny Simes} }\subset \mathcal{C}^m_{\alpha,\mbox{\tiny Bonf} }$ pointwise. 
Note that the Simes batch prediction set is not a hyper-rectangle, and cannot be obtained from the individual prediction sets of each element of the batch.
In addition, the next result shows that it provides the correct (conditional) coverage.

\begin{theorem}\label{th:Simes}
The prediction set $\mathcal{C}^m_{\alpha,\mbox{\tiny Simes} }$ satisfies \eqref{aim} and \eqref{aimcond} by using the  full-calibrated and class-calibrated $p$-values, respectively. 
\end{theorem}

To prove Theorem~\ref{th:Simes}, we establish that the Simes inequality \citep{Sim1986} holds for the class/full-calibrated $p$-values in \S~\ref{sec:proofSimes}. This comes from the fact that the conformal $p$-value family is positively dependent in a specific sense. 

The conformal $p$-values are discrete, and therefore the guarantee \eqref{aim} or \eqref{aimcond} is typically a strict inequality. 
To resolve the conservativeness of the coverage that follows from the discreteness of the conformal $p$-values, a standard solution is to use randomized conformal $p$-values \citep{vovk2013transductive}. This solution is (arguably) unattractive since decisions are randomized. Interestingly, exact coverage is possible without need for randomization for specific values of $\alpha$ detailed in the following theorem.     
\begin{theorem}\label{th:Simesexact}
 The coverage for $\mathcal{C}^m_{\alpha,\mbox{\tiny Simes} }$ is exactly $1-\alpha$ in the two following cases:
 \begin{itemize}
       \item in the iid model, for full-calibrated $p$-values, if $\alpha (n+1)/m$ is an integer;
    \item in the conditional model, for class-calibrated $p$-values if $\alpha (n_k+1)/m$ is an integer for all $k\in \range{K}$.
 \end{itemize}
\end{theorem}

The proof is given in \S~\ref{sec:proofSimesexact}.

\subsection{Adaptive Simes batch prediction set}\label{sec:adapt}

For any possible label vector $y=(y_{i})_{i\in \range{m}}\in \range{K}^m$, let
\begin{equation}\label{equm0y}
    m_0(y):=\sum_{i\in \range{m}} \ind{y_{i}=Y_{n+i}},
\end{equation}
the number of coordinates of $y$ that are equal to  the true label vector $Y=(Y_{n+i})_{i\in \range{m}}$.
Since $m_0(Y)=m$, the Simes batch prediction set 
$\mathcal{C}^m_{\alpha,\mbox{\tiny Simes} }$ has exactly the same coverage when replacing $m$ by $m_0(y)$ in the threshold. Meanwhile, using $m_0(y)$ may narrow the batch prediction set,  because $m_0(y)<m$ for any vector $y\neq Y$. 
 Unfortunately, $m_0(y)$ is unknown so that this  improved prediction region is only an `oracle' one that cannot be used. 
Our approach consists first in estimating $m_0(y)$ by 
\begin{align}
&\hat{m}_0(y) := 
(1-\lambda)^{-1} \Big(1+\sum_{i=1}^m \ind{p^{(y_{i})}_i\geq \lambda }\Big),\label{pi0estiStoreysimple}
\end{align}
which is an analogue of the so-called Storey estimator in the multiple testing literature \citep{Storey2002}. 
Here, $\lambda\in (0,1)$ is a parameter that is free but should be such that $(n+1)\lambda$ is an integer in the iid setting, or such that $(n_k+1)\lambda$ is an integer for all $k\in \range{K}$ in the conditional setting. If these conditions are too strict, we can accommodate any value of $\lambda\in (0,1)$ by adjusting the formula \eqref{pi0estiStoreysimple} to account for discreteness: the modification is minor, see \S~\ref{sec:detailAdaptive}. 

The adaptive Simes batch prediction set is 
\begin{align}
 &\begin{array}{ll}
\mathcal{C}^m_{\alpha,\mbox{\tiny A-Simes} }:=&\{ y=(y_{i})_{i\in [m]} \in \range{K}^m\::\: \\ &\:\:\:\Cp_{\mbox{\tiny A-Simes}}((p^{(y_{i})}_i)_{i\in \range{m}})>\alpha \},
\end{array}
\label{PredSimesAdapt}
\end{align}
where the $p$-value for batch $y$ and for the adaptive Simes method is given by
\begin{equation}
    \label{FASimes}
    \Cp_{\mbox{\tiny A-Simes}}((p^{(y_{i})}_i)_{i\in \range{m}}):=\min_{\ell\in \range{m}} \{\hat{m}_0(y)\cdot p_{(\ell)}(y)/\ell\},
\end{equation}
and $\hat{m}_0(y)$ is an estimator of $m_0(y)$ \eqref{equm0y}, typically as in \eqref{pi0estiStoreysimple}.

\begin{theorem}\label{th:SimesAdapt}
 The coverage for $\mathcal{C}^m_{\alpha,\mbox{\tiny A-Simes} }$ with the Storey estimator \eqref{pi0estiStoreysimple} is at least $1-\alpha$ both in the iid model (using full-calibrated $p$-values) and in the conditional model (using class-calibrated $p$-values). 
\end{theorem}

The proof is given in \S~\ref{sec:proofSimesadapt}. 
Note that the adaptive Simes method with estimator \eqref{pi0estiStoreysimple} (referred to as {\it Storey Simes} in what follows) does not provide a uniform improvement over Simes (or Bonferroni),  because $\hat{m}_0(y)>m$ is possible for some batches $y$. However, $\hat{m}_0(y)$ is typically (much) smaller than $m$ for batches $y$ which are far from the true batch. Hence, the adaptive version leads to a substantial improvement in a situation where the batch prediction set is large (`weak' signal), see examples in \S~\ref{sec:xp}. 

The tuning parameter $\lambda$ is by default chosen equal to $1/2$ but other choices are possible, see \S~\ref{sec:detailAdaptive} in the supplement.
Therein, we also provide another type of estimator, corresponding to the so-called  `quantile' estimator \citep{BKY2006,marandon2024adaptive} and for which a choice of parameter is the `median' estimator (and the corresponding method is referred to as {\it median Simes}). While we have no theoretical guarantee for median Simes, the performance of median Simes tends to be better than Storey Simes for strong signal and worse when the signal is weak, see \S~\ref{appendix-sim}.

\subsection{General $p$-value combining prediction set}\label{sec:numapprox}

We present a general method for  guaranteeing \eqref{aim} and \eqref{aimcond} using any combining function, denoted by $\Cp((p^{(y_{i})}_i)_{i\in \range{m}})$,  for the conformal $p$-values that test that the batch labels are $y\in \range{K}^m$. Consider a batch prediction set of the form 
\begin{equation}\label{genpredictionset}
  \mathcal{C}^m_{t,F}  :=\{ (y_{i})_{i\in [m]} \in \range{K}^m\::\:  \Cp((p^{(y_{i})}_i)_{i\in \range{m}}) \geq t\},
\end{equation} 

where $t$ is some threshold, possibly depending on the $p$-value vector.
From Theorems~\ref{th:Simes}~and~\ref{th:SimesAdapt}, a valid choice is $t=\alpha$ and 
$F=\Cp_{\mbox{\tiny A-Simes}}$ as in \eqref{FASimes}
with either $\hat{m}_0(y)=m$ or $\hat{m}_0(y)$ as in \eqref{pi0estiStoreysimple}. Algorithm \ref{alg:general} shows how  to find a valid empirical choice of $t$ for any $F$ (see also the simplified version given in \S~\ref{appendix-generalcombinationsalgorithm-iid}, Algorithm~\ref{alg:general_iid}, for the particular case of the iid model). Importantly for computation, the empirical threshold (line 10 in Algorithm \ref{alg:general}) does not depend on the actual scores from the calibration and test examples. 
However, in the conditional model, the threshold 
depends on $(m_k(y))_{k\in \range{K}}$ so $B$ permutations of $\range{n+m}$ should be generated for every configuration of $(m_k)_{k\in \range{K}}$ such that $\sum_{k=1}^K m_k = m$ (where $m_k\in [0,m]$). Hence, the computational cost is more severe than for the iid model, which only requires $B$ permutations of $\range{n+m}$. However, these computations can be done once for all, before observing the data for the batch.

Finally, the attentive reader may have noticed that the inequality in \eqref{genpredictionset} is not strict, which is in contrast with the previous sections. This is necessary to obtain a general valid coverage as stated in Theorem~\ref{th:gencontrol} below and is consistent with standard randomized test theory for the  batch statistic $\Cp((p^{(y_{i})}_i)_{i\in \range{m}})$  \citep{RW2005}, which is intended to be small when a rejection should be made. The prediction set \eqref{genpredictionset} can be equivalently expressed as $C^m_{t,F} = \lbrace y\in \range{K}^m: \hat{q}(y)>\alpha \rbrace$, where $\hat{q}=\hat{q}(y)$ is some batch $p$-value (see \eqref{formulapvaluedistfree} in \S~\ref{proof:gen}), possibly depending on $(m_k(y))_{k \in [K]}$ in the class-conditional case. This representation is similar to the representation in previous sections. Both representations have the same computational complexity, but the representation in the algorithm has the advantage of making it clear that the thresholds $t^{(m_k(y))_{k\in \range{K}}}$ do not depend on the data, i.e.,  the construction of the prediction set is {\it distribution free}.

\begin{theorem}\label{th:gencontrol}
The coverage of the batch prediction set $\mathcal{C}^m_{t,F}$ provided in Algorithm \ref{alg:general}.
 is at least $1-\alpha$ both in the iid model (using full-calibrated $p$-values) and in the conditional model (using class-calibrated $p$-values). For the iid model, the outer loop (lines 1, 2, 12 in Algorithm \ref{alg:general}) is not needed, see Algorithm~\ref{alg:general_iid} in \S~\ref{appendix-generalcombinationsalgorithm-iid}.
\end{theorem}

\begin{algorithm}[!htb]
\small
\SetKwInOut{Input}{Input}
\Input{Number of examples from class $k$ in the calibration set $n_k$, $k \in [K]$;  
       combining function $F$;  
       level $\alpha \in (0,1)$;  
       number of permutations $B$;  
       conformal $p$-values $(p^{(y_{i})}_i)_{i\in \range{m}}$.
}

\For{each possible allocation $h=:(h_k)_{k \in [K]}$ such that $0\leq h_k \leq m$ and $\sum_{k=1}^{K} h_k = m$
}{
    define $z=z(h)=(z_{i})_{i \in [m]} \in [K]^{m}$ as any element such that $m_{k}(z)=h_k$ for all $k \in [K]$;

    \For{each $b\in [B]$}{
        
        Generate a random permutation $\pi_b$ of $[n+m]$;
        
        Compute null conformal $p$-values:
        \[
        \hat{p}^{(z_{i})}_{i,b} \gets \frac{1+\sum_{j\in \mathcal{D}^{(z_{i})}_{{\tiny \mbox{cal}}}} \ind{\pi_b(j) \geq \pi_b(n+i)} }
        {|\mathcal{D}^{(z_{i})}_{{\tiny \mbox{cal}}}|+1}
        \]
        for $i \in [m]$;

        Compute combined statistic:\\
        $\xi^h_b \gets \Cp((\hat{p}^{(z_{i})}_{i,b},i\in \range{m}))$;
    }

    Compute threshold:\\
    $t^h \gets \xi^h_{ ( \lfloor (B+1)\alpha \rfloor )  }$,\\ where $\xi^h_{(1)} \leq \ldots \leq \xi^h_{(B)}$ are the sorted values  of  $\xi^h_1, \ldots, \xi^h_{B}$ and $\xi^h_{(0)}:=-\infty$;
}

Construct batch prediction set:\\
$\mathcal{C}^m_{t,F}  \gets \{ y  
\in \range{K}^m\::\:  \Cp((p^{(y_{i})}_i)_{i\in \range{m}}) \geq t^{(m_k(y))_{k\in\range{K}}} \}$;

\SetKwInOut{Output}{Output}
\Output{Batch prediction set $\mathcal{C}^{m}_{t,F}$.}

\caption{Constructing a batch prediction set using conformal $p$-values combination}
\label{alg:general}
\end{algorithm}
The proof is provided in \S~\ref{proof:gen}.
The method  is very flexible: combined with adaptive Simes combination $F_{\mbox{\tiny Simes}}$, any estimator $\hat{m}_0(\cdot)$ can be used, see detailed suggestions in \S~\ref{sec:detailAdaptive}. Since there is not one uniformly best estimator, and which estimator to use depends on the unknown properties of the data at hand,  it is possible to take as $\hat{m}_0(\cdot)$ the smallest of several estimators of $\hat{m}_0(\cdot)$.  More generally, any $p$-value combination can be used, for instance the Fisher combination
\begin{equation}
    \label{FFisher}
    \Cp_{\mbox{\tiny Fisher}}((p^{(y_{i})}_i)_{i\in \range{m}})= T\Big(-2\sum_{i \in \range{m}} \log(p^{(y_{i})}_i)\Big),
\end{equation}
where $T$ is the survival function of a $\chi^2(2m)$ distribution. The corresponding method is referred to as {\it Fisher} batch prediction set in what follows.
We refer to \cite{HellerSolari23}, and references within,  for more examples of such combining functions.

\subsection{Shortcut for computing bounds}\label{sec:shortcut}

Computing naively the bounds $[\ell_{\alpha}^{(k)}, u_{\alpha}^{(k)}]$ in \eqref{lbub} incurs exponential complexity and thus is difficult when both $K$ and $m$ increase.  
A pseudoalgorithm  for a computational shortcut, which reduces the time complexity for calculating the bounds from $O(K^m)$ to 
$O(K\times m^2)$, is given in \S~\ref{sec:shortcut_supplementary}. 
This shortcut is exact when $K = 2$ and the scores produced by the machine learning model are probabilities, i.e. they satisfy the relationship $S_k(x_{n+i}) = 1-S_{3-k}(x_{n+i})$ for $k \in \{1,2\}$ and $i \in [m]$. 
However, when $K > 2$ or when arbitrary scores are used, the shortcut may become conservative, resulting in wider bounds but never narrower ones. This ensures that the coverage guarantee of at least $1-\alpha$ probability is
maintained. In Appendix \ref{appendix-sim-Gaussian} we examine the performance of the shortcut in our numerical experiments. Interestingly, the bounds using the shortcut are almost identical to the bounds derived from the batch prediction set for Simes (see \S~\ref{appendix-sim-Gaussian}).  

From the bounds produced by the shortcut, it is straightforward to produce a conservative batch prediction set. The size of the set is the sum of all valid assignments of $( m_1,  \dots, m_K )$ occurrences, where $ \ell^{(k)}_\alpha \leq m_k \leq u^{(k)}_\alpha$ for each $ k \in \{1, \dots, K\} $, and \( m_1 +  \dots + m_K = m \), with each valid assignment counted by the multinomial coefficient \( \binom{m}{m_1, m_2, \dots, m_K} \), see \S~\ref{sec:shortcut_supplementary}  for more details.

Finally, we note that since for any $y\in \range{K}^m$, the rejection by Bonferroni necessarily entails rejection using Simes, then we can first apply the Bonferroni procedure, and then apply the suggested shortcut for  Simes on the $(K-R_{1})\times \dots \times (K-R_{m})$ remaining partitions, where $R_{i}$ are the number of conformal $p$ values at most $\alpha/m$ for the $i$-th example of the batch.

\section{Method using batch scores}\label{sec-batchscores}

Thus far,  we have considered methods that  combine conformal $p$-values. Next, we  present a general method for guaranteeing \eqref{aim}-\eqref{aimcond} using any function that combines the non-conformity scores of the batch.  We suggest a specific function, the estimated likelihood ratio (LRT) statistic, which has been successfully used in hypothesis testing and has asymptotic optimality properties \citep{LR2005b}. We show in \S~\ref{sec:xp} that the batch-score algorithm using the estimated LRT statistic has excellent power, but also non-negligible increased computational complexity, compared with the suggested methods that are based on combining conformal $p$-values, as shows in \S~\ref{sec:LRTbadcomplex}. The added computational complexity is due to the fact that the null distribution of the estimated LRT statistic depends on the actual scores. In contrast, the null distribution of the combination of conformal $p$-values does not depend on the actual scores (it  does, however, depend on the number of examples from each class in the calibration set for the class conditional model).

Let $G((x_{i})_{i\in \range{m}}, (y_{i})_{i\in \range{m}})$ denote the batch-score function. We suggest
$G((x_{i})_{i\in \range{m}}, (y_{i})_{i\in \range{m}}) = \prod_{i=1}^m\frac{ \max_{k\in [K]} (1-S^{(k)}(x_{i})) }{1-S^{(y_{i})}(x_{i})}. $
So our test statistic, called the {\it estimated LRT statistic}, for testing the null hypothesis that $(Y_{n+i})_{i\in \range{m}}= (y_{i})_{i\in \range{m}}$, is
 $G((X_{n+i})_{i\in \range{m}}, (y_{i})_{i\in \range{m}}).$
This test statistic is expected to have excellent   
when  $1-S^{(k)}(X_{n+i})$ is close to the probability of observing $Y_{n+i}=k$ given $X_{n+i}$. To see this, suppose the true (unknown) batch label vector is $\tilde{y}$.
Then the approximate likelihood of observing $(X_{n+i})_{i\in \range{m}}$ together with the true $\tilde{y}$ or together with the null $y$ is, respectively, 
$\Pi_{i=1}^m \left(1-S^{(\tilde{y}_i)}(X_{n+i})\right) \mP(X_{n+i})$ or $   
 \Pi_{i=1}^m \left(1-S^{(y_{i})}(X_{n+i})\right)\mP(X_{n+i}),
$ 
where $\mP(X_{n+i})$ denotes the density of $X_{n+i}$ taken at point $X_{n+i}$ (when it exists and with a slight abuse of notation). So the approximate likelihood ratio is $\Pi_{i=1}^m{ \left(1-S^{(\tilde{y}_i)}(X_{n+i})\right)}/{\left(1-S^{(y_{i})}(X_{n+i})\right)}.$  The numerator is evaluated using the maximum likelihood estimate for $\tilde{y}$ to obtain 
$G((X_{n+i})_{i\in \range{m}}, (y_{i})_{i\in \range{m}}).$

\begin{algorithm}[!htb]
\small
\SetKwInOut{Input}{Input}
\Input{Calibration and test samples data $(X_i,Y_i)_{i\in \range{m}}$, $(X_{n+i})_{i\in \range{m}}$; a batch score function $G((x_{i})_{i\in \range{m}}, (y_{i})_{i\in \range{m}})$; level $\alpha \in (0,1)$; the number of permutations $B$
}

Initialize $C_{\alpha,G}^m \gets \emptyset$

\For{each $y=(y_{i})_{i\in \range{m}}\in \range{K}^m$ }{

    \For{each $b\in [B]$}{

        Sample $m$ indices from the vector $(Y_1,\ldots,Y_n, y_{1}, \ldots, y_{m}).$ The $m$ indices are sampled with the restriction that the frequency of the classes in the `test' sample is  $\left(m_k((y_{i})_{i\in [m]})\right)_{k\in \range{K}}$. Let $((x'_{i})_{i\in \range{m(b)}}, (y'_{i})_{i\in \range{m(b)}})$ denote the vectorized data in the `test' sample.

        Compute the $b$th null batch score $G_b:=G((x'_{i})_{i\in \range{m(b)}}, (y'_{i})_{i\in \range{m(b)}}).$
    }
 
  The $p$-value for testing that $(Y_{n+i})_{i\in \range{m}}= y$, 
  is $$p^{(y)} = \frac{1+\sum_{b=1}^{B}\ind{G_b\geq  G((X_{n+i})_{i\in \range{m}},(y_{i})_{i\in \range{m}})}}{B+1}.$$  

If $p^{(y)}>\alpha$ then $C_{\alpha,G}^m \gets C_{\alpha,G}^m \cup y.$ 
}

\SetKwInOut{Output}{Output}
\Output{Batch prediction set $\mathcal{C}^{m}_{\alpha,G}$.}

\caption{Constructing a batch prediction set using batch scores}
\label{alg:batchscores}
\end{algorithm}

\begin{proposition}\label{th:batchscorescontrol}
The  coverage of the batch prediction set $\mathcal{C}^m_{\alpha,G}$ provided in Algorithm \ref{alg:batchscores}
 is at least $1-\alpha$ both in the iid model and in the conditional model.  For the iid model, the restriction in line 4 of Algorithm \ref{alg:batchscores} is not necessary. 
\end{proposition}

The proof follows from standard theory on permutation tests, see, e.g., Theorem~2.4 in \cite{angelopoulos2024theoretical}. Specifically, for the class conditional model,  the result follows since the non-coverage probability is equal to $\mP\left(p^{(y)}\leq \alpha \mid (Y_{n+i})_{i\in \range{m}}=y\right)$, which is $\leq \alpha$ because  the $B$ null batch scores generated for a specific $y$ in lines 3--6 of Algorithm \ref{alg:batchscores}  are exchangeable with the batch score test statistic when $(Y_{n+i})_{i\in \range{m}}=y$.

\begin{remark}
We presented a computationally efficient shortcut for the bounds when using the $p$-value combining methods in \S~\ref{sec:shortcut}, and demonstrated in  \S \ref{subsec-sim-large-batches} that the bounds can be useful when $m$ is large. Unfortunately, this shortcut is not possible for the estimated LRT statistic, since its (permutation) null distribution varies with the vector $y$ being tested.     
\end{remark}

\section{Experiments}\label{sec:xp}

We study the performances of the following procedures: Bonferroni~\eqref{PredBonf}, Simes~\eqref{PredSimes}, Storey Simes (adaptive Simes~\eqref{PredSimesAdapt} with the Storey estimator~\eqref{pi0estiStorey} where $\lambda=1/2$), Fisher ~\eqref{FFisher}, and the estimated LRT (\S~\ref{sec-batchscores}). We use the conditional setting, with class calibrated conformal $p$-values~\eqref{standardpvalue}. The score function $S_k(x)$ is given by an estimator of the probability that $k$ is not the label of observation $x$.

\subsection{Gaussian multivariate setting}\label{subsec-BVN}

We illustrate the substantial advantage of the new methods over Bonferroni for inferring on batch prediction sets in settings with different signal to noise ratio (SNR). 
We consider $K=3$ categories, where the distribution of the covariate in each category is bivariate normal. The centers of the three categories are (0,0), (SNR,0), and (SNR,SNR). So the classification problem is more difficult  as the SNR decreases. See \S~\ref{appendix-sim-Gaussian} for one example of this data generation. 

In Table \ref{tabBVN} we show the results for a range of SNR values, in the setting with $n= 1200$, $m=6$, and the calibration set and test sets have a fixed and  equal number of examples from each of the three categories. As expected, using Simes is uniformly better than using Bonferroni. Adaptive Simes is far superior to both when the SNR is at most 2.5.
For strong signal, using Simes produces slightly narrower  batch prediction sets than using adaptive Simes. Fisher provides the narrowest batch prediction sets when the SNR is low. However, when the SNR is strong its performance is much worse even than Bonferroni. Thus, using Fisher is only recommended in situations where  the batch prediction set is expected to be large. 
The estimated LRT statistic outperforms all the other practical methods when the SNR is $\geq 2.5$. Moreover, its batch prediction sets are only slightly wider than those obtained using the Fisher combining method when the SNR is $\leq 2.5$. All other methods, however, require less than $1/100$ of the running time that is needed for the estimated LRT method. Thus it is the preferred method overall only if the practitioner has sufficient computing power. 

\setlength{\tabcolsep}{3pt}
{\small
 \begin{table}[h!]
 \centering
\begin{tabular}{|r|rrrrr|}
   \hline
   &  &  & Storey  &  & estimated\\
   SNR & Bonf  &  Simes &  Simes & Fisher & LRT \\ 
   \hline
 1.00 & 410.52 & 384.66 & 327.55  & {\bf 274.36} & 277.58 \\ 
 1.50 & 217.69 & 187.36 & 142.98  & {\bf 107.85} & 113.88 \\ 
 2.00 & 81.63 & 65.52 & 49.12  & {\bf 37.40} & 37.76 \\ 
 2.50 & 23.51 & 17.98 & 15.08  & 14.60 & {\bf 11.91} \\ 
 3.00 & 6.42 & 5.35 & 5.18  & 7.78 & {\bf 4.35} \\ 
 3.50 & 2.46 & 2.24 & 2.27  & 5.20 & {\bf 2.02} \\ 
 4.00 & 1.39 & 1.34 & 1.37  & 4.38 & {\bf 1.28} \\ 
 4.50 & 1.07 & 1.06 & 1.08  & 4.03 & {\bf 1.03} \\ 
    \hline
 \end{tabular}
 \caption{Average batch prediction set size at each SNR for different combining methods (columns). In bold, the combining method that produces the narrowest prediction region. $\alpha = 0.1$ and $2000$ replications. For a single data generation, the average running time on a standard PC was less than 0.05 seconds for all methods but the estimated LRT, which has an average running time of 5.7 seconds.   \label{tabBVN}}
 \end{table}
}

In Appendix \ref{appendix-sim-Gaussian}, Table \ref{tabBVN2LRT}, the non-coverage probability is shown for each method, as well as the results
for median Simes (adaptive Simes~\eqref{PredSimesAdapt} with the 'median' estimator, see \eqref{pi0estiQuantile}), 
and oracle Simes that 
uses the true (unknown in practice) $m_0(y)$. As expected, oracle Simes leads to the narrowest batch prediction sets. For low SNR, the oracle statistic with the true $m_0(y)$ is far lower than all the practical test statistics. This suggests that optimizing the choice of estimate of $m_0(y)$ may improve the inference. As mentioned at the end of \S \ref{sec:numapprox}, one direction may be to use for $\hat m_0(y)$ the minimum of several good candidates. More generally, we could also use as combining function the minimum batch $p$-value from different combining functions. 
We leave for future work the investigation of the benefits from such a compound procedure. 

In Appendix \ref{appendix-sim-Gaussian}, Tables \ref{tabBVNbounds} and \ref{tabBVNbounds2}, we  show the bounds for each SNR. The bounds using Simes are slightly tighter than using Bonferroni. Interestingly, there seems to be no clear benefit for the bounds in using adaptive Simes or Fisher. However, the bounds using the method of combining batch scores with the estimated LRT are tighter when $SNR>=2.5$. Appendix \ref{subsec-sim-large-batches} shows bounds in settings with $m$ large, which are computed using the available shortcut for the methods that combine conformal $p$-values, described in \S~\ref{sec:shortcut_supplementary}. 

\subsection{Real data sets}\label{sec:data}

We use two datasets commonly used in the machine learning community, the USPS dataset \citep{LeCun1989HandwrittenDR} with  $K=10$ digits and the CIFAR-$10$ dataset \citep{Krizhevsky2009LearningML} restricted to $K=3$ classes: ``birds", ``cats" and ``dogs".
For the USPS dataset, the  calibration and batch sample sizes are $n=700$ and  $m=3$, respectively. The score functions are derived using a support-vector classifier with the linear kernel (trained with $2431$ examples).  
For the CIFAR-$10$ dataset, the calibration and batch sample sizes  are $n=2000$ and  $m=5$, respectively. We use a convolutional neural network with $8$ layers, trained with $5666$ examples with $10$ epochs and the `Adam' optimizer. 

 The coverage and   violin plots of the size of the batch prediction sets for the different methods are displayed in Table~\ref{tabCovRealData} and  Figure~\ref{fig:Power}, respectively.  
For the USPS data set, the results strongly depend on the level $\alpha$ considered. For $\alpha=0.01$, the batch prediction sets are all large and Fisher and LRT methods are the best. For $\alpha=0.05$ and $\alpha=0.1$, the best batch prediction sets are the LRT and the Simes methods. 
For the CIFAR data set, the sizes of the prediction sets are  large for all $\alpha$ levels considered (meaning that the prediction task is more difficult on this data set). The Fisher combination and LRT method are comparable and better than the other ones, followed by the Storey Simes method. These findings corroborate those of the previous section. Other qualitatively  similar results are obtained in \S~\ref{sec:appenddata}.
Finally, we note that for $\alpha=10\%$ the estimated coverage is less than one standard error (SE) below $90\%$.

\setlength{\tabcolsep}{3pt}
{
\scriptsize
 \begin{table}[h!]
 \centering
\begin{tabular}{|r|rrr|rrr|}
   \hline
   & \multicolumn{3}{|c|}{USPS} & \multicolumn{3}{|c|}{CIFAR}\\
Coverage & 0.99 & 0.95 & 0.90 & 0.99 & 0.95 & 0.90\\
   \hline
Bonf. & 1& 0.966 & 0.932  & 0.993 & 0.958 & 0.896 \\
Simes & 1& 0.966 & 0.928 & 0.993 & 0.958 & 0.890 \\
Storey & 1 & 0.986 & 0.942& 0.993 & 0.954 & 0.893  \\
LRT & 0.990 & 0.962 & 0.916& 0.992 & 0.957 & 0.892  \\
Fisher & 0.994 & 0.972 & 0.932  & 0.994 & 0.953 & 0.892\\
    \hline
 \end{tabular}
 \caption{Estimated coverage for $\alpha\in\{1\%, 5\%, 10\% \}$ and data sets USPS and CIFAR (in columns) and different procedures (in rows). Based on 500 simulation runs; all the standard errors are below $0.014$
 }\label{tabCovRealData}
 \end{table}
}

\begin{figure}[h!]
 \begin{center}
\begin{tabular}{cc}
 USPS data set & 
 CIFAR data set\vspace{-2mm}\\
\includegraphics[scale=0.18]{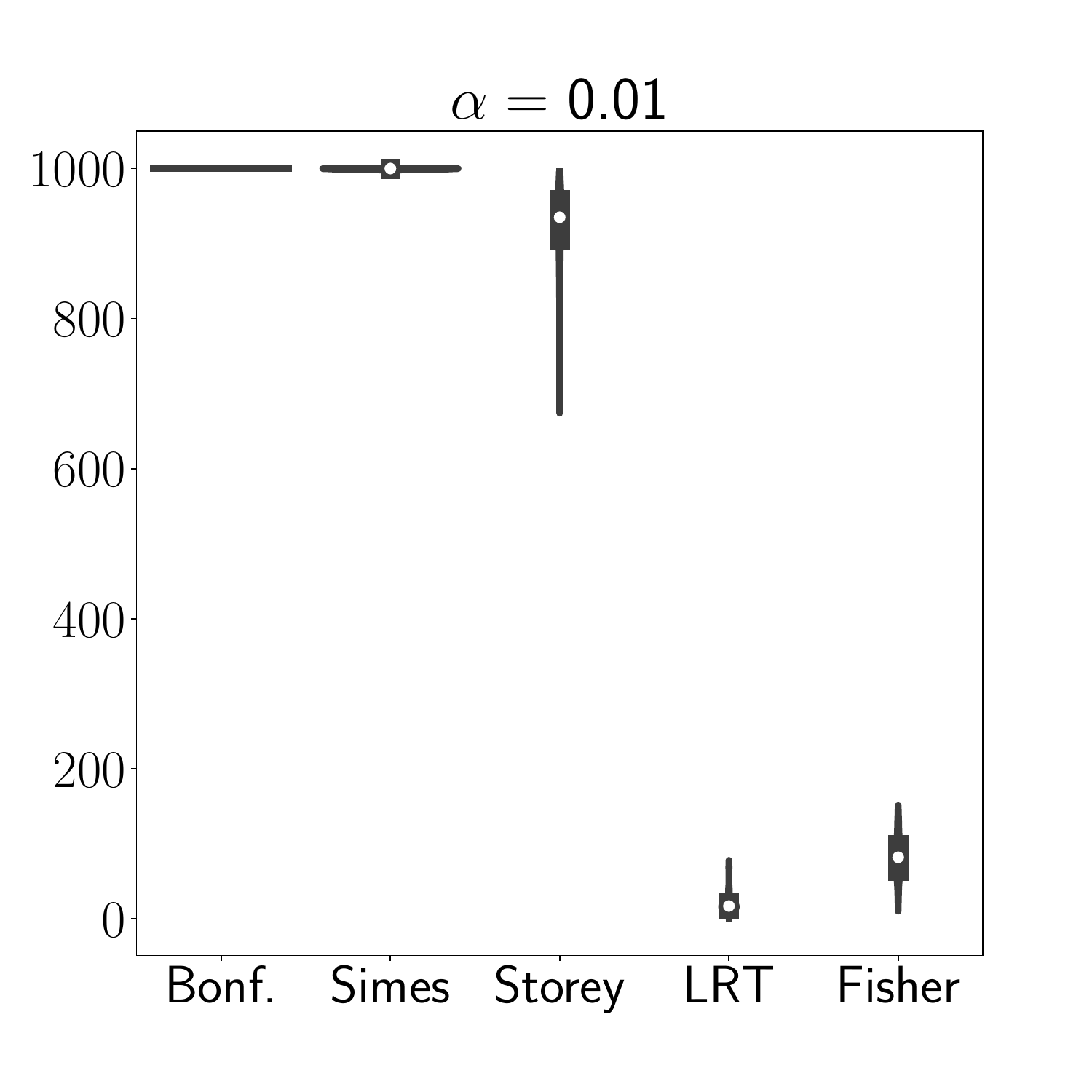}    &  
  \includegraphics[scale=0.18]{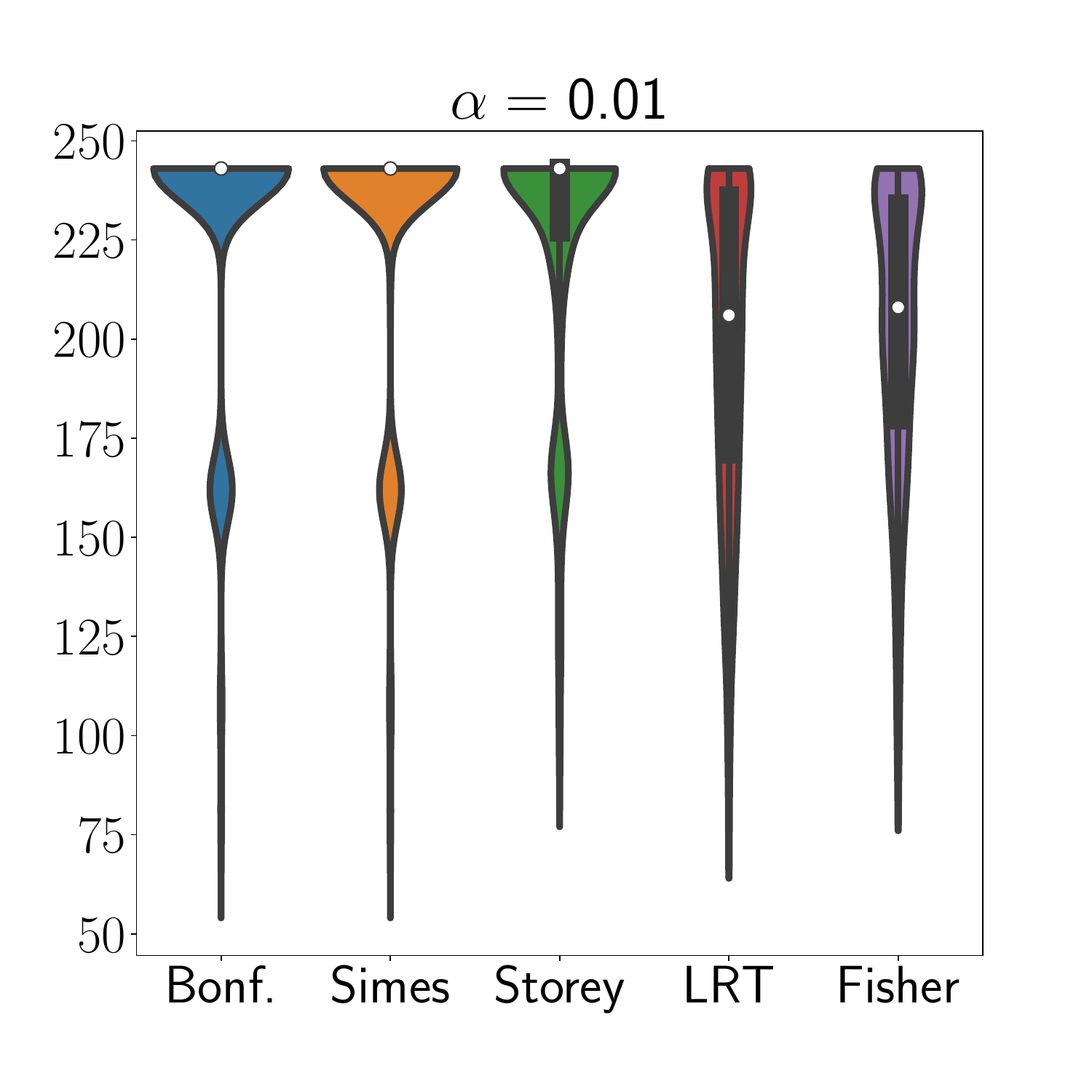} \vspace{-7mm}
\\ 
\includegraphics[scale=0.18]{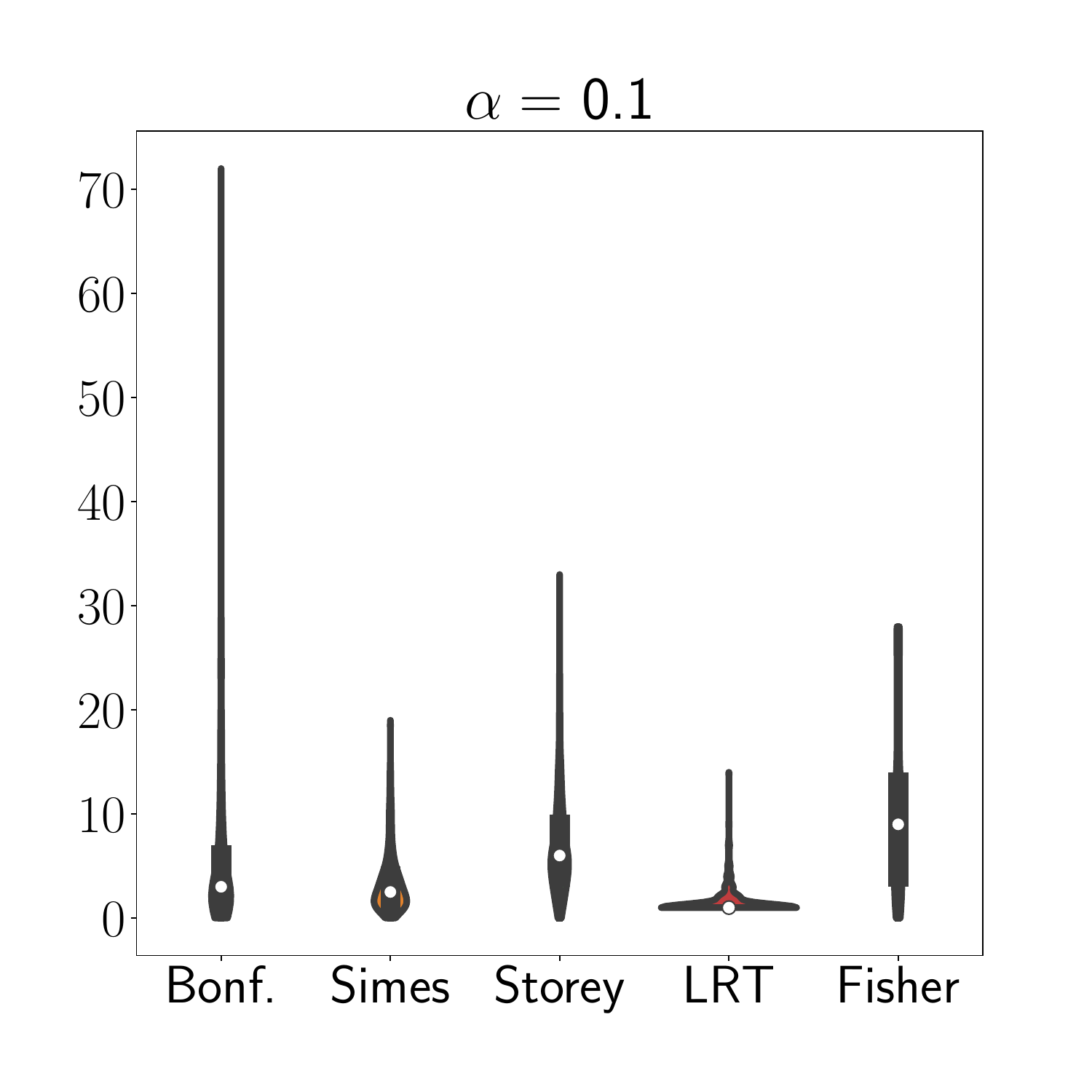}    &  
 \includegraphics[scale=0.18]{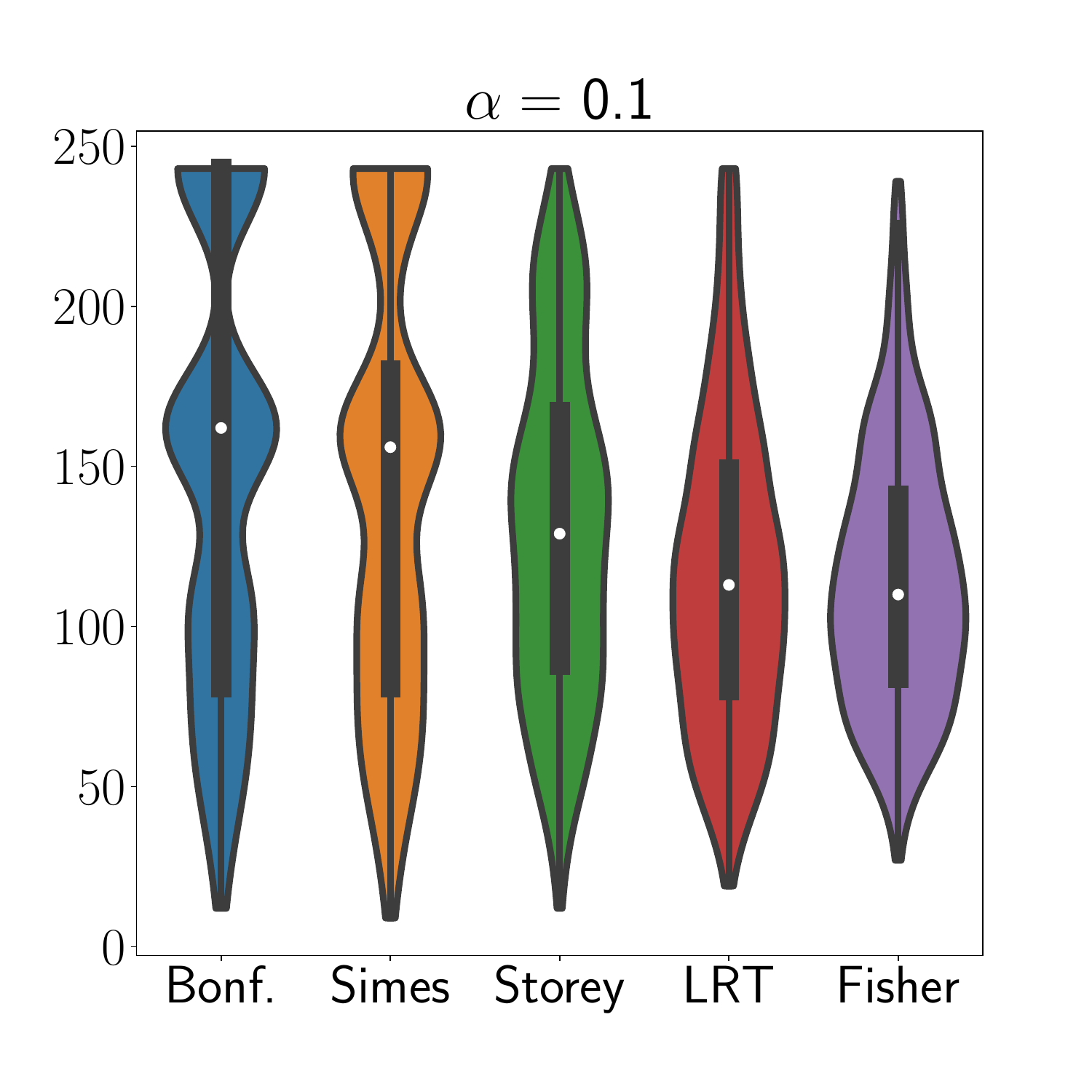} 
\end{tabular}
\end{center}
\caption{\label{fig:Power} 
Violin plots for the size of the batch prediction set for $\alpha\in\{1\%, 10\% \}$ (rows) and data sets USPS and CIFAR (columns), see details in the text. The white dot inside the inter-quartile box of the violin plot is the median. The plots for $\alpha=5\%$ (omitted) are qualitatively similar to the plots for $\alpha =10\%$.
}
\end{figure}

\section{Discussion}
For a batch of test points we provide,  with a $(1-\alpha)$ coverage guarantee, a batch prediction set or bounds for the different classes, by testing that the batch label vector is  $y\in \range{K}^m$ using two approaches: conformal $p$-value  and batch score combination tests.
We demonstrated that we can get  much narrower batch prediction sets than using Bonferroni. For the bounds, the advantage over Bonferroni is modest, but nevertheless with Simes the improvement over Bonferroni is uniform.  A further improvement is noted using the estimated LRT, the statistic suggested for the batch score combination test. However, the computation complexity is much larger with the estimated LRT, since  the permutation null distribution of the batch score combination test depends on the $n+m$ scores. This is in contrast with Simes (and all other conformal $p$-value combination tests), for which  the null distribution  depends only on the number of examples from each class in the calibration sample, not on the realized scores, so in this sense it is distribution free.

As our numerical experiments show, there is no best method for combining the conformal $p$-values. Broadly, Fisher is good for weak signal and adaptive Simes is a better choice for the remaining cases. We can also consider combining the two using the algorithm in \S~\ref{sec:numapprox} (a reasonable combining method is to take the minimum of the Fisher based and median Simes \S~\ref{sec:adapt}). The lack of an overall best combining method is not surprising, since for every combining function that is reasonable there is a data generation that is optimal for it in a specific sense \citep{Birnbaum54,heard2018choosing}.

Our examples concentrated on a fairly small batch size $m$ and class size $K$. For $m$ or $K$ large we suggested,  instead of testing all $y\in \range{K}^m$ to produce the bounds, to use a shortcut  with computational complexity $O(K\cdot m^2)$. It is exact for $K=2$, and appears tight for $K>2$ in our numerical experiments.   
Specifically for Simes type combination tests, computationally efficient algorithms have been developed in the multiple testing literature \citep{Goeman19, Andreella23}. For large $m$ and $K$ it may be  worthwhile to  consider adapting their algorithms to our set-up for greater computational efficiency. A great  challenge is to provide, for $m$ or $K$ large, efficient algorithms that directly target approximating the batch prediction set (rather than via the bounds). Relatedly, an open question is how to concisely summarize the batch prediction set when it is large.

\section*{Acknowledgements}
The authors acknowledge  grants ANR-21-CE23-0035 (ASCAI) and ANR-23-CE40-0018-01 (BACKUP) of the French National Research Agency ANR,  the Emergence project MARS of Sorbonne Universit\'e, and Israel Science Foundation grant no. 406/24.

\bibliography{biblio}

\newpage
\appendix

\section{Estimators for $m_0(y)$}\label{sec:detailAdaptive}

This section complements \S~\ref{sec:adapt}.

\subsection{Storey and quantile type estimator}

We first provide the general formula \eqref{equm0y} for the Storey-type estimator $\hat{m}_0(y)$ that can accommodate any choice of $\lambda\in (0,1)$.

First, in the iid model, the modification corresponds to a simple rounding: 
$$
\hat{m}_0(y) := (1-\lambda)^{-1}\Big(1+\sum_{i\in \range{m}} \ind{p^{(y_{i})}_i\geq \lfloor (n+1)\lambda\rfloor/(n+1) }\Big).
$$
Clearly, the formula reduces to \eqref{pi0estiStoreysimple} when $(n+1)\lambda$ is an integer.

In the conditional model, the modification corresponds to a rounding on each class: 
\begin{align}
&\hat{m}_0(y) := \kappa(y) \left(1+\sum_{k\in \range{K}} \sum_{i:y_{i}=k} \ind{p^{(k)}_i\geq \lambda_{k} }\right),\label{pi0estiStorey}
\end{align}
with $\lambda_{k}=\frac{\lfloor \lambda(n_k+1)\rfloor}{n_k+1}$ for $k\in \range{K}$. 
Above, the parameter $\kappa(y)$ is given by
\begin{equation}\label{equkappaClass}
\kappa(y)= \Big(1-\min_{k\in\range{K}}\lambda_k\Big)^{\frac{1}{m-1}}  \times \prod_{k\in\range{K}} \Big(\frac{1}{1-\lambda_k}\Big)^{\frac{m_k(y)}{m-1}},
\end{equation}
where we recall that $m_k(y)$ is given by \eqref{countk}.
 When $(n_k+1)\lambda$ is an integer for each $k\in \range{K}$, then $\lambda_k=\lambda$, $\kappa(y)=(1-\lambda)^{-1}$, and  the formula reduces to \eqref{pi0estiStoreysimple}.

Second, the `quantile' estimator \citep{BKY2006} is given by
\begin{equation}\label{pi0estiQuantile}
\hat{m}_0(y) = \frac{m-\l+1}{1-p_{(\l)}(y)},
\end{equation}
for some $\l\in \range{m}$, typically $\l=\lceil m/2\rceil$ for the `median' estimator.
The adaptive Simes batch prediction set using the quantile estimator satisfies the correct coverage in the iid model by \cite{marandon2024adaptive}.
Proving such a coverage result in the class-conditional model is an open problem, although our numerical experiments seem to indicate that the control is maintained in that case (for the median estimator).\footnote{Recall that a valid coverage for the quantile Simes procedure can be ensured by using the empirical method of \S~\ref{sec:numapprox} (not used in our numerical experiments).}

\subsection{Choice of the tuning parameters}\label{sec:choosinglambda}

We discuss the choice of the parameter $\lambda\in(0,1)$ in the Storey estimator \eqref{pi0estiStorey} (procedure denoted by $\lambda$-S for short), and of the parameter $\ell$ in the quantile estimator \eqref{pi0estiQuantile}. In the latter, we let $\ell=\lceil qm \rceil$ and discuss rather the choice of $q$ (the corresponding procedure is denoted by   $q$-Q  for short).

The results are displayed in Figure~\ref{fig:Power} for the USPS and CIFAR data sets. For the Storey estimator, while no choice of $\lambda$ seems to be universally the best, this choice affects the performance of the method: 
we observe that choosing 
$\lambda=1/2$ is better for the data set with weak signal (CIFAR) 
while choosing $\lambda$ small (and of the order of $\alpha$) is better for the data set with strong signal (USPS). This is coherent with the intuition behind the Storey estimator which implicitly supposes that the $p$-values above $\lambda$ are under the null. For the quantile procedure, the conclusion is similar to some extent, but the median procedure seems to have a good behavior for both data sets. Roughly, the latter can be seen as a Storey procedure with an adaptive choice $\lambda=p_{(\lceil m/2\rceil)}(y)$, which is able to better adapt to the signal strength. These conclusions corroborate previous findings in the multiple testing literature under independence \citep{BKY2006,BR2009}, see also \citep{Birnbaum54,heard2018choosing}.

\begin{figure}[h!]
\begin{tabular}{cc}
\hspace{-1cm} USPS data set & \hspace{-1.2cm}CIFAR data set\vspace{-7mm}\\
\hspace{-1cm} \includegraphics[scale=0.25]{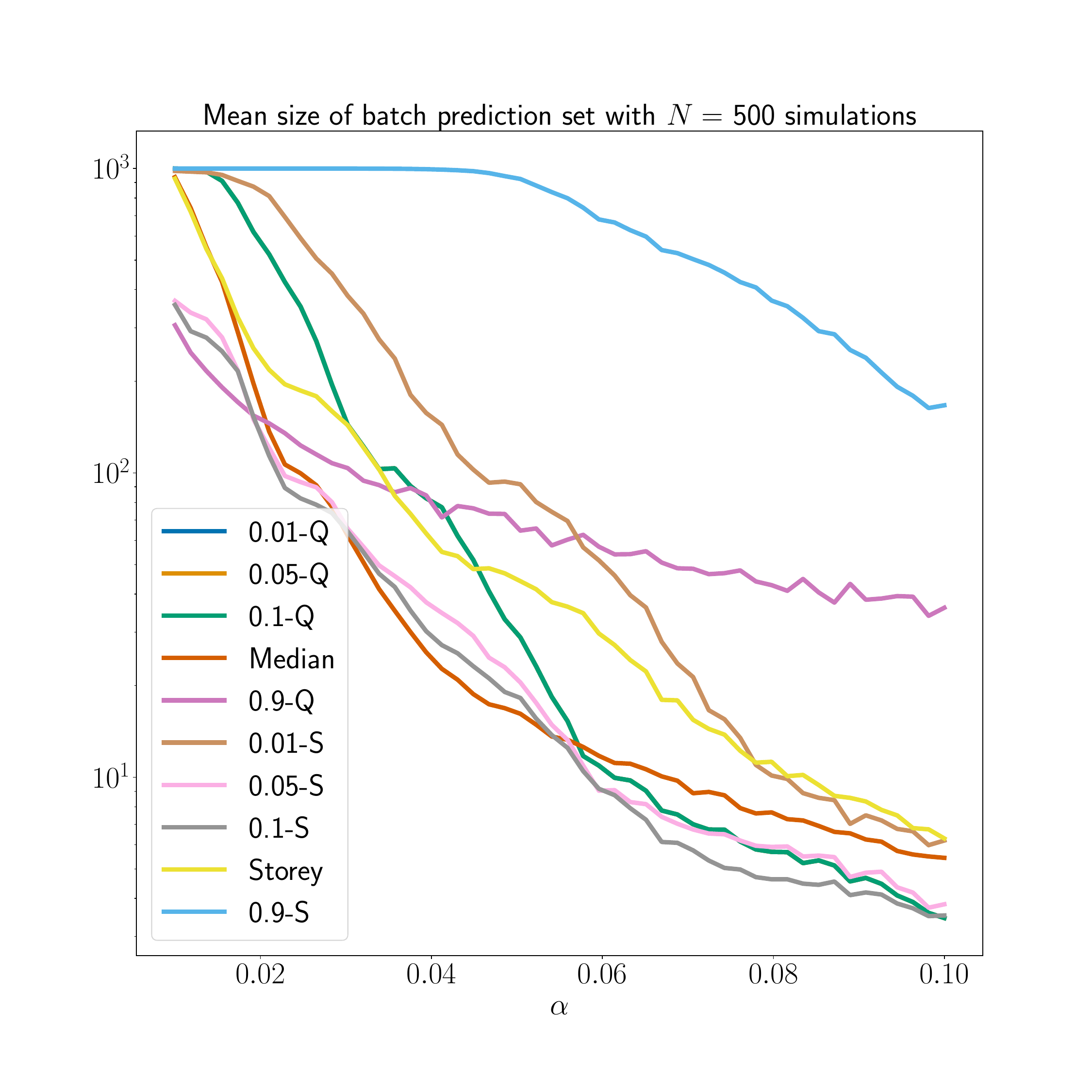}    &  \hspace{-1.2cm}  \includegraphics[scale=0.25]{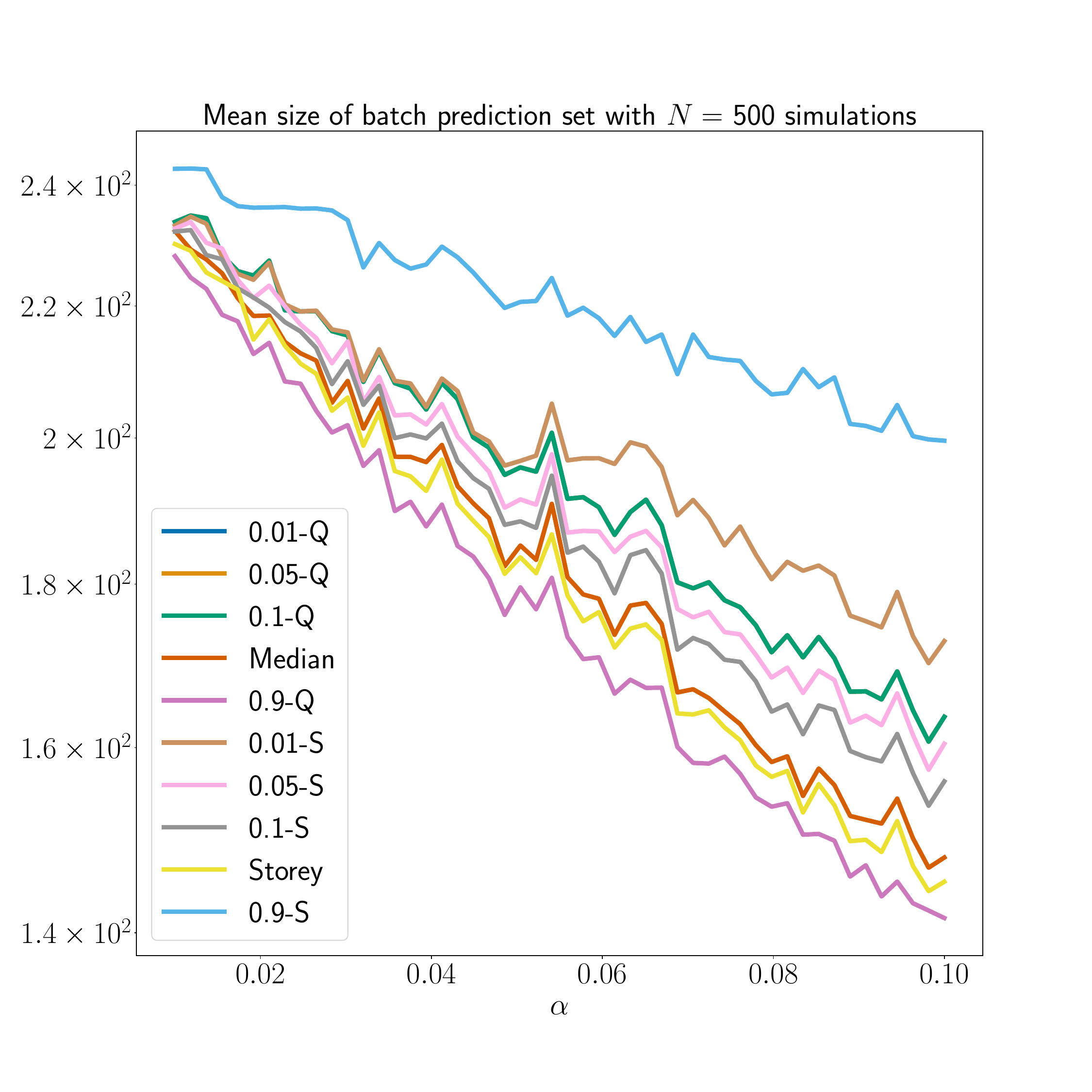} \vspace{-7mm}\\
\end{tabular}
\vspace{-0.3cm}
\caption{\label{fig:Comparison} 
 Averaged size of the batch prediction sets in function of $\alpha$ for different procedures (see text). Storey is $0.5$-S and Median is $0.5$-Q. Same setting as for Figure~\ref{fig:Power} .
}
\end{figure}

\section{Proofs}

In this section, we prove Theorems~\ref{th:Simes},~\ref{th:Simesexact}~and~\ref{th:SimesAdapt}. The proofs follow from previous literature for the iid model (and full-calibrated $p$-values):
\begin{itemize}
    \item Theorem~\ref{th:Simes} for the iid model is a consequence of \cite{BY2001} and of the fact that the full-calibrated $p$-values are PRDS \citep{bates2023testing} (see definition below);
    \item Theorem~\ref{th:Simesexact} for the iid model is a consequence of Corollary~3.5 in \cite{marandon2024adaptive};
    \item Theorem~\ref{th:SimesAdapt} for the iid model is a consequence of Corollary~3.7 in \cite{marandon2024adaptive}.
\end{itemize}
{Below, we extend these arguments to the case of  {\it the conditional model with class-calibrated $p$-values}. The main technical tool for the proof is  Lemma~\ref{lemmaPropConformalClass} (for comparison, we also recall Lemma~\ref{lemmaPropConformalFull} that was obtained for the iid case with full calibrated $p$-values)}. 
On an intuitive point of view, the main idea of this extension is that, conditionally on $(Y_{j})_{j\in \range{n+m}}$, each class-conditional conformal $p$-value $p^{(Y_{n+i})}_i$ depends on the $p$-values of the same class $(p^{(Y_{n+j})}_j)_{j\in \range{m}\backslash\{i\}:Y_{n+j}=Y_{n+i}}$ exactly in the same way as for the iid case, and are independent of the $p$-values of the other classes $(p^{(Y_{n+j})}_j)_{j\in \range{m}\backslash\{i\}:Y_{n+j}\neq Y_{n+i}}$.

Below, we write $p_i$ instead of $p^{(Y_{n+i})}_i$ for simplicity. Also, $n_i$ stands for $n_{Y_{n+i}}$ with a slight abuse of notation {(recall that $n_k$ is the cardinal of $\mathcal{D}^{(k)}_{{\tiny \mbox{cal}}}$)}.

\subsection{Proof of Theorem~\ref{th:Simes}}\label{sec:proofSimes}

It is sufficient to establish the following Simes inequality for class-calibrated $p$-values:
\begin{align}
&\P(\exists \ell \in [m],\:{p}_{(\ell)}\leq \alpha \ell/m \:|\: (Y_{j})_{j\in[n+m]})\leq \alpha\,.\label{equ:Simescond}
\end{align}
Since the families of class-calibrated $p$-values are marginally super-uniform (conditionally on  $(Y_{n+i})_{i\in[m]}$), see Proposition~\ref{prop:marginal}, and by classical FDR controlling theory \citep{BY2001},  it is enough to prove that the following PRDS property {on $m$} holds: for any nondecreasing\footnote{A set $D\subset [0,1]^m$ is nondecreasing if for $x=(x_j)_{1\leq j\leq m}\in D$ and $y=(y_j)_{1\leq j\leq m}\in \R^m$, $(\forall j\in \range{m}, x_j\leq y_j)$ implies $y\in D$.} set $D\subset [0,1]^m$, the function
$$
u\mapsto \P((p_i)_{i\in \range{m}}\in D\:|\: p_i=u, (Y_{j})_{j\in[n+m]}),
$$
is nondecreasing for all $i\in\range{m}$.

\begin{proposition}\label{prop:PRDS4Class}
    {In the conditional model,} the family of class-calibrated $p$-values is PRDS on $\range{m}$.
\end{proposition}

The proof relies on the general property of Lemma~\ref{lem:PRDSgroup}, establishing that per-group PRDS for independent groups yields entire set PRDS.
\begin{proof}
    Let us work conditionally on $(Y_{j})_{j\in[n+m]}$ and consider the partition of $\range{m}$ given by $G_k=\{j\in\range{m}\::\:Y_{n+j}=k\}$ then we know that for each $k\in [K]$, $(p_j)_{j\in G_k }$ is a family which is PRDS on $G_k$ \citep{bates2023testing}. In addition, the $p$-values $(p_j)_{j\in G_k }$ and $(p_j)_{j\in G_{k'} }$ are independent for $k\neq k'$, because the calibration samples are not the same for each (since we use class-calibrated $p$-values). Hence, we can apply Lemma~\ref{lem:PRDSgroup} to conclude.
\end{proof}

\subsection{Proof of Theorem~\ref{th:Simesexact}}\label{sec:proofSimesexact}
{To establish the result, we use the well known relationship between the Simes inequality and the FDR control of BH procedure under the full null, see, e.g., \S~2.2 in \cite{barber2017p}.} 
Let us denote for any $y=(y_{i})_{i\in \range{m}}\in \range{K}^m$,
\begin{equation}\label{BHrej}
    \wh{\ell}(\mathbf{p}(y))=\max\{ \ell\in \range{m}\::\: p_{(\ell)}(y)\leq \alpha \ell/m\},
\end{equation}
(with the convention $\wh{\ell}(\mathbf{p}(y))=0$ if the set is empty)
the number of rejections of the BH procedure \citep{BH1995} associated to the $p$-value family $\mathbf{p}(y)=(p_i^{(y_{i})})_{i\in \range{m}}$. Observe that, $y\notin \mathcal{C}^m_{\alpha,\mbox{\tiny Simes} }$ if and only if $\wh{\ell}(\mathbf{p}(y))\geq 1$. In addition, the latter holds if and only if $\sum_{i\in \range{m}}\ind{p_i^{(y_{i})}\leq (\alpha/m) (1\vee \wh{\ell}(\mathbf{p}(y)))}=1\vee \wh{\ell}(\mathbf{p}(y))$. 

Therefore, denoting  $\mathbf{p}=(p_i)_{i\in \range{m}}$ the family of class-calibrated $p$-values, we can express the non-coverage probability as follows: 
\begin{align}\label{fromSimestoBH}
     \P( (Y_{n+i})_{i\in \range{m}}\notin \mathcal{C}^m_{\alpha,\mbox{\tiny Simes}}\:|\: (Y_{j})_{j\in[n+m]})
     &=\sum_{i\in \range{m}} \E\Big[\frac{\ind{p_i\leq (\alpha/m) (1\vee  \wh{\ell}(\mathbf{p}))}}{1\vee \wh{\ell}(\mathbf{p})}\:\Big|\: (Y_{j})_{j\in[n+m]}\Big].
\end{align}
Consider $\mathbf{p}'=(p'_i)_{i\in \range{m}}$ the vector defined in Lemma~\ref{lemmaPropConformalClass} (v) with in addition $p'_j=p_j$ for $j\in \range{m}:Y_{n+j}\neq Y_{n+i}$. Combining Lemma~\ref{lemmaPropConformalClass} (v) with Lemma~\ref{BHsmallerp}, we obtain 
$$
\{p_i\leq \alpha \wh{\ell}(\mathbf{p})/m \}=\{ p_i\leq \alpha \wh{\ell}(\mathbf{p}')/m \}\subset \{ \wh{\ell}(\mathbf{p})=\wh{\ell}(\mathbf{p}')\}.
$$
Hence, by letting $L_i= 1\vee \wh{\ell}(\mathbf{p}')\in \range{m}$, which is $W_i$-measurable \et{(as defined in Lemma~\ref{lemmaPropConformalClass})}, we have that \eqref{fromSimestoBH} can be written as
\begin{align*}
     \P( (Y_{n+i})_{i\in \range{m}}\notin \mathcal{C}^m_{\alpha,\mbox{\tiny Simes}}\:|\: (Y_{j})_{j\in[n+m]})
     &=\sum_{i\in \range{m}} \E\Big[\frac{\ind{p_i\leq (\alpha/m) L_i}}{L_i}\:\Big|\: (Y_{j})_{j\in[n+m]}\Big]\\
     &=\sum_{i\in \range{m}} \E\Big[\frac{\P(p_i\leq (\alpha/m) L_i \:|\: W_i)}{L_i}\:\Big|\: (Y_{j})_{j\in[n+m]}\Big].
\end{align*}
Now, by Lemma~\ref{lemmaPropConformalClass} (ii), we have $\P(p_i\leq (\alpha/m) L_i \:|\: W_i)=\frac{\lfloor (n_i+1) (\alpha/m) L_i\rfloor}{n_i+1} = (\alpha/m) L_i$ if $(n_i+1) (\alpha/m)$ is an integer for all $i\in \range{m}$. This finishes the proof.

\subsection{Proof of Theorem~\ref{th:SimesAdapt}}\label{sec:proofSimesadapt}
\et{Recall $\lambda_{k}=\frac{\lfloor \lambda(n_k+1)\rfloor}{n_k+1}$ for $k\in \range{K}$.}  
For short, we sometimes write in this proof $\lambda_i$, $m_i$ and $n_i$ instead of $\lambda_{Y_{n+i}}$, $m_{Y_{n+i}}$ and $n_{Y_{n+i}}$ respectively, for all $i\in \range{m}$. Also, we write $\kappa$ instead of $\kappa((Y_{n+i})_{i\in \range{m}})$ and $m_k$ instead of $m_k((Y_{n+i})_{i\in \range{m}})$.

Let $G(\mathbf{p})=\hat{m}_0((Y_{n+i})_{i\in \range{m}}) = \kappa({1+\sum_{i=1}^m \ind{p_i\geq \lambda_i}})$ the  estimator of $m_0$ at the true point $(Y_{n+i})_{i\in \range{m}}$ given in \eqref{pi0estiStorey} \et{(this means that this proof deals with the general case and not only the simple Storey estimator described in \eqref{pi0estiStoreysimple})}. 
Similarly to \eqref{fromSimestoBH}, we have 
\begin{align*}
     \P( (Y_{n+i})_{i\in \range{m}}\notin \mathcal{C}^m_{\alpha,\mbox{\tiny A-Simes}}\:|\: (Y_{j})_{j\in[n+m]})
     &=\sum_{i\in \range{m}} \E\Big[\frac{\ind{p_i\leq (\alpha/G(\mathbf{p})) (1\vee  \wh{\ell}(\mathbf{p}))}}{1\vee \wh{\ell}(\mathbf{p})}\:\Big|\: (Y_{j})_{j\in[n+m]}\Big]
\end{align*}
for $\wh{\ell}(\mathbf{p})=\max\{ \ell\in \range{m}\::\: p_{(\ell)}\leq \alpha \ell /G(\mathbf{p})\}$ (with the convention $\wh{\ell}(\mathbf{p})=0$ if the set is empty).
Now we use Lemma~\ref{lemmaPropConformalClass} and the notation therein, and we observe that
\et{
$$(p_j)_{j\in \range{m}\backslash\{i\}}=( (p_j)_{j\in\range{m}\backslash\{i\}:Y_{n+j}=Y_{n+i} }, (p_j)_{j\in\range{m}:Y_{n+j}\neq Y_{n+i} })=(\Psi_i(p_i,W_i), (p_j)_{j\in\range{m}:Y_{n+j}\neq Y_{n+i} })$$
(with some abuse of notation in the ordering of the vector) 
is a function of $(p_i, W_i)$ which is nondecreasing in $p_i$}. Hence, $1/G(\mathbf{p})$ and $1\vee  \wh{\ell}(\mathbf{p})$ are functions of $(p_i, W_i)$, \et{say $\Psi_2(p_i, W_i)$ and $\Psi_3(p_i, W_i)$ respectively, which are both nonincreasing in $p_i$.} Now let
\begin{align*}
c^*(W_i)&=\max \mathcal{N}(W_i)\\
\mathcal{N}(W_i)&= \{ a/(n_i+1)\::\: a \in \range{n_i+1}, a/(n_i+1)\leq \et{\alpha\Psi_2(a/(n_i+1), W_i) \Psi_3( a/(n_i+1),W_i)}\} ,
\end{align*}
with the convention $c^*(W_i)=(n_i+1)^{-1}$ if $\mathcal{N}(W_i)$ is empty. Since $1\vee \wh{\ell}(\mathbf{p})\geq 1\vee \wh{\ell}(c^*(W_i),(p_j)_{j\in \range{m}\backslash\{i\}})$, we have 
\begin{align*}
     \P( (Y_{n+i})_{i\in \range{m}}\notin \mathcal{C}^m_{\alpha,\mbox{\tiny A-Simes}}\:|\: (Y_{j})_{j\in[n+m]})
     &\leq \sum_{i\in \range{m}} \E\Big[\frac{\P(p_i\leq c^*(W_i), p_i\in \mathcal{N}(W_i)\:|\: W_i)}{\et{\Psi_3( c^*(W_i),W_i)}}\:\Big|\: (Y_{j})_{j\in[n+m]}\Big]\\
     &\leq \sum_{i\in \range{m}} \E\Big[\frac{c^*(W_i)}{\et{\Psi_3( c^*(W_i),W_i)}}\:\Big|\: (Y_{j})_{j\in[n+m]}\Big]\\
     &\leq \et{\alpha}\sum_{i\in \range{m}} \E\Big[\et{\Psi_2(1/(n_i+1), W_i)}\Big|\: (Y_{j})_{j\in[n+m]}\Big],
\end{align*}
where the first inequality comes from the definition of $\mathcal{N}(W_i)$ and $c^*(W_i)$ and from the fact that $\et{\Psi_3( c^*(W_i),W_i)}$ is $W_i$-measurable; the second inequality comes from Lemma~\ref{lemmaPropConformalClass} (ii); and the third one comes from the fact that $c^*(W_i)$ is in $\mathcal{N}(W_i)$ and $\Psi_2(a/(n_i+1), W_i)$ is nonincreasing in $a$. Given the notation of Lemma~\ref{lemmaPropConformalClass} (v), this leads to 
\begin{equation}\label{firstbigstepAdapt}
    \P( (Y_{n+i})_{i\in \range{m}}\notin \mathcal{C}^m_{\alpha,\mbox{\tiny A-Simes}}\:|\: (Y_{j})_{j\in[n+m]})
     \leq \et{\alpha} \sum_{i\in \range{m}} \E\Big[\frac{1}{G(\mathbf{p}')}\Big],
\end{equation}
where $\mathbf{p}'=(p'_j)_{j\in \range{m}}$ is such that $p'_i=(n_i+1)^{-1}$, $(p'_j)_{j\in \range{m}:Y_{n+j}=Y_{n+i}} \sim \et{\mathcal{L}_{i,(Y_{j})_{j\in \range{n+m}}}}$ and for each $k\neq Y_{n+i}$, $(p'_j)_{j\in \range{m}:Y_{n+j}=k} \sim \et{\mathcal{L}^{(k)}_{(Y_{j})_{j\in \range{n+m}}}}$ where the distribution of $\et{\mathcal{L}_{i,(Y_{j})_{j\in \range{n+m}}}}$ and $\et{\mathcal{L}^{(k)}_{(Y_{j})_{j\in \range{n+m}}}}$ are defined in Lemma~\ref{lemmaPropConformalClass}. Also note that $(p'_j)_{j\in \range{m}:Y_{n+j}=Y_{n+i}}$ and all $(p'_j)_{j\in \range{m}:Y_{n+j}=k}$, $k\neq Y_{n+i}$, are independent vectors, so that the distribution of $\mathbf{p}'$ is well specified. Now observe that \et{(all expectations/probabilities below are taken implicitly conditionally on $(Y_{j})_{j\in \range{n+m}}$)}
\begin{align*}
    \E\Big[\frac{1}{G(\mathbf{p}')}\Big]&=\E\Big[\frac{1/\kappa}{1+\sum_{j=1}^m \ind{p'_j\geq \lambda_j}}\Big]\\
    &= \E\Big[\frac{1/\kappa}{1+\sum_{k\neq Y_{n+i}} \sum_{j: Y_{n+j}=k} \ind{p'_j\geq \lambda_j} + \sum_{j\in \range{m}\backslash\{i\}:Y_{n+j}=Y_{n+i}} \ind{p'_j\geq \lambda_j} }\Big]\\
    &= \E\Big[\frac{1/\kappa}{1+\sum_{k\neq Y_{n+i}}  \mathcal{B}(m_k,\nu_k) + \mathcal{B}(m_i-1,\nu'_i)
     }\Big],
\end{align*}
by using Lemma~\ref{lemmaPropConformalClass} (iii), (iv), 
where $\mathcal{B}(a,b)$ denotes (independent) binomial variables of parameters $a$ and $b$, and where
$\nu_k=U^{(k)}_{(\lfloor (n_k+1)\lambda\rfloor -1)}$ (with the convention $\nu_k=1$ if $\lfloor (n_k+1)\lambda\rfloor\leq 1$) and $\nu'_i=U_{(\lfloor (n_i+1)\lambda\rfloor)}$ (with the convention $\nu'_i=1$ if $\lfloor (n_i+1)\lambda\rfloor=0$).
The latter comes from the fact that for $j\in \range{m}$ such that $Y_{n+j}=k\neq Y_{n+i}$, 
\begin{align*}
\P(p'_j\geq \lambda_j \:|\: (U^{(k)}_{(1)},\dots,U^{(k)}_{(n_k)}))
&= \P\Big(\sum_{s\in A_i} \ind{s\geq S_{n+j}} > \lfloor \lambda(n_k+1)\rfloor -1 \:\Big|\: (U^{(k)}_{(1)},\dots,U^{(k)}_{(n_k)})\Big)\\
&=1-(1- U^{(k)}_{(\lfloor (n_k+1)\lambda\rfloor -1)}) = U^{(k)}_{(\lfloor (n_k+1)\lambda\rfloor -1)}.
\end{align*}
Similarly, for $j\neq i$ such that $Y_{n+j}=Y_{n+i}$, $\P(p'_j\geq \lambda_j \:|\: (U_{(1)},\dots,U_{(n_i+1)}))=U_{(\lfloor (n_i+1)\lambda\rfloor)}$.

Now, by Lemma~\ref{lem:beta}, we have 
$\nu_k\sim \beta(n_k+2-\lfloor (n_k+1)\lambda\rfloor, \lfloor (n_k+1)\lambda\rfloor -1)$ and $\nu'_i\sim \beta(n_i+2-\lfloor (n_i+1)\lambda\rfloor, \lfloor (n_i+1)\lambda\rfloor )$.
Let $\nu$ be the random variable
$$
\nu= (\nu'_i)^{m_i/m} \prod_{k\neq Y_{n+i}} (\nu_k)^{m_k/m}.
$$
By the stochastic domination argument of Lemma~\ref{lem:domin}, we have
\begin{align*}
   & \E\Big[\frac{1}{1+\sum_{k\neq Y_{n+i}}  \mathcal{B}(m_k,\nu_k) + \mathcal{B}(m_i-1,\nu'_i)}\:\Big|\:(\nu_k)_{k\neq Y_{n+i}}, \nu'_i\Big]\\
   &\leq
   \E\Big[\frac{1}{1+\sum_{k\neq Y_{n+i}}  \mathcal{B}(m-1,\nu)}\:\Big|\:(\nu_k)_{k\neq Y_{n+i}}, \nu'_i\Big]\leq 1/(m \nu), 
\end{align*}
where we used Lemma~\ref{lem:BY} in the last inequality. As a result, 
\begin{align*}
   \sum_{i\in \range{m}}  \E\Big[\frac{1}{G(\mathbf{p}')}\Big]
   &\leq \kappa^{-1} m^{-1}\sum_{i\in \range{m}}\E\Big((\nu'_i)^{-(m_i-1)/(m-1)} \prod_{k\neq Y_{n+i}} (\nu_k)^{-m_k/(m-1)}\Big)\\
   &=\kappa^{-1} m^{-1}\sum_{i\in \range{m}}\E\big((\nu'_i)^{-(m_i-1)/(m-1)}\big) \prod_{k\neq Y_{n+i}} \E\big((\nu_k)^{-m_k/(m-1)}\big),
\end{align*}
by using the independence between the variables $\nu'_i$, $\nu_k$, $k\neq Y_{n+i}$. By Jensen's inequality, the last display is at most 
\begin{align*}
&  m^{-1}\sum_{i\in \range{m}}\kappa^{-1} (\E((\nu'_i)^{-1}))^{(m_i-1)/(m-1)} \prod_{k\neq Y_{n+i}} (\E((\nu_k)^{-1})))^{m_k/(m-1)}\\
  &= m^{-1}\sum_{i\in \range{m}} \kappa^{-1}\Big(\frac{n_i+1}{n_i+1-\lfloor (n_i+1)\lambda\rfloor}\Big)^{(m_i-1)/(m-1)} \prod_{k\neq Y_{n+i}} \Big(\frac{n_k}{n_k+1-\lfloor (n_k+1)\lambda\rfloor}\Big)^{m_k/(m-1)}\\
&\leq m^{-1}\sum_{i\in \range{m}} \kappa^{-1}\Big(\frac{1}{1-\lambda_i}\Big)^{(m_i-1)/(m-1)} \prod_{k\neq Y_{n+i}} \Big(\frac{1}{1-\lambda_k}\Big)^{m_k/(m-1)}\leq 1,
\end{align*}
because $\E (\nu_k^{-1})= \frac{n_k}{n_k+1-\lfloor (n_k+1)\lambda\rfloor}\leq \frac{n_k+1}{n_k+1-\lfloor (n_k+1)\lambda\rfloor}$ and $\E ((\nu'_i)^{-1})= \frac{n_i+1}{n_i+1-\lfloor (n_i+1)\lambda\rfloor}$ by Lemma~\ref{lem:beta} and by the definition \eqref{equkappaClass} of $\kappa$.
Combining the latter with \eqref{firstbigstepAdapt} gives the result.

\subsection{Proof of Theorem~\ref{th:gencontrol}}\label{proof:gen}

\et{Let us prove the result for the iid model (the proof for the conditional model is similar). 
Recall the definition of $(\hat{p}_{i,b})_{i\in \range{m}}=(\hat{p}^{(z^h_{i})}_{i,b})_{i\in \range{m}}$ (not depending on $z^h$ for the iid model, see Algorithm~\ref{alg:general_iid}), for $1\leq b\leq B$ in Algorithm~\ref{alg:general}. Since the scores $S_{Y_{i}}(X_i)$, $i\in [n+m]$, are iid and have no ties, and $p$-values $(p_i)_{i\in \range{m}}=(p^{(Y_{n+i})}_i)_{i\in \range{m}}$ involve only ranks between those scores, we have that  
the variables $(\hat{p}_{i,b})_{i\in \range{m}}$, $1\leq b\leq B$, and $(p_i)_{i\in \range{m}}$ are iid. This means that $\xi_b=\xi^h_b$, $1\leq b\leq B$, and $\xi:=\Cp((p_i)_{i\in \range{m}})$ are iid and thus exchangeable. Letting
\begin{equation}
\hat{q}=(B+1)^{-1}\Big(1+\sum_{b=1}^B \ind{\xi_b\leq \xi}\Big),\label{formulapvaluedistfree}    
\end{equation}
we thus have by \cite{RW2005} that $\P(\hat{q}\leq \alpha)\leq\alpha$.
Now, we have
\begin{align*}
    \P((Y_{n+i})_{i\in \range{m}}\notin \mathcal{C}^m_{t,F})&=\P(\Cp((p^{(Y_{n+i})}_i)_{i\in \range{m}}) < \xi_{(\lfloor (B+1)\alpha)\rfloor)})=\P(\xi < \xi_{(\lfloor (B+1)\alpha)\rfloor)})
    =\P(\hat{q}\leq \alpha)\leq \alpha,
\end{align*}
which concludes the proof.
}
\section{Technical results}

The next result is a variation of results in appendices of \cite{marandon2024adaptive,gazin2024selecting}. 

\begin{lemma}[For full-calibrated $p$-values]\label{lemmaPropConformalFull}
    Let us consider the scores $S_j=S_{Y_j}(X_j)$, $j\in \range{n+m}$, and assume them to be exchangeable and have no ties almost surely. Consider the full-calibrated $p$-values \eqref{standardpvalue} $p_i:=p_i^{(Y_{n+i})} $, $i\in \range{m}$, and let for any fixed $i\in \range{m}$, 
    \begin{align*}
W_i&:=(A_i,(S_{n+j})_{j\in\range{m}\backslash\{i\}});\\
A_i&:=\{S_{j},j\in \range{n}\}\cup \{S_{n+i}\} =: \{a_{i,(1)},\dots,a_{i,(n+1)}\};\\
\Psi_i(u,W_i) &:= \left(\frac{1}{n+1}\Big(\ind{a_{i,(\lceil u(n+1)\rceil )}< S_{n+j}}+\sum_{s\in A_i} \ind{s\geq S_{n+j}} \Big)\right)_{j\in\range{m}\backslash\{i\}},
\end{align*} 
with $a_{i,(1)}>\dots>a_{i,(n+1)}$. Then we have
\begin{itemize}
    \item[(i)] $\mathbf{p}_{-i}:=(p_j)_{j\in\range{m}\backslash\{i\}}$ is equal to $\Psi_i(p_i,W_i)$ and  $u\in [0,1] \mapsto \Psi_i(u,W_i)\in \R^{m-1}$ is a nondecreasing function (in a coordinate-wise sense for the image space);
    \item[(ii)] $(n+1)p_i$ is uniformly distributed on $\range{n+1}$ and independent of $W_i$;
    \item[(iii)] the distribution of $\mathbf{p}_{-i}$ conditionally on $p_i=(n+1)^{-1}$ is the same as if all the scores were all iid $U(0,1)$. In particular, this distribution is equal to a distribution $\mathcal{D}_i$  which is defined as follows: $\mathbf{p}'_{-i}:=(p'_j)_{j\in\range{m}\backslash\{i\}}\sim \mathcal{D}_i$ if, 
    conditionally on the ordered statistics $U_{(1)}>\dots>U_{(n+1)}$ of an iid sample of uniform random variables $(U_1,\dots,U_{n+1})$, the variables 
$(p'_j)_{j\in\range{m}\backslash\{i\}}$ are iid with common cdf $F(x)=(1-U_{(\lfloor (n+1)x\rfloor +1)})\ind{(n+1)^{-1}\leq x<1} + \ind{ x\geq 1}$.
    \item[(iv)] Let $(p'_j)_{j\in \range{m}}$ such that $p'_i=(n+1)^{-1}$ and $p'_j=(n+1)^{-1}\sum_{s\in A_i} \ind{s\geq S_{n+j}} $ for $j\neq i$. Then, $(p'_j)_{j\in \range{m}}$ is $W_i$-measurable and almost surely, for all $j\neq i$, $p'_j\leq p_j$ when $p_j\leq p_i$ and $p'_j= p_j$ when $p_j> p_i$.
\end{itemize}  
\end{lemma}

The next lemma adapts Lemma~\ref{lemmaPropConformalFull} to the class conditional model (with class-calibrated $p$-values). \et{In a nutshell, it says that the previous lemma applies within each class and uses the independence between scores of different classes (conditionally on all the labels). }

\begin{lemma}[For class-calibrated $p$-values]\label{lemmaPropConformalClass}
    Let us consider the scores $S_j=S_{Y_j}(X_j)$, $j\in \range{n+m}$, and assume that for all $k\in \range{K}$, the scores $S_j,$ $j\in \range{n+m}:Y_{j}=k$, are exchangeable\et{, independent of the scores $S_j,$ $j\in \range{n+m}:Y_{j}\neq k$} and have no ties almost surely. Consider the class-calibrated $p$-values \eqref{standardpvalue} $p_i:=p_i^{(Y_{n+i})} $, $i\in \range{m}$, and let for any fixed $i\in \range{m}$, $n_i=|\mathcal{D}^{(Y_{n+i})}_{{\tiny \mbox{cal}}}|$ and
    \begin{align*}
W_i&:=(A_i,(S_{n+j})_{j\in\range{m}\backslash\{i\}}, (S_j)_{j\in\range{n}: Y_{j}\neq Y_{n+i}});\\
A_i&:=\lbrace S_j, j\in \mathcal{D}^{(Y_{n+i})}_{{\tiny \mbox{cal}}}\rbrace\cup \{S_{n+i}\} = \{a_{i,(1)},\dots,a_{i,(n_i+1)}\};\\
\Psi_i(u,W_i) &:= \left(\frac{1}{n_i+1}\Big(\ind{a_{i,(\lceil u(n_i+1)\rceil )}< S_{n+j}}+\sum_{s\in A_i} \ind{s\geq S_{n+j}} \Big)\right)_{j\in\range{m}\backslash\{i\}:Y_{n+j}=Y_{n+i} },
\end{align*} 
with $a_{i,(1)}>\dots>a_{i,(n_i+1)}$. Then we have
\begin{itemize}
    \item[(i)] $(p_j)_{j\in\range{m}\backslash\{i\}:Y_{n+j}=Y_{n+i} }$ is equal to $\Psi_i(p_i,W_i)$ and  $u\in [0,1] \mapsto \Psi_i(u,W_i)$ is a nondecreasing function (in a coordinate-wise sense for the image space);
    \item[(ii)] Conditionally on $(Y_{j})_{j\in \range{n+m}}$, the variable $(n_i+1)p_i$ is uniformly distributed on $\range{n_i+1}$ and independent of $W_i$ and $(p_j)_{j\in\range{m}:Y_{n+j}\neq Y_{n+i} }$;
    \item[(iii)] the distribution of $(p_j)_{j\in\range{m}\backslash\{i\}:Y_{n+j}=Y_{n+i} }$ conditionally on $p_i=(n_i+1)^{-1}$ and $(Y_{j})_{j\in \range{n+m}}$ is the same as if all the scores were all iid $U(0,1)$. In particular, this distribution is equal to a distribution $\et{\mathcal{L}_{i,(Y_{j})_{j\in \range{n+m}}}}$  which is defined as follows: $(p'_j)_{j\in\range{m}\backslash\{i\}:Y_{n+j}=Y_{n+i}}\sim \et{\mathcal{L}_{i,(Y_{j})_{j\in \range{n+m}}}}$ if, 
    conditionally on the ordered statistics $U_{(1)}>\dots>U_{(n_i+1)}$ of an iid sample of uniform random variables $(U_1,\dots,U_{n_i+1})$ (independent of everything else), the variables 
$(p'_j)_{j\in\range{m}\backslash\{i\}:Y_{n+j}=Y_{n+i}}$ are iid with common cdf 
$$F(x)=(1-U_{(\lfloor (n_i+1)x\rfloor +1)})\ind{(n_i+1)^{-1}\leq x<1} + \ind{ x\geq 1}.$$
\item[(iv)] For $k\neq Y_{n+i}$, conditionally on $(Y_{j})_{j\in \range{n+m}}$, the distribution of $(p_j)_{j\in\range{m}:Y_{n+j}=k }$  is the same as if all the scores were all iid $U(0,1)$. In particular, this distribution is equal to a distribution $\et{\mathcal{L}^{(k)}_{(Y_{j})_{j\in \range{n+m}}}}$  which is defined as follows: $(p'_j)_{j\in\range{m}:Y_{n+j}=k}\sim \et{\mathcal{L}^{(k)}_{(Y_{j})_{j\in \range{n+m}}}}$ if, 
    conditionally on the ordered statistics $U^{(k)}_{(1)}>\dots>U^{(k)}_{(n_k)}$ of an iid sample of uniform random variables $(U^{(k)}_1,\dots,U^{(k)}_{n_k})$ (independent of everything else), the  variables 
$(p'_j)_{j\in\range{m}:Y_{n+j}=k}$ are iid with common cdf 
$$F^{(k)}(x)=(1-U^{(k)}_{(\lfloor (n_k+1)x\rfloor )})\ind{(n_k+1)^{-1}\leq x<1} + \ind{ x\geq 1}.$$
    \item[(v)] Let $(p'_j)_{j\in \range{m}:Y_{n+j}=Y_{n+i}}$ such that $p'_i=(n_i+1)^{-1}$ and $p'_j=(n_i+1)^{-1}\sum_{s\in A_i} \ind{s\geq S_{n+j}} $ for \et{$j\neq i$ with $Y_{n+j}=Y_{n+i}$.} Then, $(p'_j)_{j\in \range{m}:Y_{n+j}=Y_{n+i}}$ is $W_i$-measurable and almost surely, for all \et{$j\neq i$ with $Y_{n+j}=Y_{n+i}$,} $p'_j\leq p_j$ when $p_j\leq p_i$ and $p'_j= p_j$ when $p_j> p_i$.
\end{itemize}  
\end{lemma}

\begin{proof}
    Let us prove (i), we have for $j\in\range{m}\backslash\{i\}$ with $Y_{n+j}=Y_{n+i}$,
    \begin{align}
        p_j&= \frac{1}{|\mathcal{D}^{(Y_{n+j})}_{{\tiny \mbox{cal}}}|+1}\Big(1+\sum_{\l\in \mathcal{D}^{(Y_{n+j})}_{{\tiny \mbox{cal}}}} \ind{S_\l\geq S_{n+j}} \Big)\nonumber\\
        &=\frac{1}{|\mathcal{D}^{(Y_{n+i})}_{{\tiny \mbox{cal}}}|+1}\Big(1+\sum_{s\in A_i} \ind{s\geq S_{n+j}} - \ind{S_{n+i}\geq S_{n+j}} \Big),\label{equinterm}
    \end{align}
    which gives the relation because $S_{n+i}=a_{i,( p_i(n_i+1) )}$. Since the monotonicity property is clear, this gives (i).

    Point (ii) comes from the fact that the scores $\lbrace S_j, j\in \mathcal{D}^{(Y_{n+i})}_{{\tiny \mbox{cal}}}\rbrace\cup \{S_{n+i}\}$ have not ties and are exchangeable conditionally on all other scores (and of $(Y_{j})_{j\in \range{n+m}}$).

    For proving (iii), we first note that the calibrated $p$-values are ranks of exchangeable scores with not ties. Hence, the distribution of the $p$-value vector is free from the distribution scores and thus is the same as if the scores were generated as iid $U(0,1)$. \et{Hence, the latter assumption is made for the rest of the proof.} Now, by (i), we have for all $j\in\range{m}\backslash\{i\}$ with $Y_{n+j}=Y_{n+i}$, \et{and if $p_i=(n_i+1)^{-1}$},
    $$
p_j=\frac{1}{n_i+1}\Big(1+\sum_{s\in A_i\backslash\{ a_{i,(1)}\}} \ind{s\geq S_{n+j}} \Big),
    $$
    \et{because $a_{i,(1)}=S_{n+i}$ in that case. Hence, the $p_j$'s, $j\in\range{m}\backslash\{i\}$ with $Y_{n+j}=Y_{n+i}$,}
    
    are iid conditionally on $A_i$ and $(Y_{j})_{j\in \range{n+m}}$. In addition, the common marginal cdf at a point $x$ is given by
    \begin{align*}
\P\Big(1+\sum_{s\in A_i\backslash\{ a_{i,(1)}\}} \ind{s\geq S_{n+j}}\leq x(n_i+1)\Big)&=\P\Big(\sum_{s\in A_i\backslash\{ a_{i,(1)}\}} \ind{s\geq S_{n+j}}< \lceil x(n_i+1) \rceil\Big)\\
&=\P\Big( a_{i,(\lceil x(n_i+1) \rceil+1)}<S_{n+j}\Big),
    \end{align*}
    provided that $1\leq x(n_i+1)< n_i+1$ and the above probabilities being taken conditionally on $A_i$ and $(Y_{j})_{j\in \range{n+m}}$. 
    The result follows because we considered uniformly distributed scores.

    Point (iv) is similar to point (iii), starting directly from the following relation: for all $j\in\range{m}$ with $Y_{n+j}=k$,
    $$
p_j=\frac{1}{n_k+1}\Big(1+\sum_{s\in \{U^{(k)}_{(1)},\dots,U^{(k)}_{{(n_k)}}\}} \ind{s\geq S_{n+j}} \Big),
    $$
where $U^{(k)}_1>\dots>U^{(k)}_{n_k}$ are the ordered elements of $\{S_j,j\in \mathcal{D}^{(k)}_{{\tiny \mbox{cal}}}\}$. 

Finally, let us prove point (v): first $p'_j\leq p_j$ is obvious from \eqref{equinterm}. Second, if  $j\in\range{m}\backslash\{i\}$ with $Y_{n+j}=Y_{n+i}$ is such that $p_j> p_i$, this means $S_{n+j}<S_{n+i}$ and thus $p'_j = p_j$ from \eqref{equinterm}. The result is proved.
\end{proof}

\begin{lemma}[Lemma D.6 of \cite{marandon2024adaptive}]\label{BHsmallerp} 
Write $\wh{\ell}=\wh{\ell}(\mathbf{p})$ for \eqref{BHrej} with any $p$-value family $\mathbf{p}=(p_i)_{i\in \range{m}}$. Fix any $i\in \{1,\dots,m\}$ and consider two collections $\mathbf{p}=(p_i)_{i\in \range{m}}$ and $\mathbf{p}'=(p'_i)_{i\in \range{m}}$ which satisfy almost surely that
\begin{align}\label{propppprime}
\forall j\in \range{m}, \left\{\begin{array}{cc} p'_j\leq p_j& \mbox{ if } p_j\leq p_i;\\p'_j = p_j& \mbox{ if } p_j> p_i.\end{array}\right.  
\end{align}
Then we have almost surely 
$
\{p_i\leq \alpha \wh{\ell}(\mathbf{p})/m \}=\{ p_i\leq \alpha \wh{\ell}(\mathbf{p}')/m \}\subset \{ \wh{\ell}(\mathbf{p})=\wh{\ell}(\mathbf{p}')\}.
$
\end{lemma}

\begin{lemma}\label{lem:beta}
    For $V_{(1)}>\dots>V_{(\ell)}$ the order statistics of $\ell$ iid uniform variables on $[0,1]$, we have for all $a\in \range{\ell}$, $V_{(a)} \sim \beta(\ell+1-a,a)$. In addition, if $a<\ell$, $\E(1/V_{(a)})= \ell/(\ell-a)$.
\end{lemma}

\begin{lemma}[\cite{klenke2010stochastic}]\label{lem:domin}
For $Z_1,\dots,Z_m$ independent Bernoulli variables of respective parameters $\nu_i\in [0,1]$, $i\in \range{m}$, the Poisson binomial variable $\sum_{i\in \range{m}} Z_i$ is stochastically larger than a binomial variable of parameters $m$ and $\nu=\prod_{i\in \range{m}} \nu_i^{1/m}$.
\end{lemma}

\begin{lemma}[Lemma 1 of \cite{BKY2006}]\label{lem:BY}
    If $T$ is a Binomial variable with parameter $m-1\geq 0$ and $\nu\in(0,1]$, we have
  $$
  \E[1/(T+1)]=(1-(1-\nu)^{m})/(m\nu)\leq 1/(m\nu). 
  $$  
\end{lemma}
\et{
The next lemma has been suggested by an anonymous referee.
Recall the definition of PRDS given in \S~\ref{sec:proofSimes}.
\begin{lemma}[Lemma~A.2 of \cite{Bogomolov23}]\label{lem:PRDSgroup}
    Let $(p_i)_{i\in \range{m}}$ be a $p$-value family, $(G_k)_{k\in [K]}$ a partition of $\range{m}$ and assume the following:
    \begin{itemize}
        \item Independence between groups: for $k\neq k'$, $(p_i)_{i\in G_k}$ is independent of $(p_i)_{i\in G_{k'}}$;
        \item PRDS inside each group: for $k\in \range{m}$, $(p_i)_{i\in G_k}$ is a $p$-value family which is PRDS on $G_k$.  
    \end{itemize}
    Then the $p$-value family $(p_i)_{i\in \range{m}}$ is PRDS on $\range{m}$.
\end{lemma}
We provide a proof for completeness.
\begin{proof}
  Fix $i\in \range{m}$ and a nondecreasing set $D\subset [0,1]^m$, and prove that the function
$
u\mapsto \P((p_i)_{i\in \range{m}}\in D\:|\: p_i=u)
$
is nondecreasing. Denote $k_i$ the unique $k$ such that $i\in G_k$ and let
$$
D_{k_i,(p_j)_{j\notin G_{k_i}}}=\{(p_j)_{j\in G_{k_i}}\in [0,1]^{G_{k_i}}\::\: (p_j)_{j\in \range{m}}\in D\}
$$
which is clearly a nondecreasing (measurable) set of $[0,1]^{G_{k_i}}$. By using the two assumptions, we have that $(p_i)_{i\in G_{k_i}}$ is PRDS on $G_k$ conditionally on $(p_j)_{j\notin G_{k_i}}$. Hence, 
$$
u\mapsto \P((p_j)_{j\in G_{k_i}}\in D_{k_i,(p_j)_{j\notin G_{k_i}}} \:|\: p_i=u, (p_j)_{j\notin G_{k_i}})
$$
is nondecreasing. We obtain the result by integrating with respect to $(p_j)_{j\notin G_{k_i}}$.
\end{proof}
}
\section{Computational shortcut for the combinations of conformal $p$-values method}\label{sec:shortcut_supplementary}

Computing the batch prediction set for our methods is in general of complexity of order $K^m$ times the cost of computing the combining function (e.g., order $m$ for Fisher, or $m\log m$ for Simes or adaptive Simes)\footnote{In general, the cost of computing the $p$-value family $(p^{(k)}_i, k\in [K], i\in [m])$ is negligible wrt $K^m$.}.  The aim of this section is to reduce this complexity when the user only want to report lower/upper bounds for $m_k(Y)$, $k\in \range{K}$ \eqref{countk}. We also discuss the issue of reconstructing the batch prediction set from these bounds.

\subsection{Shortcut for computing the bounds}

Naively computing the bounds $[\ell_{\alpha}^{(k)}, u_{\alpha}^{(k)}]$, $k \in [K]$, in  (\ref{lbub}), which are derived from the Simes conformal prediction set in  (\ref{PredSimes}) or its adaptive version in  (\ref{PredSimesAdapt}), results in an exponential complexity of $O(K^m)$.
 This quickly becomes impractical for large batch sizes. To address this issue, we introduce a novel shortcut that allows for a more efficient computation of these bounds, with a computational complexity of at most $O(K \times m^2)$.

This shortcut applies to both the full-calibrated and class-calibrated conformal $p$-values. Proposition \ref{prop:shortcut} shows that it is exact when $K=2$ and the scores produced by the machine learning model are probabilities. However, when $K>2$ or when arbitrary scores are used, the shortcut becomes conservative, potentially yielding wider bounds but never narrower ones.
This ensures that the coverage guarantee of at least $1-\alpha$ probability is maintained.

Algorithm \ref{alg:shortcut} provides the pseudocode for the shortcut to compute the bounds $[\ell_{\alpha}^{(k)}, u_{\alpha}^{(k)}]$ derived from the (adaptive) Simes conformal prediction set.

\begin{algorithm}[!htb]
\SetKwInOut{Input}{Input}
\Input{Full-calibrated or class-calibrated conformal $p$-values $(p_i^{(k)})_{i \in [m], k \in [K]}$, level $\alpha \in (0,1)$, an estimator $\hat{m}_0(\mathbf{p})$ that is monotone in the $p$-values $\mathbf{p}=(p_i)_{i \in [m]}$.}

\For{each $k\in [K]$}{

Sort $(p_i^{(k)})_{i \in [m]}$ in decreasing order and store as $a_1 \geq \ldots \geq a_m$;

Sort $(\max\{p^{(j)}_i, j\neq k\})_{i \in [m]}$ in decreasing order and store as $b_1 \geq \ldots \geq b_m$;

\For{each $v\in\{m, \ldots,0\}$}{

  $(q_1,\ldots,q_m) \gets (a_1,\ldots,a_{v},b_1,\ldots,b_{m-v})$;

  Sort $(q_i)_{i \in [m]}$ in increasing order and store as $q_{(1)} \leq \ldots \leq q_{(m)}$;

  $\displaystyle h_{v,k} \gets \min\Big( \frac{\hat{m}_0(\mathbf{q})}{\ell} q_{(\ell)}, \ell \in [m] \Big)$

}  

$\ell_{\alpha}^{(k)} \gets \min(v\in \{0,\ldots,m\}: h_{v,k} > \alpha)$;

$u_{\alpha}^{(k)} \gets \max(v\in \{0,\ldots,m\}: h_{v,k} > \alpha)$;

}

\SetKwInOut{Output}{Output}

\Output{ $[\ell_{\alpha}^{(k)},u_{\alpha}^{(k)}]$, $k \in [K]$ }

\caption{Shortcut for computing the bounds $[\ell_{\alpha}^{(k)},u_{\alpha}^{(k)}]$, $k \in [K]$, with (adaptive) Simes predition set.} \label{alg:shortcut}
\end{algorithm}

\begin{proposition}\label{prop:shortcut}
For any $\alpha \in (0,1)$, let $[\ell_{\alpha}^{(k)}, u_{\alpha}^{(k)}]$, $k \in [K]$ be the bounds defined by  (\ref{lbub}), derived from the Simes prediction sets in  (\ref{PredSimes}) or its adaptive version in  (\ref{PredSimesAdapt}). Algorithm \ref{alg:shortcut} returns the bounds $[\tilde\ell_{\alpha}^{(k)}, \tilde u_{\alpha}^{(k)}]$ such that $\tilde \ell_{\alpha}^{(k)} \leq \ell_{\alpha}^{(k)}$ and $\tilde u_{\alpha}^{(k)} \geq u_{\alpha}^{(k)}$ for all $k \in [K]$, with a computational complexity of at most $O(K \times m^2)$. In addition, when $K = 2$ and the scores produced by the machine learning model are probabilities, i.e., $S_k(x_{n+i}) = 1 - S_{3-k}(x_{n+i})$ for $k \in \{1, 2\}$ and $i \in [m]$, it holds that $\tilde \ell_{\alpha}^{(k)} = \ell_{\alpha}^{(k)}$ and $\tilde u_{\alpha}^{(k)} = u_{\alpha}^{(k)}$ for all $k \in [K]$.
\end{proposition}

\begin{proof}
First, let us establish that the time complexity of the algorithm is
$O(K\times m^2)$.
To produce the sorted concatenation of two sorted vectors $a_1,\ldots,a_{m-i}$ and $b_1,\ldots,b_i$ takes linear time, i.e. $O(m)$. This merging process, which generates the sorted concatenation, is repeated $m+1$ times for each $k$. 
As a result, for each $k$, this step 
contributes $O(m^2)$, leading to an overall complexity of $O(K\times m^2)$.

We first discuss the case where $\hat{m}_0 = m$, meaning the estimator is the constant $m$.
Let $\mathbf{p}=(p_i)_{i\in [m]}$ denotes a vector of $p$-values, with the sorted values represented as $p_{(1)} \leq \ldots \leq p_{(m)}$. 
Simes' test is defined as $\displaystyle F_{\mbox{\tiny Simes}}(\mathbf{p}) = \min\Big( \frac{m}{\ell} p_{(\ell)}, \ell \in [m] \Big)$. 
This test is monotonic, meaning that if
$\mathbf{p}\leq \mathbf{q}$ componentwise (i.e.
$p_{(i)} \leq q_{(i)}$ for all $i \in [m]$), then 
$F_{\mbox{\tiny Simes}}(\mathbf{p}) \leq F_{\mbox{\tiny Simes}}(\mathbf{q})$.

By definition, $v \notin \mathcal{N}_k(\mathcal{C}^m_{\alpha,\mbox{\tiny Simes} })$ if $F_{\mbox{\tiny Simes}}(\mathbf{p}(y))  \leq \alpha$ for all $y \in [K]^m$ such that $m_k(y)=v$, for any $v \in \{0,\ldots,m\}$.

Then, for some $\mathbf{q}=(q_i)_{i\in [m]}$ with $\mathbf{q}\geq \mathbf{p}(y)$ for all $y \in [K]^m$ such that $m_k(y)=v$,  $F_{\mbox{\tiny Simes}}(\mathbf{q})  \leq \alpha$ implies $v \notin \mathcal{N}_k(\mathcal{C}^m_{\alpha,\mbox{\tiny Simes} })$. However, $F_{\mbox{\tiny Simes}}(\mathbf{q})  > \alpha$ does not necessarily imply $v \in \mathcal{N}_k(\mathcal{C}^m_{\alpha,\mbox{\tiny Simes} })$.

Given $k$ and $v$, Algorithm \ref{alg:shortcut} identifies a suitable vector $\mathbf{q}=\mathbf{q}_{v,k}$ such that $\mathbf{q}\geq \mathbf{p}(y)$ for all $y \in [K]^m$ where $m_k(y)=v$.  
Then we let
$$
\tilde{\mathcal{N}}_k=\{v \in \{0,\ldots,m\}\::\: F_{\mbox{\tiny Simes}}(\mathbf{q}_{v,k}) > \alpha\},
$$
which ensures $\tilde{\mathcal{N}}_k \supseteq \mathcal{N}_k(\mathcal{C}^m_{\alpha,\mbox{\tiny Simes} })$
The resulting bounds are given by 
$[\tilde{\ell}_{\alpha}^{(k)},\tilde{u}_{\alpha}^{(k)}]=[\min \tilde{\mathcal{N}}_k, \max \tilde{\mathcal{N}}_k]$, which guarantees that $\tilde{\ell}_{\alpha}^{(k)} \leq \ell_{\alpha}^{(k)}$ and $\tilde{u}_{\alpha}^{(k)} \geq u_{\alpha}^{(k)}$ for every $k \in [K]$.

We now need to demonstrate that  Algorithm \ref{alg:shortcut} produces a vector $\mathbf{q}$ such that $\mathbf{q}\geq \mathbf{p}(y)$ for all $y \in [K]^m$ such that $m_k(y)=v$. 

For any $y \in [K]^m$ such that $m_k(y) = v$, the vector $\mathbf{p}(y)$ consists of $v$ conformal $p$-values $p_{i_1}^{(k)}, \ldots, p_{i_v}^{(k)}$ 
and $m-v$ conformal $p$-values $p_{i_{v+1}}^{(j_1)}, \ldots, p_{i_{m}}^{(j_{m-v})}$, where $i_1,\ldots, i_m$ is a permutation of $[m]$ and $j_1, \ldots, j_{m-v} \in [K]\setminus \{k\}$. If we consider the vector $\mathbf{p}(\tilde{y})$, which is formed by $p_{i_1}^{(k)}, \ldots, p_{i_v}^{(k)}$ and the maximum values $\max(p_{i_{v+1}}^{(j)}, j\neq k), \ldots, \max(p_{i_{v+m}}^{(j)}, j\neq k)$, we can conclude that $\mathbf{p}(\tilde{y}) \geq \mathbf{p}(y)$. Since the vector $\mathbf{q}$ in Algorithm \ref{alg:shortcut} is constructed using the largest $v$ values from $(p^{(k)}_i)_{i \in [m]}$ and the largest $m-v$ values from $(\max(p_{i}^{(j)}, j\neq k))_{i \in [m]}$, it follows that $\mathbf{q} \geq \mathbf{p}(\tilde{y}) \geq \mathbf{p}(y)$ for all $y \in [K]^m$ such that $m_k(y) = v$. This establishes the conservativeness of the shortcut for $K\geq 2$ and for any scores produced by the machine learning model.

If $K=2$ and the scores produced by the machine learning model are probabilities, then we have the relationship $S_k(x_{n+i}) = 1-S_{3-k}(x_{n+i})$ for $k \in \{1,2\}$ and $i \in [m]$. 
Given this relationship, there exists a permutation $i_1,\ldots,i_m$ such that the sequence  
$S_k(x_{n+{i_{j_1}}}) \leq \ldots \leq S_k(x_{n+{i_{j_m}}})$ is nondecreasing, while the sequence $S_{3-k}(x_{n+{i_{j_1}}}) \geq \ldots \geq S_{3-k}(x_{n+{i_{j_m}}})$ is nonincreasing. Consequently, the ranks of $S_k(x_{n+{j_1}}), \ldots, S_k(x_{n+{j_m}})$ within the set $(S_{y_j}(x_j))_{j\in \mathcal{D}^{(k)}_{{\tiny \mbox{cal}}} }$ will be nondecreasing, while the ranks of $S_{3-k}(x_{n+{j_1}}), \ldots, S_{3-k}(x_{n+{j_m}})$ within the set $(S_{y_j}(x_j))_{j\in \mathcal{D}^{(3-k)}_{{\tiny \mbox{cal}}} }$ will be nonincreasing. 
Since these ranks are proportional to the conformal $p$-values, it follows that $p_{i_1}^{(k)} \leq \ldots \leq p_{i_m}^{(k)}$ and  $p_{i_1}^{(3-k)} \geq \ldots \geq p_{i_m}^{(3-k)}$. 

Consider $y \in [K]^m$ such that $m_k(y) = v$. Let the vector $\mathbf{p}(y^*)$ consist of the $v$  largest values from ($p_{i}^{(k)})_{i \in [m]}$, specifically $p_{i_{m-v+1}}^{(k)}, \ldots, p_{i_{m}}^{(k)}$. 
Consequently, the remaining 
 $m-v$ values in $\mathbf{p}(y^*)$ are $p_{i_{1}}^{(3-k)}, \ldots, p_{i_{m-v}}^{(3-k)}$, i.e. 
the largest
 $m-v$ values from ($p_{i}^{(3-k)})_{i \in [m]}$. Thus, we have $\mathbf{p}(y^*) \geq \mathbf{p}(y)$ for all $y \in [K]^m$ such that $m_k(y) = v$. Furthermore, by construction, $\mathbf{q}$ in Algorithm \ref{alg:shortcut} is equal to $\mathbf{p}(y^*)$. Therefore $F_{\mbox{\tiny Simes}}(\mathbf{q})  \leq \alpha$ if and only if  $F_{\mbox{\tiny Simes}}(\mathbf{p}(y))  \leq \alpha$ for all $y \in [K]^m$ such that $m_k(y)=v$.  This establishes the exactness of the shortcut when $K=2$ and $S_k(x_{n+i}) = 1-S_{3-k}(x_{n+i})$ for $k \in \{1,2\}$ and $i \in [m]$.

The validity of the shortcut for the adaptive version of Simes follows from the required monotonicity of the estimator: if
$\mathbf{p}(y)\leq \mathbf{q}$, then 
$\hat{m}_0(\mathbf{p}(y)) \leq \hat{m}_0(\mathbf{q})$ holds for any $y \in [K]^m$. This, combined with
$F_{\mbox{\tiny A-Simes}}(\mathbf{p}(y)) \leq \alpha$ if and only if 
$F_{\mbox{\tiny Simes}}(\mathbf{p}(y)) \leq m \alpha / \hat{m}_0(\mathbf{p}(y))$ yields the desired result.
\end{proof}

\subsection{Extension to other combining functions}

Algorithm \ref{alg:shortcut2} presents a more general approach for any $p$-value vector combining function \( F(\mathbf{p}) \), which is symmetric and monotone in the $p$-values \( \mathbf{p} = (p_i)_{i \in [m]} \). It requires the empirical threshold $
t = \xi_{(\lfloor (B+1)\alpha \rfloor)}
$
from Theorem \ref{th:gencontrol}, which depends on \( (m_k)_{k \in \range{K}} \) in the conditional model, i.e.
$t = t(\alpha, (m_k)_{k \in [K]})$.
The proof that Algorithm \ref{alg:shortcut2} yields conservative yet valid bounds is analogous to the previous result and is therefore omitted.

\begin{algorithm}[h!]
\SetKwInOut{Input}{Input}
\Input{Full-calibrated or class-calibrated conformal $p$-values $(p_i^{(k)})_{i \in [m], k \in [K]}$, level $\alpha \in (0,1)$,
$p$-value vector combining function $F(\mathbf{p})$ that is symmetric and monotone in the $p$-values $\mathbf{p}=(p_i)_{i \in [m]}$ and the corresponding critical value $t=t(\alpha, (m_k)_{k \in [K]})$.}

\For{each $k\in [K]$}{

 Sort $(p_i^{(k)})_{i \in [m]}$ in decreasing order and store as $a_1 \geq \ldots \geq a_m$;

 Sort $(\max\{p^{(j)}_i, j\neq k\})_{i \in [m]}$ in decreasing order and store as $b_1 \geq \ldots \geq b_m$;

 \For{each $v\in\{m, \ldots,0\}$}{

   $(q_1,\ldots,q_m) \gets (a_1,\ldots,a_{v},b_1,\ldots,b_{m-v})$;

   Sort $(q_i)_{i \in [m]}$ in increasing order and store as $q_{(1)} \leq \ldots \leq q_{(m)}$;

   $\displaystyle h_{v,k} \gets \ind{ F(\mathbf{q}) \geq \min\{t(\alpha, m_k = v, m_j), j \neq k  \} }$

}  

 $\ell_{\alpha}^{(k)} \gets \min(v\in \{0,\ldots,m\}: h_{v,k} > 0)$;

 $u_{\alpha}^{(k)} \gets \max(v\in \{0,\ldots,m\}: h_{v,k} > 0 \}$;

}

\SetKwInOut{Output}{Output}

\Output{ $[\ell_{\alpha}^{(k)},u_{\alpha}^{(k)}]$, $k \in [K]$ }

\caption{General shortcut for computing the bounds $[\ell_{\alpha}^{(k)},u_{\alpha}^{(k)}]$, $k \in [K]$.} \label{alg:shortcut2}
\end{algorithm}

\subsection{Batch prediction set reconstruction from the bounds}

As described in the previous subsections, from the bounds $[\ell_{\alpha}^{(k)},u_{\alpha}^{(k)}]$, $k \in [K]$, it is straightforward to produce a conservative batch prediction set $\mathcal{ \tilde C}^m_{\alpha}$ such that $\mathcal{ \tilde C}^m_{\alpha} \supseteq \mathcal{  C}^m_{\alpha}$. The cardinality of the conservative set  $\mathcal{\tilde C}^m_{\alpha}$ is the sum of all valid assignments of $( m_1,  \dots, m_K )$ occurrences, where $ \ell^{(k)}_\alpha \leq m_k \leq u^{(k)}_\alpha$ for each $ k \in \{1, \dots, K\} $, and \( m_1 +  \dots + m_K = m \), with each valid assignment counted by the multinomial coefficient \( \binom{m}{m_1, m_2, \dots, m_K} \):
$$|\mathcal{ \tilde C}^m_{\alpha}| = \sum_{\substack{ (m_1,\ldots,m_K)\, :\, \sum_{k=1}^{K} m_k=m ,\\ \ell^{(k)}_\alpha \leq m_k \leq u^{(k)}_\alpha \, \forall k \in [K]}}
 \binom{m}{m_1, m_2, \dots, m_K}.$$

For the reading zip code example, from Table \ref{tabUSPS}, we derive the bounds $[\ell^{(k)}_\alpha , u^{(k)}_\alpha]$ with $\alpha = 0.05$, which are as follows: 
\[
[1,2],\ [0,0],\ [0,0],\ [0,0],\ [1,1],\ [0,2],\ [0,2],\ [0,0],\ [0,1],\ [0,0] \quad \text{for} \quad k = 1, \dots, 10.
\]
The assignments $(m_1, \ldots, m_{10})$ that satisfy $m_1 + \ldots + m_{10} = 5$ and $\ell^{(k)}_\alpha \leq m_k \leq u^{(k)}_\alpha$ for each $k \in \{1, \dots, 10\}$ are ten:
\[
\begin{aligned}
(1, 0, 0, 0, 1, 0, 2, 0, 1, 0), & \quad (1, 0, 0, 0, 1, 1, 1, 0, 1, 0), & \quad (1, 0, 0, 0, 1, 1, 2, 0, 0, 0), \\
(1, 0, 0, 0, 1, 2, 0, 0, 1, 0), & \quad (1, 0, 0, 0, 1, 2, 1, 0, 0, 0), & \quad (2, 0, 0, 0, 1, 0, 1, 0, 1, 0), \\
(2, 0, 0, 0, 1, 0, 2, 0, 0, 0), & \quad (2, 0, 0, 0, 1, 1, 0, 0, 1, 0), & \quad (2, 0, 0, 0, 1, 1, 1, 0, 0, 0), \\
(2, 0, 0, 0, 1, 2, 0, 0, 0, 0). &
\end{aligned}
\]

The corresponding multinomial coefficients are 60, 120,  60,  60,  60,  60,  30,  60,  60 and 30, respectively. 
This results in a cardinality of the conservative set $|\tilde{\mathcal{C}}^m_{\alpha,\text{\tiny Simes}}| = 600$, compared to $|\mathcal{C}^m_{\alpha,\text{\tiny Simes}}| = 6$ given in Table \ref{tabUSPS}. 
This indicates that reconstructing the prediction set solely from the bounds is quite imprecise. For instance, the assignment $(2, 0, 0, 0, 1, 2, 0, 0, 0, 0)$ corresponds to $\binom{5}{2, 0, 0, 0, 1, 2, 0, 0, 0, 0}  =  30$ vectors of size 5, which include two 0s, one 4, and two 5s. 

While $\tilde{\mathcal{C}}^m_{\alpha}$ is not accurate in general, we can combine this information with individual conformal prediction sets $\mathcal{C}^m_{i,\alpha}$, $i\in \range{m}$ to allows for a more accurate batch prediction set reconstructed from the bounds. For this, specific shortcuts can be investigated to compute the individual conformal prediction sets $\mathcal{C}^m_{i,\alpha}$, $i\in \range{m}$. More specifically, for Simes' method, we can always use the Bonferroni individual prediction set to obtain a new batch prediction set from the bounds {\it both with low complexity that can only improve over $\mathcal{C}^m_{\alpha,\text{\tiny Bonf}}$}. 
In addition, the following example shows that this improvement can be strict: we see this as an important `proof of concept'.

For the example of one batch of the CIFAR dataset given in Figure \ref{animals} with $m=10$, $K=3$, and $\alpha=0.1$, the Bonferroni individual conformal prediction sets $\mathcal{C}^m_{i,\alpha}$ are $\{3\}$ for $i = 8$ and $\{1,2,3\}$ for $i =1,2,3,4,5,6,7,9,10$. On the other hand, the Simes bounds $[\ell^{(k)}_\alpha , u^{(k)}_\alpha]$ are $[0,8]$, $[0,9]$, and $[1,10]$ for $k=1,2,3$, which improve upon Bonferroni's $[0,9]$, $[0,9]$, and $[1,10]$. Consequently, the vector $(1, 1, 1, 1, 1, 1, 1, 3, 1, 1)$ must be excluded from $\mathcal{C}^m_{\alpha,\text{\tiny Bonf}}$ because it violates the constraint that the number of 1s must not exceed 8.

\subsection{Simulation results for large batches of test points}\label{subsec-sim-large-batches}
In order to demonstrate the feasibility and usefulness of the shortcut, we carried out simulations with test samples of size $m\in \{200, 2000\}$. Specifically, we considered the Gaussian multivariate setting described in \S~\ref{subsec-BVN} with classes one and two only. The calibration sample has an equal number of examples from each of the two classes.  

Table \ref{tabBVN2classeslowerbounds}  shows results for the case that the test sample has an equal number of examples from each of the two classes. Since the calibration set in each class is 400 examples, the smallest possible class conditional conformal $p$-value is 1/401, so the Bonferroni adjusted $p$-value is at least $m\cdot 1/401$ and the lower bounds are zero. However, Simes and modified Simes have informative lower bounds, and the tightness of the lower bounds increases with the SNR. The computational complexity is very reasonable, running in less than 0.01 seconds for $m=200$, and in 0.20-0.22 seconds for $m=2000$.

 \begin{table}[ht]
 \centering
 \begin{tabular}{|r|rrr|rrr|}
   \hline
   &\multicolumn{3}{|c|}{$m=200$}&\multicolumn{3}{|c|}{$m=2000$}\\ 
     &  & & Storey- &  &  & Storey- \\ 
   SNR & Bonf & Simes & Simes & Bonf & Simes &  Simes  \\ 
   \hline
  1.00 & 0.00 & 1.12 & 2.13 & 0.00 & 9 & 20  \\ 
  1.50 & 0.00 & 6.98 & 10.56 & 0.00 & 62 & 103 \\ 
 2.00 & 0.00 & 23.31 & 26.70 & 0.00 & 212 & 253  \\ 
 2.50 & 0.00 & 40.73 & 41.85 & 0.00 & 413 & 428  \\ 
 3.00 & 0.00 & 58.02 & 57.60 & 0.00 & 729 & 730  \\ 
 3.50 & 0.00 & 73.28 & 72.85 & 0.00 & 830 & 829  \\ 
 4.00 & 0.00 & 84.19 & 83.83 & 0.00 & 830 & 829  \\ 
 4.50 & 0.00 & 90.75 & 90.22 & 0.00 & 897 & 896  \\ 
    \hline
 \end{tabular}
 \caption{The average lower bound for class one at each SNR, for $m=200$ (columns 2--4) and for $m=2000$ (columns 5--7), at level $\alpha = 0.1$, for the following $p$-value combining functions: Bonferroni, Simes, and adaptive Simes using Storey's estimator  (see detailed data generation in text). The fraction of test sample examples from class one is half.  Based on 100 simulations. \label{tabBVN2classeslowerbounds}}
 \end{table}

Table \ref{tabBVN2classeslowerbounds2} provides  results when the distribution of the classes is uneven in the test sample. As in the setting of Table \ref{tabBVN2classeslowerbounds}, Bonferroni's combination method provides only trivial lower bounds so it is not shown. We also omit adaptive Simes since the performance is very similar to that of using Simes combining function. With Simes combining function, we see that as the fraction of test samples from class two increases or the signal strengthens, the probability of detecting that there are examples from class two increases and the expected lower bound increases. The detection of a lower bound being positive is important in many application. For example, in ecology, this is proof that an animal population is not extinct in an area. In medicine, the detection of evidence that a treatment can be positive (class one) in some patients  but negative (class two) in other patients suggests a qualitative interaction that can prompt further investigation.

 \begin{table}[ht]
 \centering
 \begin{tabular}{|c|c|cc|cc|}
   \hline
  & Fraction in test sample   &\multicolumn{2}{|c|}{Average lower bound}&\multicolumn{2}{|c|}{Probability of non-zero lower bound}\\ 
  SNR&  from class one  &  class 1& class 2 & class 1 & class 2 \\ 
   \hline
2   & 1 &  61.9&   0.004& 1.000 & 0.002   \\ 
   & 0.95&  57.9& 0.258 & 1.000 & 0.1334   \\ 
      & 0.9 & 54.0  & 1.225 &  1.000& 0.4242   \\ 
            & 0.7 & 37.39  & 10.12 & 1.000 & 0.9909   \\ 
3   & 1 & 133.6 & 0.008 &  1.000&    0.004\\ 
   & 0.95&  126.3&  2.028&  1.000&  0.680  \\ 
      & 0.9 &  118.3&7.242  &  1.000&   0.9756 \\ 
            & 0.7 &  88.05& 31.99 & 1.000 &  1.000  \\ 
4   & 1 &  175.7&  0.006&  1.000&    0.003\\ 
   & 0.95& 166.3 & 4.828 &  1.000 &  0.9832  \\ 
      & 0.9 & 157.0 & 13.24 & 1.000 &   0.9999 \\ 
            & 0.7 &  119.9&  47.88& 1.000 & 1.000   \\ 
 
    \hline
 \end{tabular}
 \caption{The average lower bound (columns 2 and 3) and probability that the lower bound is non-trivial (columns 4 and 5) using the Simes combination function, for each class at SNR=3, for $m=200$, at level $\alpha = 0.1$. For each SNR, each row has a different relative frequency of the number of examples from class one (in the calibration set, half the examples are from class one, see text for details). Based on 10000 simulations. \label{tabBVN2classeslowerbounds2}}
 \end{table}

\section{General $p$-value combining prediction set algorithm for the iid model}\label{appendix-generalcombinationsalgorithm-iid}
In \S~\ref{sec:numapprox} we provided the most general method of obtaining $1-\alpha$ level prediction sets using combinations of conformal $p$-values. Algorithm \ref{alg:general} shows the construction for the class conditional model. Its computational complexity is $B$ times the number of unique frequency distributions of $\range{K}^m$ vectors. For completeness, we provide in Algorithm \ref{alg:general_iid} an algorithm for the iid model, which requires only $B$ permutations. 

\begin{algorithm}[!htb]
\small
\SetKwInOut{Input}{Input}
\Input{Number of examples in the calibration set $n$;  
       combining function $F$;  
       level $\alpha \in (0,1)$;  
       number of permutations $B$;  
       conformal $p$-values $(p^{(y_{i})}_i)_{i\in \range{m}}$.
}

    \For{each $b\in [B]$}{
        
        Generate a random permutation $\pi_b$ of $[n+m]$;
        
        Compute null conformal $p$-values:
        \[
        \hat{p}_{i,b} \gets \frac{1+\sum_{j = 1}^{n} \ind{\pi_b(j) \geq \pi_b(n+i)} }
        {n+1}
        \]
        for $i \in [m]$;

        Compute combined statistic:\\
        $\xi_b \gets \Cp((\hat{p}_{i,b},i\in \range{m}))$;
    }

    Compute threshold:\\
    $t \gets \xi_{ ( \lfloor (B+1)\alpha \rfloor ) }$,\\ where $\xi_{(1)} \leq \ldots \leq \xi_{(B)}$ are the ordered test statistics and $\xi_{(0)}=-\infty$;

Construct batch prediction set:\\
$\mathcal{C}^m_{t,F}  \gets \{ y = (y_{i})_{i\in [m]} \in \range{K}^m\::\:  \Cp((p^{(y_{i})}_i)_{i\in \range{m}}) \geq t \}$;

\SetKwInOut{Output}{Output}
\Output{Batch prediction set $\mathcal{C}^{m}_{t,F}$.}

\caption{Constructing a $1-\alpha$ level batch prediction set, using combinations of conformal $p$-values, for the iid model}
\label{alg:general_iid}
\end{algorithm}

\section{Additional numerical experiments}\label{appendix-sim}

\subsection{Gaussian multivariate setting}\label{appendix-sim-Gaussian}

We provide more results for the data generation described in \S~\ref{subsec-BVN}. Figure \ref{BVNfig} shows the data available in one data generation. Table \ref{tabBVNlist} shows the batch prediction set for this batch using Bonferroni at $\alpha = 0.1$, as well as the Bonferroni and Simes $p$-values for each $y$ in the batch. Had the analyst used Simes instead of Bonferroni at $\alpha = 0.1$, the batch prediction set size would have been 25\% smaller. 

\begin{figure}[h!]
\begin{center}
\includegraphics[scale = 0.4]{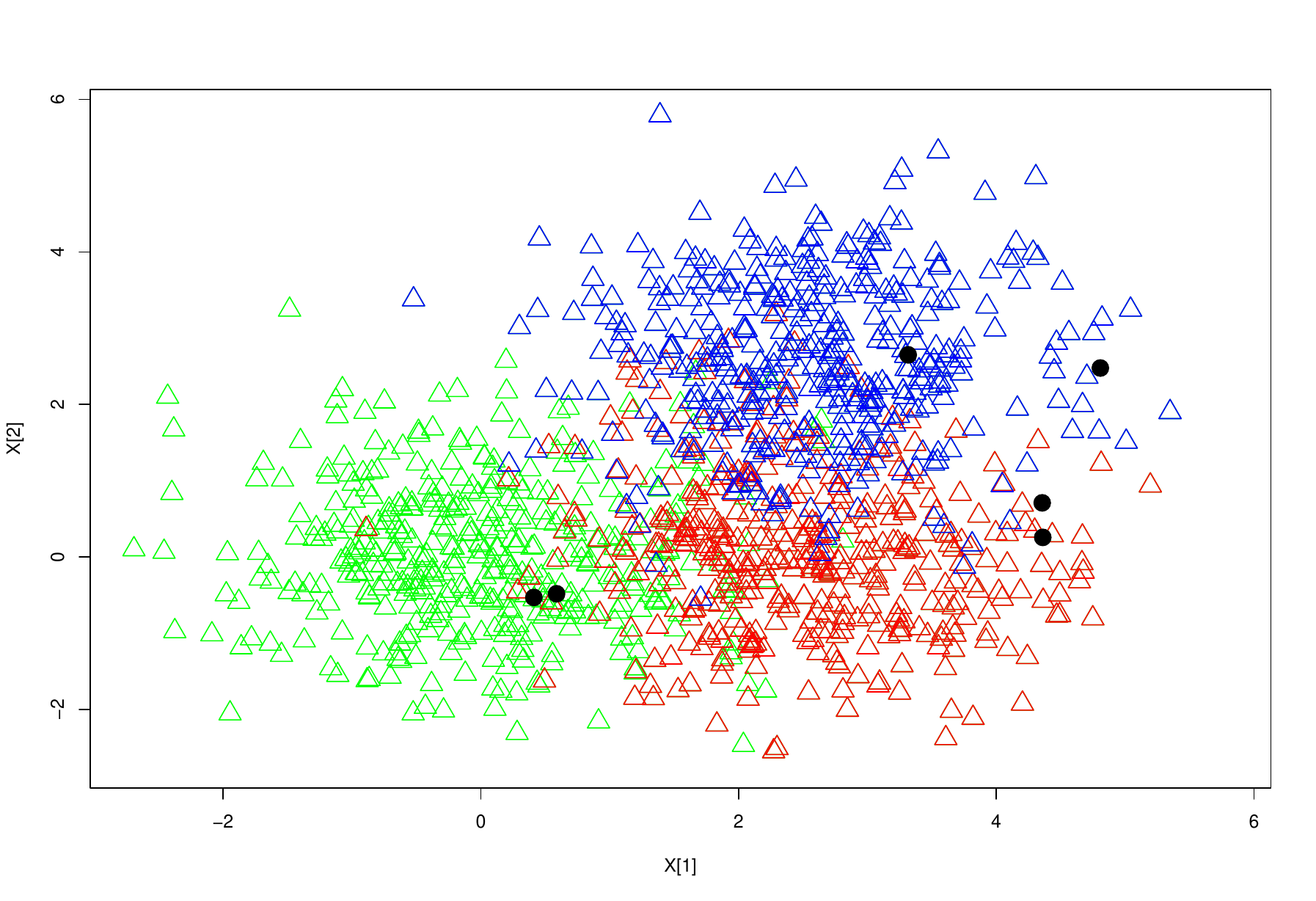}
\end{center}
\vspace{-5mm}
\caption{\label{BVNfig} Illustration of one data generation with SNR = 2.5. The batch of six test samples are in black. There are 400 calibration examples from each class (class one in green, class two in red, and class three in blue).  At $\alpha = 0.1$, the size of the prediction set using Bonferroni and Simes is 32 and 24, respectively.}
\end{figure}

\begin{table}[ht]
\centering
\begin{tabular}{|rrrrrr|rr|}
  \hline
  $Y_1=1$ & $Y_2=1$ & $Y_3=2$ & $Y_4=2$ & $Y_5=3$ & $Y_6=3$ &Bonf & Simes \\ 
  \hline
 1 & 1 & 2 & 3 & 2 & 3 & 0.12 & 0.07 \\ 
   2 & 1 & 2 & 3 & 2 & 3 & 0.12 & 0.07 \\ 
   1 & 2 & 2& 3 & 2 & 3 & 0.12 & 0.07 \\ 
   2 & 2 & 2 & 3 & 2 & 3 & 0.12 & 0.07 \\ 
   1 & 1 & 3 & 3 & 2 & 3 & 0.12 & 0.07 \\ 
   2 & 1 & 3 & 3 & 2 & 3 & 0.12 & 0.07 \\ 
   1 & 2 & 3 & 3 & 2 & 3 & 0.12 & 0.07 \\ 
   2 & 2 & 3 & 3 & 2 & 3 & 0.12 & 0.07 \\ 
   1 & 1  & 2  & 2  & 2  & 3  & 0.12 & 0.12 \\ 
   2  & 1  & 2  & 2  & 2  & 3  & 0.12 & 0.12 \\ 
   1  & 2  & 2  & 2  & 2  & 3  & 0.12 & 0.12 \\ 
   2  & 2  & 2  & 2  & 2  & 3  & 0.12 & 0.12 \\ 
   1  & 1  & 3  & 2  & 2  & 3  & 0.12 & 0.12 \\ 
   2  & 1  & 3  & 2  & 2  & 3  & 0.12 & 0.12 \\ 
   1  & 2  & 3  & 2  & 2  & 3  & 0.12 & 0.12 \\ 
   2  & 2  & 3  & 2  & 2  & 3  & 0.12 & 0.12 \\ 
   2  & 2  & 3  & 3  & 3  & 3  & 0.15 & 0.12 \\ 
   1  & 2  & 3  & 3  & 3  & 3  & 0.15 & 0.12 \\ 
   1  & 1  & 2  & 3  & 3  & 3  & 0.15 & 0.15 \\ 
   2  & 1  & 2  & 3  & 3  & 3  & 0.15 & 0.15 \\ 
   1  & 2  & 2  & 3  & 3  & 3  & 0.15 & 0.15 \\ 
   2  & 2  & 2  & 3  & 3  & 3  & 0.15 & 0.15 \\ 
   1  & 1  & 3  & 3  & 3  & 3  & 0.15 & 0.15 \\ 
   2  & 1  & 3  & 3  & 3  & 3  & 0.15 & 0.15 \\ 
   2  & 2  & 3  & 2  & 3  & 3  & 0.33 & 0.16 \\ 
   1  & 2  & 3  & 2  & 3  & 3  & 0.33 & 0.19 \\ 
   2  & 2  & 2  & 2  & 3  & 3  & 0.33 & 0.24 \\ 
   2  & 1  & 3  & 2  & 3  & 3  & 0.37 & 0.24 \\ 
   1  & 2  & 2  & 2  & 3  & 3  & 0.33 & 0.33 \\ 
   1  & 1  & 3  & 2  & 3  & 3  & 0.37 & 0.37 \\ 
   2  & 1  & 2  & 2  & 3  & 3  & 0.48 & 0.48 \\ 
   1  & 1  & 2  & 2  & 3  & 3  & 1 & 0.65 \\ 
   \hline
\end{tabular}
\caption{\label{tabBVNlist} The batch prediction set using Bonferroni at $\alpha =0.1$, as well as the Bonferroni and Simes $p$-values for each $y$. }
\end{table}

Table \ref{tabBVN2LRT} adds the {\it median  } and the {\it oracle } adaptive Simes procedure, that uses respectively
\eqref{pi0estiQuantile}
 with $\l=\lceil m/2\rceil$ and $\hat{m}_0(y)=m_0(y)$ as estimator, to the comparison in Table \ref{tabBVN}. It also provides the estimated non-coverage for each method. Using oracle adaptive Simes is by far the best, but this is not a practical method since $m_0(y)$ is  unknown.

\setlength{\tabcolsep}{3pt}
{\small
 \begin{table}[h!]
 \centering
\begin{tabular}{r|rrrrrrr|rrrrrrr|}
   \hline
   & \multicolumn{7}{|c|}{Expected size of batch prediction set} & \multicolumn{7}{|c|}{Probability of non-coverage}\\
     &  &  & Storey & Median & Oracle &  & &&&   Storey & Median & Oracle &  & \\
   SNR & Bonf & Simes & Simes &  Simes &  Simes & Fisher & LRT & Bonf & Simes & Simes & Simes &  Simes &  Fisher  &LRT\\ 
   \hline
 1.00 & 410.52 & 384.66 & 327.55 & 346.09 & {\it 160.22} & {\bf 274.36} & 277.58 & 0.10 & 0.10 & 0.10 & 0.10 & 0.10 & 0.10 & 0.11 \\ 
 1.50 & 217.69 & 187.36 & 142.98 & 154.47 &  {\it 70.56} & {\bf 107.85} & 113.88 & 0.09 & 0.09 & 0.10 & 0.09 & 0.09 & 0.10 & 0.10 \\ 
 2.00 & 81.63 & 65.52 & 49.12 & 50.35 &  {\it 26.32} & {\bf 37.40} & 37.76 & 0.08 & 0.08 & 0.09 & 0.09 & 0.08 & 0.09 & 0.10 \\ 
 2.50 & 23.51 & 17.98 & 15.08 & 14.53 & {\it 9.05} & 14.60 & {\bf 11.91} & 0.10 & 0.11 & 0.11 & 0.11 & 0.11 & 0.10 & 0.11 \\ 
 3.00 & 6.42 & 5.35 & 5.18 &   4.90 & {\it 3.57} & 7.78 & {\bf 4.35} & 0.08 & 0.09 & 0.09 & 0.08 & 0.09 & 0.09 & 0.09 \\ 
 3.50 & 2.46 & 2.24 & 2.27 & 2.21 &  {\it 1.79} & 5.20 & {\bf 2.02} & 0.08 & 0.08 & 0.09 & 0.09 & 0.08 & 0.08 & 0.09 \\ 
 4.00 & 1.39 & 1.34 & 1.37 & 1.38 &  {\it 1.22} & 4.38 & {\bf 1.28} & 0.08 & 0.08 & 0.08 & 0.08 & 0.08 & 0.08 & 0.09 \\ 
 4.50 & 1.07 & 1.06 & 1.08 & 1.09 &  {\it 1.03} & 4.03 & {\bf 1.03} & 0.09 & 0.09 & 0.09 & 0.09 & 0.09 & 0.10 & 0.09 \\ 
    \hline
 \end{tabular}
 \caption{The average batch prediction set size at each SNR (columns 2--8) and probability of non-coverage (columns 9--15) for the batch conformal prediction inference at level $\alpha = 0.1$, for the following $p$-value combining functions: Bonferroni, Simes, adaptive Simes using Storey's estimator and the median estimator (see detailed data generation in text), oracle Simes, Fisher, and the estimated LRT. In bold, the (practical) combining method that produces the narrowest prediction region (oracle adaptive Simes is in italic).  Based on 2000 simulations. For a single data generation, the average running time on a standard PC was less than 0.05 seconds for all methods but the estimated LRT, which has an average running time of 5.7 seconds.   \label{tabBVN2LRT}}
 \end{table}
}

Table \ref{tabBVNbounds} provides the average sum of lower and upper bounds for the three classes by the different methods. The goal in the comparisons in this table are two fold. First, to assess how conservative the shortcut suggested in \S~\ref{sec:shortcut_supplementary} for computational efficiency is. Using Simes  (columns 3 and 4), it appears that the shortcut produces almost the same exact bounds for low SNR, and the inflation (i.e., smaller lower bounds and higher upper bounds with the shortcut) for high SNR is tiny. Using adaptive Simes (columns 6 and 7), it appears that there is a light inflation for all SNRs, and it is larger than using Simes.  
The second goal is to compare the efficiency of each combining method. For this purpose, we also provide Table \ref{tabBVN2classeslowerbounds2} that includes the estimated LRT but is based on a smaller number of simulations (since the bounds take 100 times longer to compute with the estimated LRT). As expected, the bounds using Simes are tighter than using Bonferroni, but the advantage is small. A more pronounced difference is with respect to oracle Simes, but it is not a practical method since $m_0(y)$ is unknown in practice. The bounds using Fisher 
is worse than other methods for SNR $\geq 2.5$, and better for the upper bound if SNR $\leq 2$. The bounds using the estimated LRT tend to be the tightest  among the practical methods considered.

{\small
 \begin{table}[h!]
 \centering
\begin{tabular}{|c|ccccccc|}
\hline
&       & & Shortcut& Oracle & Storey& Shortcut Storey& \\
SNR &        Bonf&  Simes&  Simes& Simes&  Simes&  Simes& Fisher \\ \hline
1 &  0.1735 & {\bf 0.1799}  &   {\bf 0.1799 }&    {\it  0.3056} &   0.1601 &        0.1598 & 0.0959\\
 1.5 &0.5731 &0.5923&     0.5923&     {\it  0.8769} &   {\bf 0.5998} &        0.5973 &0.4691\\
 2&   1.3846& 1.4423 &    1.4423&     {\it  1.8984 } &  {\bf 1.4692}   &      1.4665& 1.3304\\
2.5& 2.6567& {\bf 2.7494} &    {\bf 2.7494}  &   {\it  3.2361} &   2.7424  &       2.7375& 2.4744\\
3 &  3.9335& {\bf 4.0222} &    {\bf 4.0222}   &   {\it 4.4062} &    3.9831 &        3.9718 &3.4714\\
3.5 &5.0332 & {\bf 5.0740} &    {\bf 5.0740} &    {\it  5.2971} &   5.0384  &       5.0297& 4.2149\\
4 &  5.6546 & {\bf 5.6741 }  &   5.6725  &   {\it  5.7897} &    5.6505  &       5.6431& 4.6495\\
 4.5 &5.9349 & {\bf 5.9403} &   5.9320 &    {\it  5.9729} &    5.9112 &        5.9031& 4.9124\\
  \hline
  1  & 16.4065 &16.2350   & 16.2350 &    {\it 14.6186}  & 15.9986    &    16.0034 & {\bf 15.5516}\\
 1.5 &14.6781& 14.3764  &  14.3764&      {\it12.8131 }&  14.1638 &       14.2222 & {\bf 13.6339}\\
 2   &12.3595 &11.9946 &   11.9946   &   {\it10.8056 }&  11.9328  &      12.0616 & {\bf 11.6715}\\
 2.5 &10.0392 & {\bf 9.7815}    & {\bf 9.7815}  &    {\it 9.0074}&    9.8433   &      9.9388& 10.0506\\
 3   & 8.2403  & {\bf 8.0921}   &  {\bf 8.0921}   &   {\it 7.6426 } &  8.1661   &      8.2092 & 8.8527\\
 3.5  &6.9937  & {\bf 6.9344}  &   6.9348 &     {\it 6.7016 }&   6.9844  &       6.9952 & 8.0839\\
 4   & 6.3479  & {\bf 6.3242}  &   6.3280  &    {\it 6.2107 }&   6.3514  &       6.3611  &7.6670\\
 4.5  &6.0651  & {\bf 6.0595 } &   6.0693  &    {\it 6.0271 }&   6.0884   &      6.0979 & 7.4120\\ \hline
 \end{tabular}
 \caption{\label{tabBVNbounds}  Sum of average lower bounds $\sum_{k=1}^3\l^{(k)}_\alpha$   (rows 1--8) and upper bounds $\sum_{k=1}^3 u^{(k)}_\alpha$ (rows 9--16) of $\sum_{k=1}^3m_k(Y)=3$ at each SNR for different batch conformal prediction inferences at level $\alpha = 0.1$. Estimation with an average over 10000 replications. 
 The most informative practical bound has highest lower bounds / lowest upper bounds among the practical methods (in bold). Oracle Simes is in italic.  
 }
 \end{table}
}

{\small
 \begin{table}[h!]
 \centering
\begin{tabular}{|c|ccccc|ccccc|}
\hline
& \multicolumn{5}{|c|}{Sum of the average lower bounds} & \multicolumn{5}{|c|}{Sums of the average upper bounds} \\
& & & Storey & & Estimated & & & Storey & & Estimated\\ 
SNR  & Bonf & Simes & Simes & Fisher & LRT & Bonf & Simes &  Simes & Fisher & LRT \\ 
   \hline
 1 & 0.21 & {\bf 0.22} &  0.20 & 0.13 & 0.14 & 16.41 & 16.21 &   15.96  & 15.54 & {\bf 15.47} \\ 
   1.5 & 0.58 & 0.59  & {\bf 0.60}  & 0.45 & 0.55 & 14.66 & 14.32  & 14.14  & 13.60 & {\bf 13.45} \\ 
   2 & 1.36 & 1.41  & 1.46  & 1.32 & {\bf 1.55} & 12.34 & 12.01  & 11.94  & 11.70 & {\bf 11.37} \\ 
   2.5 & 2.58 & 2.67  & 2.70  & 2.47 & {\bf 2.81} & 10.07 & 9.83  & 9.83  & 10.04 & {\bf 9.57} \\ 
   3 & 4.03 & 4.10  & 4.05  & 3.49 & {\bf 4.25} & 8.16 & 8.03  & 8.11  & 8.86 & {\bf 7.92} \\ 
   3.5 & 5.00 & 5.05  & 5.05  & 4.22 & {\bf 5.20} & 7.04 & 6.96  & 6.99  & 8.09 & {\bf 6.85} \\ 
   4 & 5.67 & 5.70  & 5.65  & 4.64 & {\bf 5.76} & 6.33 & 6.30  & 6.35  & 7.67 & {\bf 6.25} \\ 
   4.5 & 5.93 & 5.94  & 5.93  & 4.95 & {\bf 5.98} & 6.07 & 6.06  & 6.07  & 7.38 & {\bf 6.02} \\ 
    \hline
 \end{tabular}
 \caption{\label{tabBVNbounds2}  Sum of average lower bounds $\sum_{k=1}^3\l^{(k)}_\alpha$ (columns 2--6) and upper bounds $\sum_{k=1}^3 u^{(k)}_\alpha$ (columns 7--11) for $\sum_{k=1}^3m_k(Y)=m$  at each SNR for the different batch conformal prediction inferences with average batch size presented in Table \ref{tabBVN},  at level $\alpha = 0.1$. Estimation with an average over 2000 replications. 
 The most informative bound has highest lower bounds / lowest upper bounds (in bold).  
 }
 \end{table}
}

\subsection{USPS and CIFAR data sets}\label{sec:appenddata}

To obtain a visualization different from the one of \S~\ref{sec:data}, Figure~\ref{fig:Power2} displays the averaged size of batch prediction sets as well as an estimation of the coverage in function of $\alpha$ in the same setting as Figure~\ref{fig:Power}. The conclusions are analogue. The fluctuations of the coverage around $\alpha$ in Figure~\ref{fig:Power2} for the Bonferroni, Simes and Storey procedures  is due to the uncertainty of the empirical estimation of the coverage.

\begin{figure}[h!]
\begin{tabular}{cc}
\hspace{-1cm} USPS data set & \hspace{-1.2cm}CIFAR data set\vspace{-7mm}\\
\hspace{-1cm} \includegraphics[scale=0.25]{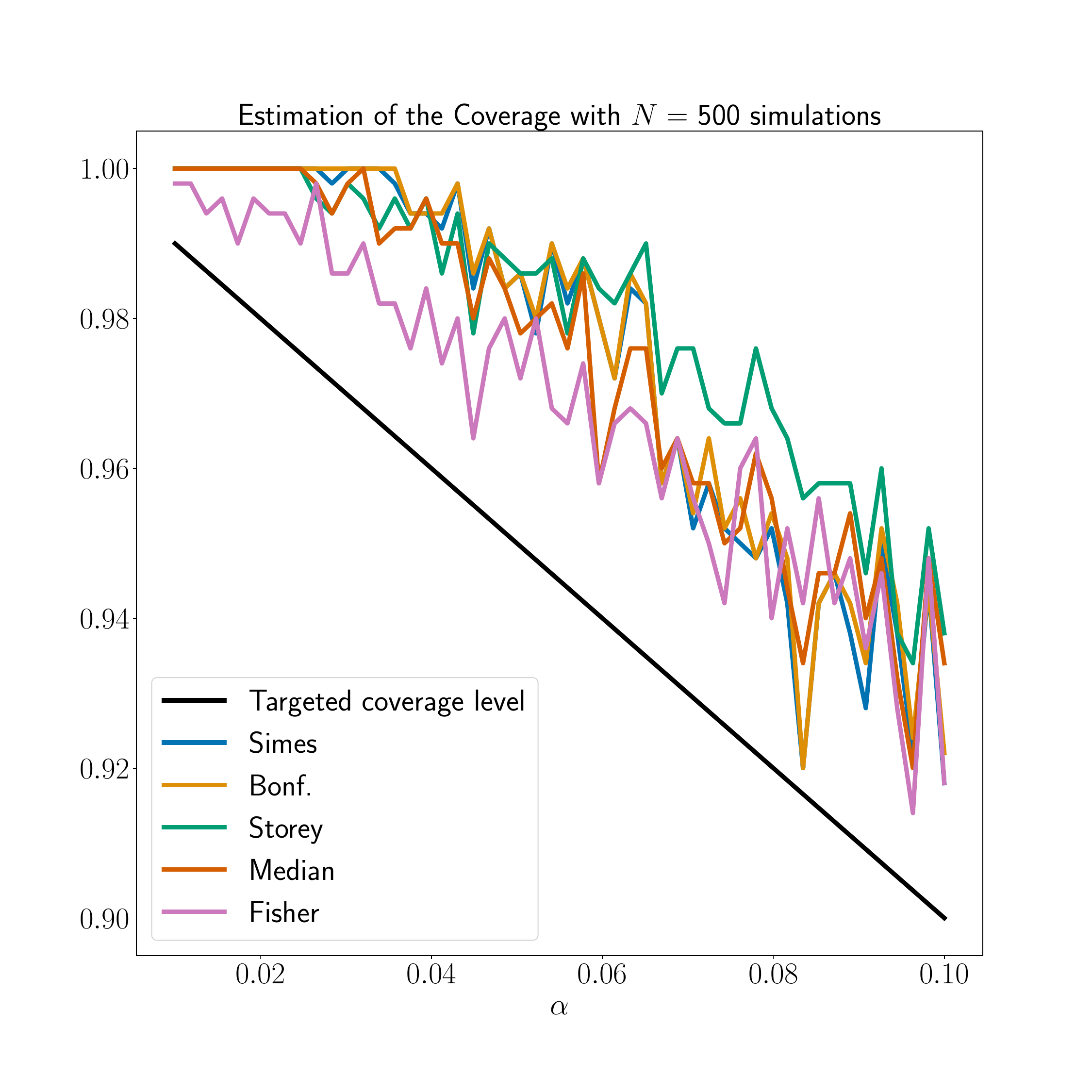}    &  \hspace{-1.2cm}  \includegraphics[scale=0.25]{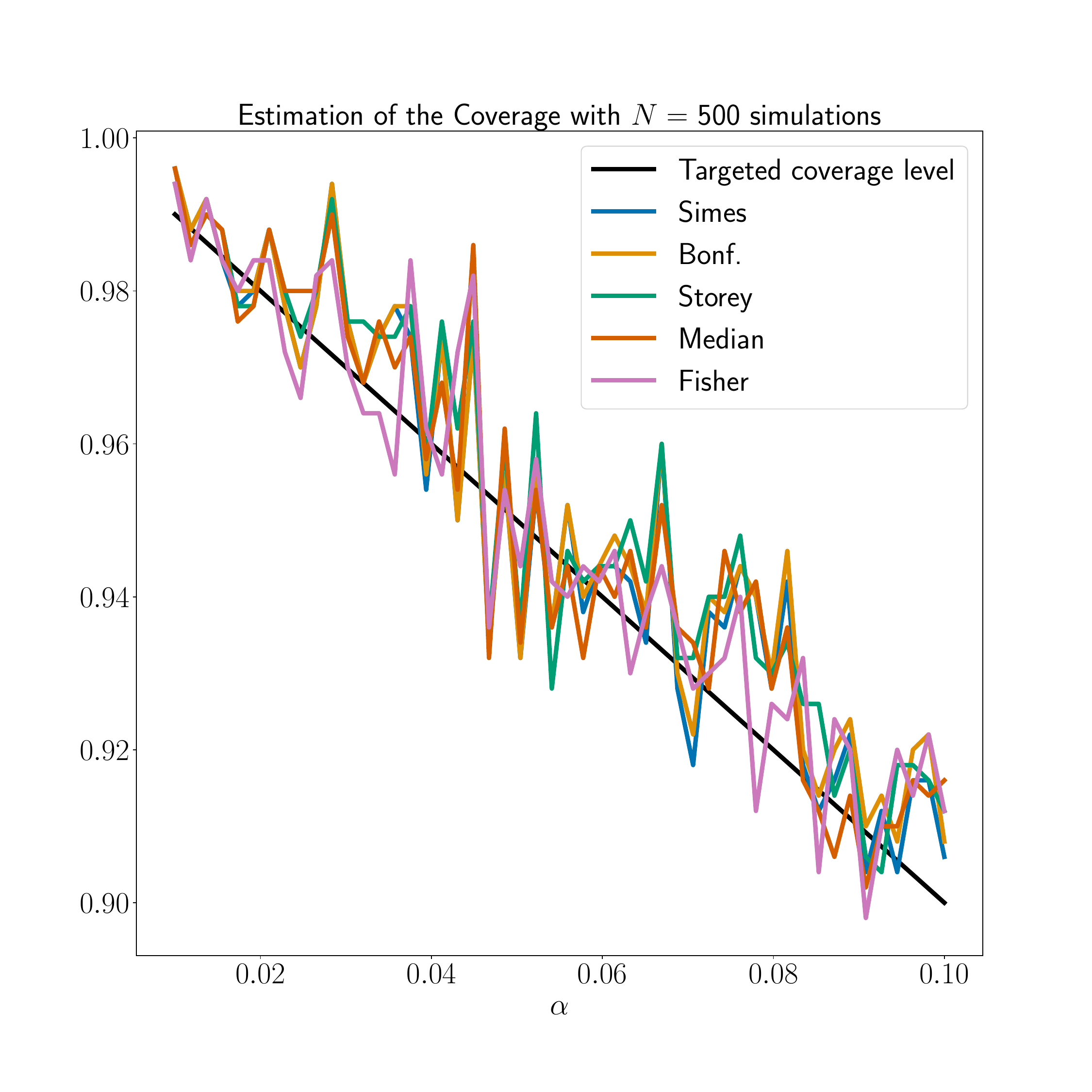} \vspace{-7mm}\\
\hspace{-1cm} \includegraphics[scale=0.25]{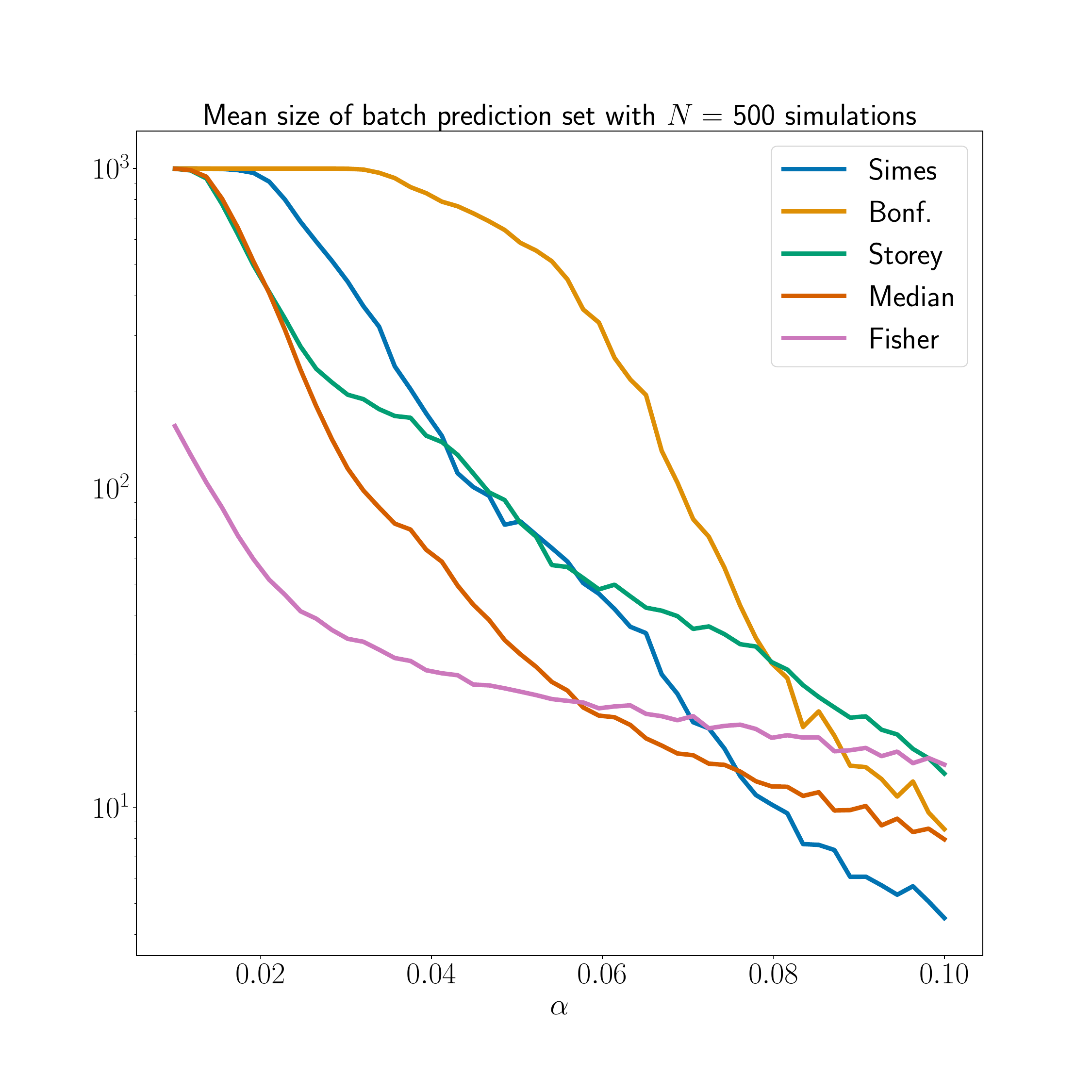}    &  \hspace{-1.2cm}  \includegraphics[scale=0.25]{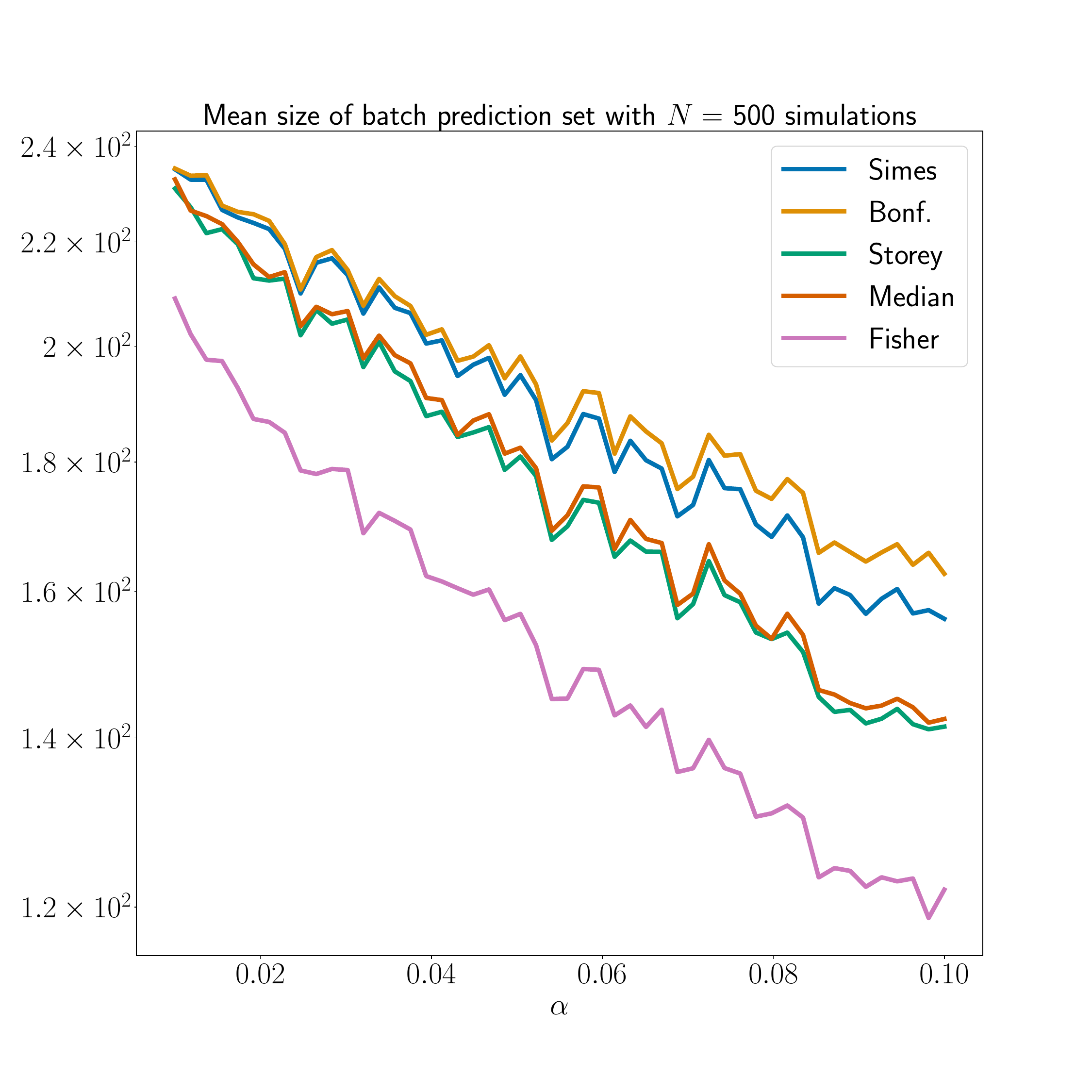} \vspace{-7mm}\\
\end{tabular}
\vspace{-0.3cm}
\caption{\label{fig:Power2} 
Top: averaged coverage of the batch prediction sets in function of $\alpha$ for different procedures.
Bottom: averaged size of the batch prediction sets in function of $\alpha$ for different procedures. Same setting as for Figure~\ref{fig:Power}. The standard error for the USPS dataset is below $0.013$, and is below  $0.014$ for the CIFAR dataset.
}
\end{figure}

\subsection{Survey animal populations for CIFAR data set}\label{sm-subsec-CIFAR}

In this section, we illustrate the task of predicting the counts for each category (task (ii) in the main text) for the batch displayed in Figure~\ref{animals}. The lower and upper bounds for the number of each animal in this batch are given in Table~\ref{table:BoundsCIFAR}. As in the previous section, while the improvement of the new methods are significant for the size of the batch prediction sets, it is more modest for the bounds.

\begin{figure}[h!]
\hspace{1.6cm}\includegraphics[scale=0.6]{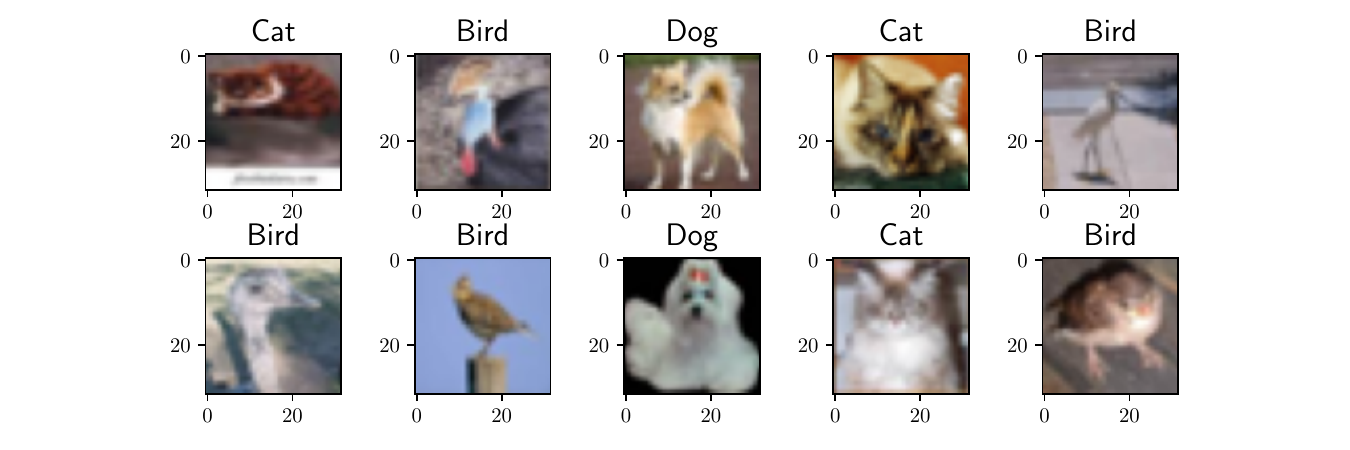}
\caption{\label{animals} One batch of the CIFAR dataset~\citep{Krizhevsky2009LearningML}. }
\end{figure}

 \begin{table}[h!]
 \centering
\begin{tabular}{|l|llllll|}
\hline
 & Simes & Bonf. & Storey & Median & Fisher & LRT \\
\hline
Bird & 0 ; 9 & 0 ; 9 & 0 ; 9 & 0 ; 8 & 0 ; 8 & 0 ; 7 \\
Cat & 0 ; 10 & 0 ; 9 & 0 ; 9 & 0 ; 8 & 0 ; 8 & 0 ; 9 \\
Dog & 0 ; 10 & 0 ; 10 & 0 ; 10 & 0 ; 10 & 0 ; 10 & 0 ; 10 \\
Size & 27216 & 39366 & 24459 & 20680 & 12653 & 11313 \\
\hline
\end{tabular}

 \caption{Bounds for the particular batch of Figure~\ref{animals} from the CIFAR data set at level $\alpha=0.1$. The number of birds, cats, and dogs in the batch is 5, 3, and 2, respectively. \label{table:BoundsCIFAR} }
 \end{table}

\subsection{Full versus class calibrated $p$-values under label shift}\label{sec:ClassvFull}

\et{In this section, we illustrate the importance of the \emph{class conditional model} and the conditional guarantee~\eqref{aimcond} with the CIFAR dataset. The calibration sample is of size $n=2000$, with $10\%$ of birds, $30\%$ of cats and $60\%$ of dogs. The test sample is $m=5$ with $2$ birds ($40\%$), $3$ cats ($60\%$), hence without dogs. They are both drawn without replacement in the CIFAR data set. While the distribution of $X$ given $Y$ is the same, there is a significant label shift between the calibration and test samples. Hence, using full-calibrated $p$-values is not appropriate and we should rely on class-conditional $p$-values to retain the guarantees \eqref{aimcond} and thus the $(1-\alpha)$-coverage under this specific data-generation process. }

\et{The coverage of the different approaches are approximated with $1000$ replications and reported in Figure~\ref{fig:ClassvFull}. This corroborates the theoretical findings: the full calibrated approaches can miss the nominal coverage by a lot in this case, whereas the class calibrated approaches ensure the correct coverage. }

\begin{figure}[h!]
\begin{tabular}{cc}
\hspace{-1cm} Class-calibrated $p$-values & \hspace{-1.2cm} Full-calibrated $p$-values\vspace{-7mm}\\
\hspace{-1cm} \includegraphics[scale=0.25]{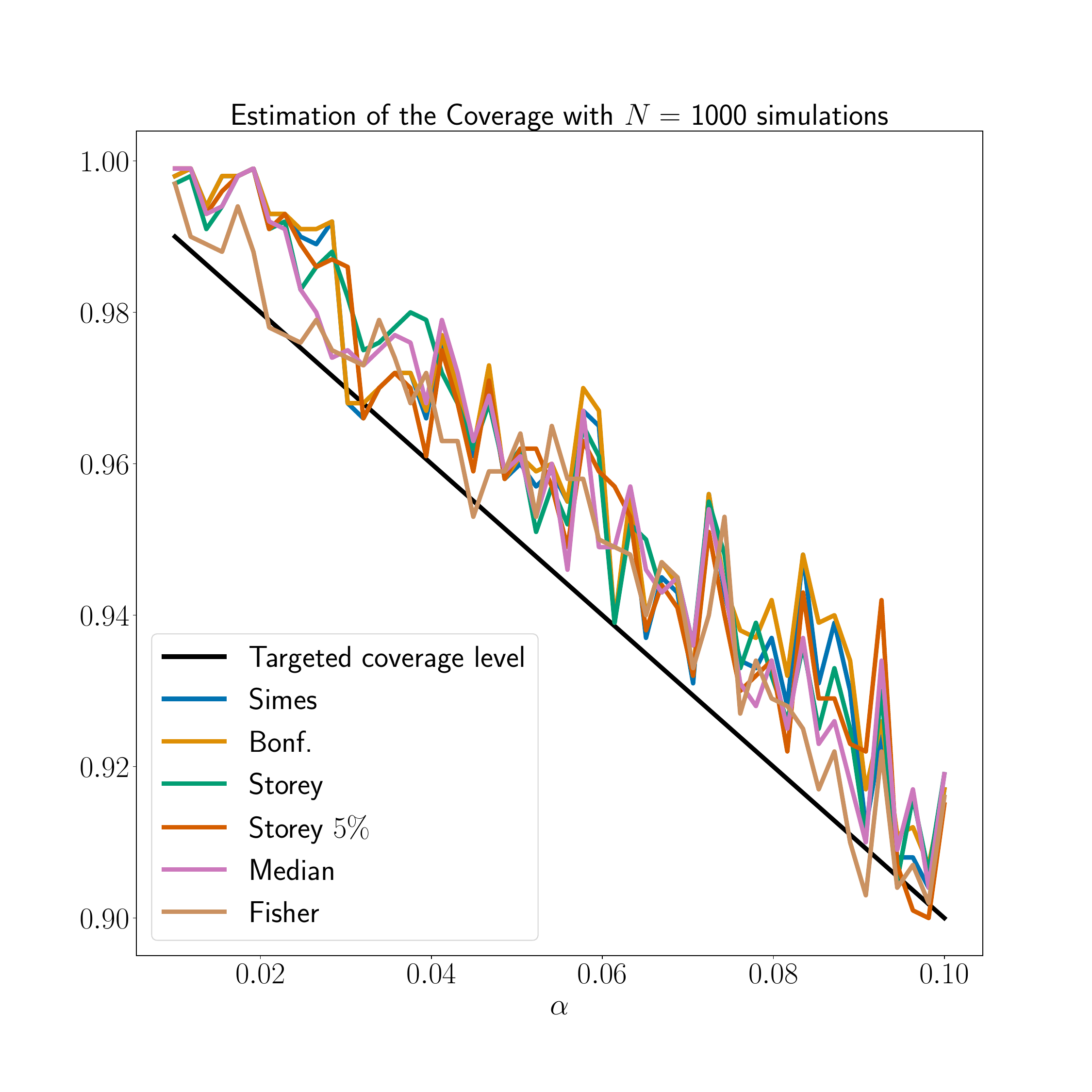}    &  \hspace{-1.2cm}  \includegraphics[scale=0.25]{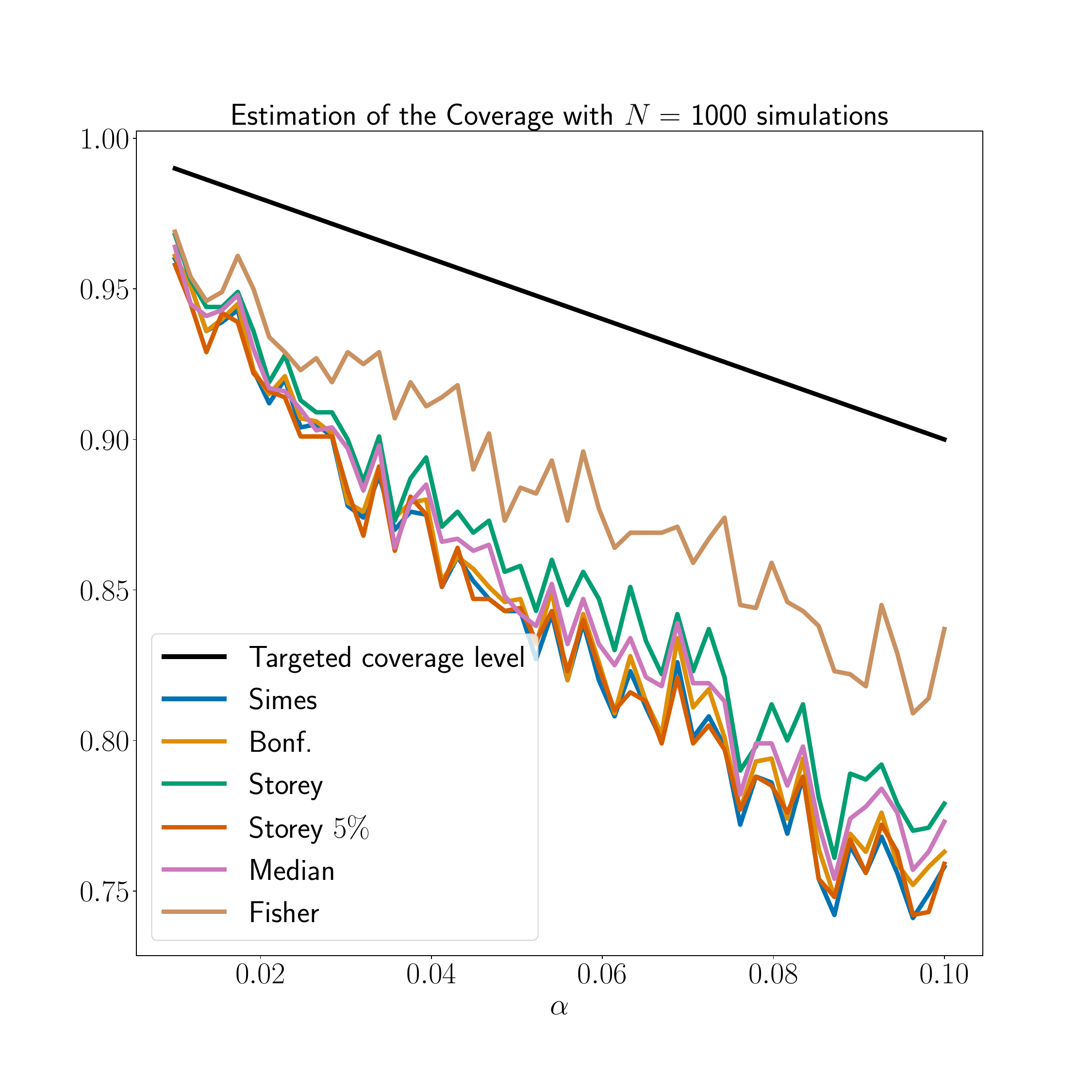} \vspace{-7mm}\\
\end{tabular}
\vspace{-0.3cm}
\caption{\label{fig:ClassvFull} 
 Estimated coverage for class calibrated $p$-values (left) and full calibrated $p$-values (right) in the label shift setting described in \S~\ref{sec:ClassvFull}. The estimated standard errors are below $0.009$ with the class-calibrated $p$-values, and are all below $0.014$ with the full-calibrated $p$-values.
}
\end{figure}

\subsection{LRT computation time}\label{sec:LRTbadcomplex}

\et{Table~\ref{tabTimes} provides the time to compute one batch prediction set for different methods with the CIFAR and USPS datasets, averaged over $500$ simulations and in the simulation setting of \S~\ref{sec:data}. As one can see, the LRT is by far the most computationally demanding method.}

\setlength{\tabcolsep}{3pt}
{\small
 \begin{table}[h!]
 \centering
\begin{tabular}{|c|ccc|ccc|}
\hline
&\multicolumn{6}{|c|}{Targeted coverage}\\
& \multicolumn{3}{|c|}{USPS dataset } & \multicolumn{3}{|c|}{CIFAR dataset} \\
 & 0.99 & 0.95 & 0.90 & 0.99 & 0.95 & 0.90 \\
\hline
Simes & 0.027 & 0.027 & 0.026 & 0.008 & 0.007 & 0.006 \\
Bonf. & 0.024 & 0.023 & 0.023 & 0.003 & 0.003 & 0.003 \\
Storey & 0.030 & 0.030 & 0.030 & 0.005 & 0.005 & 0.005 \\
Median & 0.028 & 0.027 & 0.027 & 0.004 & 0.004 & 0.004 \\
Fisher & 0.072 & 0.072 & 0.072 & 0.015 & 0.015 & 0.015 \\
LRT & 5.690 & 5.668 & 5.656 & 7.475 & 7.461 & 7.507 \\
\hline
\end{tabular}
\caption{Mean time (in second) over $500$ replications of different procedures (in rows) and for different targeted $1-\alpha$ (in columns). The setting is the same as the one in \S~\ref{sec:data}.}\label{tabTimes}
\end{table}
}

\end{document}